\documentclass[12pt]{article}

\usepackage[T1]{fontenc}
\usepackage[english]{babel}
\usepackage{csquotes}
\usepackage[
style=apa, 
backend=biber, 
natbib=true, 
maxcitenames=2, 
mincitenames=1,
uniquename=false,
uniquelist=false,
sortcites=true,
sorting=ynt]
{biblatex}
\bibliography{orderstat}

\usepackage[hidelinks]{hyperref}
\hypersetup{colorlinks=true, allcolors=blue}

\usepackage{titlesec} 
\usepackage{bookmark}
\usepackage{graphicx}
\usepackage{amsmath}
\usepackage{amsfonts}
\usepackage{amssymb}
\usepackage[mathscr]{eucal}
\usepackage{bm}
\usepackage[normalem]{ulem}
\usepackage{url}
\usepackage{tablefootnote}
\usepackage{placeins}
\usepackage{bbm}
\usepackage{mathtools}
\usepackage{enumitem}
\usepackage{authblk}
\usepackage{pdflscape}

\usepackage{algorithm}
\usepackage{algpseudocode}

\usepackage{xcolor}
\newcommand{\dsphere}{\mathbb{S}^{d-1}}
\newcommand{\RR}{\mathbb{R}}
\newcommand{\E}{\mathbb{E}}
\newcommand{\Var}{\operatorname{Var}}
\newcommand{\Cov}{\operatorname{Cov}}
\newcommand{\sgn}{\operatorname{sgn}}
\newcommand{\dd}{\,\mathrm{d}}

\usepackage{amsthm}
\newtheorem{theorem}{Theorem}[section]
\newtheorem{corollary}{Corollary}[theorem]
\newtheorem{proposition}[theorem]{Proposition}
\newtheorem{lemma}{Lemma}

\author[1]{E. Mackay}
\author[2]{J. Richards}
\author[3]{P. Jonathan}
\affil[1]{Department of Engineering, University of Exeter, Penryn, TR10 9FE, UK. e.mackay@exeter.ac.uk. ORCID: 0000-0001-7121-4231}
\affil[2]{School of Mathematics and Maxwell Institute for Mathematical Sciences,  University of Edinburgh, Edinburgh, EH9 3FD, UK. jordan.richards@ed.ac.uk. ORCID: 0000-0002-0697-2551}
\affil[3]{School of Mathematical Sciences, Lancaster University, Lancaster, LA1 4YF, UK. p.jonathan@lancaster.ac.uk. ORCID: 0000-0001-7651-9181}

\title{Diagnostic Tools for Extreme Value Regression Models}
\date{\today}

\begin{document}
\bookmarksetup{startatroot} 
\currentpdfbookmark{Main Paper}{main_paper_bookmark}

\maketitle    

\begin{abstract}
Visual and quantitative goodness-of-fit diagnostics are an important tool in the practitioner's toolbox. The need for convincing and reliable diagnostics is particularly clear when fitting extreme value regression models, which are used for extrapolation far beyond the observable range of the response variable, and often evaluated at unobserved covariate values. Despite this, few diagnostics have been developed for extreme value regression models, and those available often suffer in terms of interpretability or scalability on low-dimensional or non-Euclidean covariate domains, often encountered in modern applications. Moreover, existing methods tend to offer a global perspective on model fit; that is, they quantify goodness-of-fit across the entire dataset, without offering insight into regions of the covariate space where the model fit may be poor. We propose two novel visual diagnostics for extreme value regression models: the \textit{standardised tail plot} and the \textit{normalised residual plot}. By considering the asymptotic distribution of normalised exceedance probabilities, we show that uncertainty bounds for our plots are approximately independent of the sample size used in their construction. This allows us to propose visual diagnostics which can efficiently and consistently compare goodness-of-fit at both a global and regional level, despite varying sample sizes over regions of the covariate domain. Following a discussion of summary statistics for global and regional goodness-of-fit, we provide two applications of extreme value regression models that illustrate how our diagnostics can be used to perform model comparison (across thousands of candidate models) and provide actionable findings that support model design.
\end{abstract}



%%%%%%%%%%%%%%%%%%%%%%%%%%%%%%%%%%%%%%%%%%%%%%%%%%%%%%%
\section{Introduction} \label{sec:intro}
Extreme value regression involves modelling the extremes of a response variable conditional on covariates. Such models have been applied in a wide range of settings, including hydrology \citep{lee2020application, anzolin2024nonstationary}, environmental risk assessment \citep{hundecha2008nonstationary, krock2022nonstationary, le2022non, majumder2025semi}, climate \citep{chavez2005, renard2012bayesian, cheng2014non, vasiliades2015nonstationary, robin2020nonstationary}, coastal protection \citep{dixon1999effect, razmi2017non, ragno2019generalized, baldan2022importance}, offshore engineering \citep{Randell2016, hansen2020directional, Zanini2020, Barlow2023, tendijck2024practical}, surrogate models for structural responses \citep{gramstad2020sequential, zhao2024surrogate}, finance and insurance \citep{mcneil2000estimation, chavez2016extreme, hambuckers2018understanding}, public health \citep{nadarajah2023extreme}, life expectancy \citep{einmahl2019limits}, energy demand forecasting \citep{sigauke2017modelling}, and sports analytics \citep{pauli2001penalized}.  

Beyond the standard regression setting with a univariate response variable, extreme value regression models have recently been applied to the estimation of multivariate extremes \citep[see, e.g.,][]{murphy2024deep, murphy2024inference, simpson2024estimating, majumder2023semiparametric, wadsworth2024statistical, demonte2025generative, mackay2026spar, murphy2026exploring}. After a transformation to pseudo-polar coordinates, the joint extremes of a random vector can be modelled in terms of the extremes of a radial variable conditioned on pseudo-angle \citep{mackay2025spar}. In this context, the pseudo-angle can be treated as a covariate in a regression model for the radial variable. In the following, without loss of generality, we will consider the problem of the regression of a univariate response variable conditioned on covariates. 

A key step in statistical modelling is the assessment of the fit of a model to observed data. Assessing model fit is particularly pertinent in the context of extreme value modelling; if a model does not agree with the observed distribution of sub-asymptotic extreme events, then this gives little confidence when using the model to extrapolate outside the range of the observations, to make inference on truly extreme events. The fit of an extreme value model can be assessed using quantitative diagnostics and statistical tests, but various visual diagnostics have also been proposed. These can serve a range of purposes. In the present work, we focus on the following: i) assessing how well the model fits to data; ii) identifying regions of covariate space where the model fit can be improved; iii) facilitating comparison between candidate models. The first objective is to assess whether differences between the model and observations are within the range expected due to sampling effects. Here, we are interested in testing the null hypothesis that an estimated distribution is equal to the true data-generating distribution. To identify parts of the covariate domain where a model is performing poorly requires the use of targeted \textit{local} or \textit{regional} diagnostics, that describe model performance for a particular value of a covariate or region of the covariate domain. Moreover, diagnostics that describe the performance of a model \textit{globally}, that is, for all observed covariate-response pairs, can fail to identify models that perform poorly at a local or regional level \citep{zhao2021diagnostics}. Regarding the third objective, regression models can involve the specification of various hyperparameters used to control the flexibility and smoothness of the covariate dependencies, particularly in a semi-parametric or fully non-parametric setting. The availability of diagnostics for assessing the quality of model fit across the covariate domain aids in comparing various candidate models and selecting optimal hyperparameters. Selecting between candidate models based on comparison of a loss function (e.g., a log-likelihood) or an information criterion provides a relative comparison, but does not provide a quantification of whether the selected model is a good fit to the data. More generally, the adequacy of information criteria in selecting appropriate model complexity in an extreme value regression setting is the topic of some debate \citep[see e.g.][]{gelman2014understanding, zhang2023information}.

In this work, we introduce several new visual and quantitative diagnostic tools for extreme value regression models, designed to meet the three objectives given above. The goodness-of-fit diagnostics proposed here can be be applied to any type of parametric, semi-parametric, or non-parametric extreme value regression model, whether this is a threshold exceedance (or peaks-over-threshold) model, a block maxima regression model, or otherwise; the only criteria is that the approach provides estimates of conditional distribution functions. The key feature of the proposed diagnostics is that they are designed to assess the fit of a model in the upper tail of the conditional distribution. For threshold exceedance regression models, a (potentially non-stationary) threshold must also be estimated. In the present work, we do not consider diagnostics for threshold selection. Instead, the focus is on assessing the fit of a model after a threshold has been selected. For further discussion of threshold selection, see, for example, \citet{scarrott2012review, murphy2025}.

A key feature of our proposed diagnostics is that they can be applied to situations where the there is large variation in the response over the covariate domain. In these cases, partitioning the domain into regions where the response distribution is approximately stationary would result in a large number of regions, making it infeasible to manually inspect diagnostic plots for every region. The visual diagnostics introduced in this work enable an assessment of the model performance on a standardised scale, so that goodness-of-fit information from many regions can be combined and displayed in a single plot. This provides a fast visual assessment of the model performance over the entire covariate domain. 

The paper is organised as follows. In Section~\ref{sec:existing}, we briefly review common diagnostics for stationary univariate extreme value models, and discuss the challenges that arise when extending these approaches to the regression setting. We also discuss related diagnostic plots for models for conditional distributions. Section~\ref{sec:standard_regional} introduces two novel diagnostic plots; the \textit{standardised tail plot} and the \textit{normalised residual plot}. The standardised tail plot provides a consistent visualisation of the fit of an extreme value regression model across all regions of the covariate space. The plot makes use of an asymptotic property of the sampling distribution of exceedance probabilities of extreme order statistics, which we demonstrate provides a good approximation even with relatively small sample sizes. In contrast, the normalised residual plot does not rely on asymptotic results, and shows the significance of the deviation between the model and observations as a function of rank, with the deviation on a standard normal scale. This type of plot can provide an indication of the model fit at all probability levels, rather than just at the extremes.Section~\ref{sec:summary} discusses how the information in the standardised tail and normalised residual plots can be summarised, to give a quantification of global model performance. As well as considering the use of standard goodness-of-fit statistics, we introduce two new summary statistics that are defined in terms of the residuals used to construct our diagnostic plots. We derive the asymptotic distributions of these statistics and consider their sensitivity for detecting lack of fit in the upper tails of a distribution, compared to standard goodness-of-fit tests. The proposed diagnostics are illustrated in Section~\ref{sec:example} using two examples involving both simulated data and surrogate model output, with covariate dimension $d=4$ and $d=3$, respectively. Finally, conclusions are presented in Section~\ref{sec:conclusion}. For readers wishing to apply the diagnostics, R and MATLAB code is available from \url{https://github.com/edmackay/Diagnostics-for-extreme-value-regression}.

%%%%%%%%%%%%%%%%%%%%%%%%%%%%%%%%%%%%%%%%%%%%
\section{Existing model diagnostics} \label{sec:existing}
\subsection{Diagnostic plots for univariate extreme value models} \label{sec:stationary_plots}

Suppose we have a sample of observations $\{y_1,\dots,y_n\}$ of a random variable $Y\in\mathcal{Y}\subseteq\RR$ with continuous distribution function $F_Y$. These samples are used to estimate a model distribution function, denoted by $\hat{F}_Y$. For univariate extreme value models, commonly-used diagnostic plots include probability-probability (PP) plots and quantile-quantile (QQ) plots \citep{Coles2001}. These plots make use of the fact that the probability integral transform (PIT) of $Y$, denoted by $P = F_Y(Y)$, is uniformly distributed on $[0,1]$. Therefore, if $\hat{F}_Y = F_Y$, the values $\{\hat{F}_Y (y_1),\dots,\hat{F}_Y (y_n)\}$ are a sample from a uniform distribution. Let $y_{(1)}\le \cdots \le y_{(n)}$ denote the ordered sample, where $k$ is the (ascending) rank associated with $y_{(k)}$. Each ordered observation $y_{(k)}$ is assigned an empirical non-exceedance probability $p_k\in[0,1]$. The definition of $p_k$ varies between practitioners, with $p_k = (k-0.5)/n$ or $p_k = k/(n+1)$ being common choices. The value assigned to $p_k$ is known as the plotting position, and the most appropriate choice depends on the type of diagnostic plot \citep[see, e.g.,][]{cunnane1978unbiased, leon1984another, arnell1986unbiased}.

A PP plot consists of the pairs
\begin{equation*}
    \left\{ \left(p_k,\, \hat{F}_Y(y_{(k)})\right) : k=1,\dots,n \right\},
\end{equation*}
and a QQ plot consists of the pairs
\begin{equation*}
    \left\{ \left( \hat{F}_Y^{-1}(p_{k}),\, y_{(k)}\right) : k=1,\dots,n \right\},
\end{equation*}
where $\hat{F}_Y^{-1}$ is the model's estimated quantile function. For both PP and QQ plots, if $\hat{F}_Y$ is a reasonable model, then the points should lie close to the unit diagonal. Some works \citep[e.g.,][]{beirlant2004} advocate the use of QQ plots after transformation to some standard distribution $F_0$, so that differences between a model and observations are invariant to the data-generating distribution. A transformed QQ plot consists of the pairs
\begin{equation} \label{eq:transform_QQ}
    \left\{ \left(F_0^{-1}(p_k),\, F_0^{-1}\big(\hat{F}_Y(y_{(k)})\big)\right) : k=1,\dots,n \right\}.
\end{equation}
In the extreme value literature, $F_0$ is often taken as the standard exponential distribution\footnote{In the context of regression models, $-\log\big(1-\hat{F}_Y(y_{(k)})\big)$ is known as the Cox-Snell residual \citep{cox1968general}. See e.g. \citet{heffernan2001extreme} for application to extreme value regression models.} as its use accentuates differences in the upper tail of $F_Y$; in more general statistical applications, where deviations in the bulk of data are more of interest, $F_0$ is often taken as the standard normal distribution. 

As an alternative to PP and QQ plots, exceedance probability (EP) plots are popular in the engineering literature \citep[see, e.g.,][]{Randell2016, hansen2020directional}. They consist of two sets of points overlaid on the same plot, with one set of points corresponding to the observations and the other corresponding to the fitted model:
\begin{align*}
    \mathrm{Observations:} & \quad \left\{ \left(y_{(k)}, 1-p_k \right) : k=1,\dots,n \right\},\\
    \mathrm{Model:} & \quad \left\{ \left(y, 1 - \hat{F}_Y(y)\right) : y\in [a,b] \right\}.
\end{align*}
The exceedance probabilities are shown on a logarithmic scale to provide a better visualisation of the upper tail of the  distribution. The range of values $[a,b]$ used for the model may exceed the observed range $[y_{(1)},y_{(n)}]$ in order to illustrate how the model extrapolates from observations. In some contexts, the return period is shown instead of the exceedance probability, but since these are reciprocal quantities, return level plots provide the same diagnostic information. While we do not focus on EP plots hereafter, we present their details here to illustrate the connection between PP, QQ, and EP plots.

%%%%%%%%%%%%%%%%%%%%%%%%%%%%%%%%%%%%%%%%%%%%
\subsection{Challenges for regression models} \label{sec:challenges}
For regression models, the underlying probability model assumes that we have a response variable $Y\in\mathcal{Y}\subseteq\RR$ whose distribution is conditional on a covariate $\bm{X} \in \mathcal{X}\subseteq \RR^d$, $d\in\mathbb{N}_{>0}$, with conditional distribution function $F_{Y|\bm{X}}$.  Data used for estimating $F_{Y|\bm{X}}$ typically take the form of samples of pairs $\{(\bm{x}_1,y_1),\dots,(\bm{x}_n,y_n)\}$. Defining an empirical estimate of the conditional distribution function is more challenging than for the unconditional $F_Y$, as is constructing useful visual diagnostics. To see this, we note that the covariate vector $\bm{X}$ may be random or deterministic, and may take either discrete or continuous values (or a combination thereof). If $\bm{X}$ has a continuous density, then the expected number of observations at any given value $\bm{X}=\bm{x}$ is zero, and estimating $F_{Y|\bm{X}}(y \, | \, \bm{x})$ by empirical ranking at discrete values of $\bm{x}$ is not possible. To circumvent this issue, one option is to consider an aggregate distribution $\Pr(Y\le y \,|\, \bm{X} \in \mathcal{R}(\bm{x}))$, where $\mathcal{R}(\bm{x})$ is some region around $\bm{x}$, then compare the model and observations in this region. The downside of this approach is that it loses information about any variation in the distribution of $Y$ over the region. See \citet{Randell2016, hansen2020directional} for examples of the use of this type of diagnostic.

If there is large variation in the conditional distribution $Y|\bm{X}$ over the covariate domain, then assessing the model performance will require partitioning the covariate domain into many regions. As the number of regions grows, it becomes infeasible to manually inspect separate diagnostic plots for each region. Moreover, each region may contain a different number of observations, and the tail shape and scale of the conditional distribution $Y|\bm{X}$ may vary across regions, making it cumbersome to visually assess all regional diagnostics in a consistent manner.

Given a regional partitioning of the covariate domain $\mathcal{X}$, we aim to design visual tools for diagnosing regional goodness-of-fit which satisfy the following properties:
\begin{enumerate}[label=(\alph*), noitemsep]
    \item Summarise information about the model fit across multiple regions of the covariate domain, in a consistent manner;
    \item Preserve information about the non-stationary distribution;
    \item Are invariant to differences in the shape and scale of the tails of  $Y|\bm{X}$ across regions;
    \item Are invariant to the sample size in each region;
    \item Are applicable regardless of the dimension $d$ of the covariate set; and
    \item Provide a visual assessment of the significance of deviations between the estimated model and observations.
\end{enumerate}

%%%%%%%%%%%%%%%%%%%%%%%%%%%%%%%%%%%%%%%%%%%%
\subsection{Diagnostics for conditional distribution models} \label{sec:conditional_dist}
Relatively few diagnostics have been developed specifically for extreme value regression models. However, diagnostics have been proposed for the closely-related topics of quantile regression (QR) \citep{Koenker2005} and more general modelling of conditional distributions. While QR models are often estimated non-parametrically, without an explicit model for $F_{Y|\bm{X}}$, QR can be used to infer $F_{Y|\bm{X}}$ by interpolating between estimated quantile functions at different probability levels. So, in both the case of QR and extreme value regression, we require an assessment of a model for $F_{Y|\bm{X}}$; we denote this estimate by $\hat{F}_{Y|\bm{X}}$. 

A common visual diagnostic for QR models is a plot of conditional quantiles of the response variable against the covariates, typically with $d=1$, as shown in \autoref{fig:nonstat_problem}. These plots are used to assess whether, broadly, the quantiles ``track'' the data cloud, with approximately the correct proportion of the data falling between each quantile level. These can be useful when the covariate is one-dimensional ($d=1$), but require subjective assessment. Moreover, it is infeasible to produce such plots in the presence of higher-dimensional covariates. 

\begin{figure}[t]
	\centering
	\includegraphics[width=\textwidth]{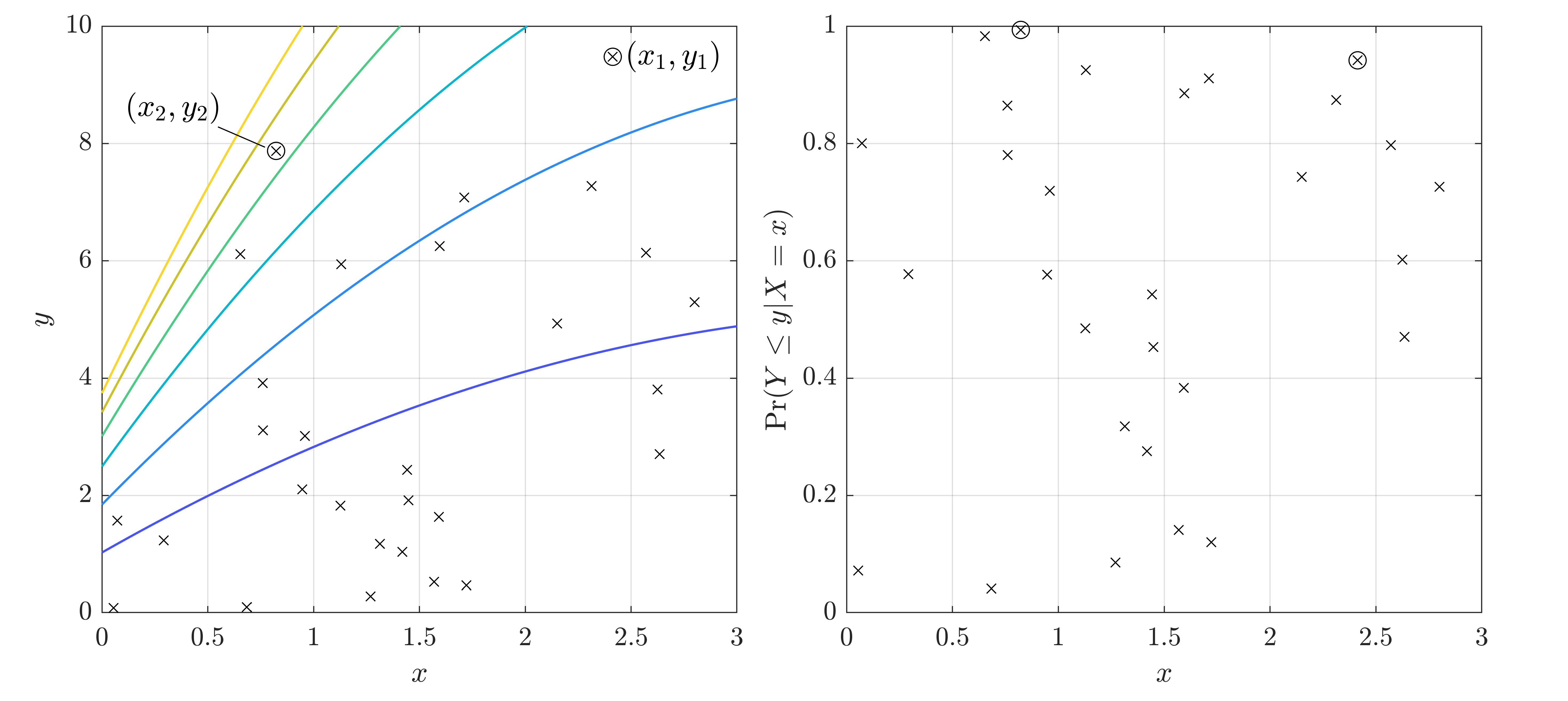}\\
	\caption{Left: Illustration of the problem of defining an ordering in non-stationary settings. Coloured lines show conditional quantiles of $Y|(X=x)$ for the data-generating distribution at exceedance probabilities $10^{-3},10^{-2.5},\dots,10^{-0.5}$. Black crosses show a random sample from the joint distribution. Right: Conditional non-exceedance probabilities of the sample shown on the left. These values are uniformly distributed and independent of the covariate. In both plots, the two circled points are the two largest values of the response variable with coordinates $(x_1, y_1)$ and $(x_2, y_2)$. The response on the right is larger than the response on the left ($y_1>y_2$), but is less `extreme' locally, in that it has a lower conditional non-exceedance probability, i.e., $\Pr(Y\le y_1 | X=x_1) < \Pr(Y\le y_2 | X=x_2)$.}
	\label{fig:nonstat_problem}
\end{figure}

Diagnostics for conditional distribution models have been proposed by, for example, \citet{fasiolo2020scalable}. These include visualisations of differences between model and empirical estimates, such as heat maps of differences in conditional densities plotted against covariate and quantile level. While these plots can identify regions of poor model fit and facilitate comparison between models, they do not provide a visual assessment of the significance of the deviations, do not scale well with dimension $d$, and are not invariant to the sample sizes or differences in the tails of conditional distribution functions. 

To design useful visual diagnostics for extreme value regression models that satisfy the desired properties outlined in Section~\ref{sec:challenges}, we consider extensions of univariate visual diagnostics to a regression setting. This follows by noting that, for any covariate value $\bm{x}\in \mathcal{X}$, the conditional \textit{local} PIT, $P_{\bm{x}}\coloneq F_{Y|\bm{X}}(Y|\bm{x}),$ of $Y|(\bm{X}=\bm{x})$ is uniformly distributed on $[0,1]$ \citep{oreilly1973conditional}. In the stationary case, ordering the observations also induces an ordering of the corresponding estimates of non-exceedance probabilities: for an ordered sample $y_{(1)} \le \cdots \le y_{(n)},$ we have $F_Y (y_{(1)}) \le \cdots \le F_Y (y_{(n)})$. However, this is not true in the conditional setting. For a sample $\{(\bm{x}_1,y_1),\dots,(\bm{x}_n,y_n)\}$, ordering the response variables does not induce an ordering of the conditional PITs, $F_{Y|\bm{X}} (y_k | \bm{x}_k)$; see \autoref{fig:nonstat_problem}. Therefore, QQ and EP plots cannot be produced for regression settings without the regional pooling, described in Section~\ref{sec:challenges}, where an aggregate distribution over a region is considered and ordering of the observations induces an ordering of the PITs. However, PP plots and transformed QQ plots can be produced by ordering the estimated model probabilities $\hat{F}_{Y|\bm{X}}(y_k | \bm{x}_k)$. In this case, information about the non-stationary distribution is preserved. Transformed QQ plots on exponential margins have been used as diagnostics for extreme value regression models \citep[see, e.g.,][]{heffernan2001extreme, richards2022modelling, murphy2024inference}. We discuss these plots in more detail in Section~\ref{sec:exp-plot}. 

Conditional PIT values can be used to construct regional and local diagnostics. Suppose now that the covariate domain $\mathcal{X}$ is partitioned into $B$ non-overlapping regions, or `bins', $\mathcal{B}_1,\dots,\mathcal{B}_B$ such that $\cup_{b=1}^B\mathcal{B}_b=\mathcal{X}$ and $\cap_{b=1}^B\mathcal{B}_b=\emptyset$. Denote the random covariate in bin $\mathcal{B}_b$ as ${\bm{X}_b \coloneq \bm{X} | (\bm{X}\in \mathcal{B}_b)}$. We then define global, regional, and local PIT residuals as\break ${P\coloneq F_{Y|\bm{X}}(Y|\bm{X})}$, $P_b\coloneq F_{Y|\bm{X}}(Y|\bm{X}_b),$ and $P_{\bm{x}}\coloneq F_{Y|\bm{X}}(Y|\bm{x})$, respectively, and the corresponding model-based estimates as $\hat{P}$, $\hat{P}_b,$ and $\hat{P}_{\bm{x}}$. Diagnostic plots can be produced for either the global, regional, or local PITs, and used to assess overall goodness of fit. The need for local or regional diagnostics was emphasised by \citet{zhao2021diagnostics}, who showed that diagnostics for global model performance can fail to identify models that perform poorly at a local level. For example, if $\hat{F}_{Y|\bm{X}} = F_{Y|\bm{X}},$ then $\hat{P} \sim \mbox{U}(0,1)$, but the reverse implication is not true. Instead, they noted that $\hat{F}_{Y|\bm{X}} = F_{Y|\bm{X}}$ if and only if $\hat{P}_{\bm{x}} \sim \mbox{U}(0,1)$ for each $\bm{x}\in\mathcal{X}$. An equivalent condition is $\hat{P}\sim \mbox{U}(0,1)$ and $\hat{P} \perp \!\!\! \perp \bm{X}$, that is, global PITs are uniform and independent of the covariate. 

Three strategies for PIT-based model checking have been applied in the literature: (i) testing uniformity of $\hat{P}$ and independence from $\bm{X}$; (ii) testing the uniformity of $\hat{P}_{\bm{x}}$ at each value $\bm{x} \in \mathcal{X}$; and (iii) testing uniformity of $\hat{P}_b$ in each bin. Whilst strategy (i) is perhaps the most straightforward, it has drawbacks. For covariates with low dimension $d\leq2$, graphical methods can be applied to visually check for dependence. For example, \citet{diebold1998} applied this strategy in the context of probabilistic forecast evaluation. In moderate to high dimensions, modern statistical tests of independence, such as distance correlation \citep{szekely2007measuring} or the Hilbert-Schmidt Independence Criterion \citep{gretton2005}, can be used to test independence of $\hat{P}$ and $\bm{X}$. However, these tests tend to have high power in detecting dependence, but low descriptive power in explaining the type of dependence detected. This means that when the diagnostics indicate the model fit needs to be improved, i.e., $\hat{P}$ and $\bm{X}$ are dependent, it is then difficult to determine how this should be done. In this respect, strategies (ii) and (iii) can be more illuminating. 

Implementation of strategy (ii) is more challenging. As noted above, if $\bm{X}$ has a continuous density, the expected number of observations at any given $\bm{x}$ is zero, so the distribution of $\hat{P}_{\bm{x}}$ cannot be estimated by empirical ranking. \citet{zhao2021diagnostics} proposed a solution to this by fitting an additional regression model to estimate $\hat{P}_{\bm{x}}$ at each $\bm{x}\in\mathcal{X}$. This adds an additional layer of model complexity to the production of the diagnostics. It would be preferable to have simple diagnostics that can quantify model performance without substantial additional computations. 

A compromise is to follow strategy (iii) and consider regional diagnostics. Although a uniform distribution of $\hat{P}_b$ in each bin does not guarantee model consistency, if the bins $\mathcal{B}_b$ are sufficiently small then regional diagnostics provide a reasonable indication of good fit. The uniformity of $\hat{P}_b$ can be assessed using the standard visual diagnostics (PP and transformed QQ plots) that are used for univariate extreme value models. In a non-extremal setting, regional PIT histograms have been used to check for uniformity and assess regional performance by, e.g., \citet{hamill2001interpretation, gneiting2007probabilistic}. For QR models, a commonly used regional diagnostic is the `\textit{worm plot}' \citep{buuren2001worm, buuren2007worm}. This consists of a series of detrended transformed QQ plots for covariate bins, with the transformation $F_0 = \Phi$ for $\Phi$ the standard normal distribution (see Eq.~\eqref{eq:transform_QQ}). Denote the $n_b$ observations of the covariate $\bm{X}$ in bin $\mathcal{B}_b$ by $\bm{x}_i^{(b)},$ for $i=1,\dots,n_b$ with $\sum_{b=1}^Bn_b=n$, and the corresponding observed responses by $\{y^{(b)}_i\}_{i=1}^{n_b}$. Denote the estimated non-exceedance probabilities at the observed values in bin $\mathcal{B}_b$ as $\hat{p}^{(b)}_{i} = \hat{F}_{Y|\bm{X}}(y^{(b)}_{i}|\bm{x}_{i}^{(b)})$, $\bm{x}^{(b)}_i\in\mathcal{B}_b$, $i=1,\dots,n_b$. The ordered model probabilities are denoted $\hat{p}^{(b)}_{(1)} \le \cdots \le \hat{p}^{(b)}_{(n_b)}$. The worm plot consists of the points
\begin{equation*}
    \left\{ \left(\Phi^{-1} \left(p_{k}^{(b)}\right), \Phi^{-1} \left(\hat{p}^{(b)}_{(k)}\right) - \Phi^{-1} \left(p_{k}^{(b)}\right)\right) : k=1,\dots,n_b \right\},
\end{equation*}
where $p_{k}^{(b)}$ are reference probability levels, usually defined as $p_{k}^{(b)}=(k-0.5)/n_b$. The abscissa is a reference normal quantile, and the ordinate is the deviation of the normalised quantile residuals from the reference levels\footnote{When the model distribution function is continuous, the quantity $\Phi^{-1} \left(\hat{p}^{(b)}_{(k)}\right)$ is the randomised quantile residual, defined by \citet{dunn1996randomized}.}. Confidence bounds for the deviations are often included, based on the assumption that deviations are normally distributed with zero mean and variance of the $k$-th ordinate being $p_{k}^{(b)} (1-p_{k}^{(b)})/\left(n_b \phi^2\left( \Phi^{-1}(p_{k}^{(b)})\right)\right)$, where $\phi$ is the standard normal density function. This assumption is asymptotically exact for the central quantiles \citep[see Thm. 10.3 of][]{David2003}, but is less accurate for extremes, making worm plots a poor diagnostic tool for extreme value analyses; in the limit as $n_b \to \infty$, the sampling distribution of the largest normal order statistic converges to a Gumbel distribution, for which the normal approximation is not accurate. 

Visual assessment of diagnostic plots across a small number of covariate bins is feasible, but manually checking separate plots to assess uniformity across all bins $\mathcal{B}_1,\dots,\mathcal{B}_B$ becomes impractical as $B$ grows. While one could overlay transformed QQ plots for each bin $\mathcal{B}_b$ on a single plot, without some prior standardisation, the varying sample size $n_b$ impacts the sampling distributions for the bin-wise order statistics. For example, in the case of the worm plot, the variance of the central order statistics is $O(1/n_b)$, and the variance of the largest order statistic is $O(1/\log(n_b))$, making visual assessment of the significance of deviations for different sample sizes problematic.

The use of standardised diagnostics has been considered in the non-extremal setting. The `\textit{stabilized probability plot}' \citep{michael1983stabilized} uses the sine-squared distribution, with distribution function $F_S(s) = \sin^2(\pi s/2)$, $s\in[0,1]$. This distribution has the property that, as the sample size tends to infinity, the central order statistics have equal variance: if $S_{(1)}\le \cdots \le S_{(n)}$ are an ordered sample from the sine-squared distribution, then, as $n\to\infty$ and $k/n\to p\in(0,1)$, the variance of $nS_{(k)}$ is $1/\pi^2$. However, the result does not hold for the upper or lower extremes, making this type of plot less useful as a diagnostic for extreme value models. 

Another issue to consider when overlaying diagnostic plots for multiple regions is the plotting positions. If we wish to see if there is any systematic bias in the model at a given rank or quantile, then we can calculate the mean value of the diagnostic over each region. When doing this, it is important not to conflate bias from the model with bias from the choice of plotting positions. A further disadvantage of the worm plot is that there is no closed form solution for the expected values of normal order statistics \citep{royston1982}, so any choice of plotting position $p_{k}^{(b)}$ will introduce some bias.

The worm plots and stabilized probability plots motivate a similar development for extremal regression models. The diagnostics we propose below aim to overcome some of the drawbacks with existing approaches, outlined above.

%%%%%%%%%%%%%%%%%%%%%%%%%%%%%%%%
\section{Standardised regional diagnostic plots} \label{sec:standard_regional}
\subsection{Outline} \label{sec:plot_outline}
We now propose diagnostics which aim to meet objectives (a)-(f) in Section~\ref{sec:challenges}. Our approach follows the third strategy for assessment of models for conditional distributions, described in Section~\ref{sec:conditional_dist}, based on a regional uniformity assessment of conditional PIT values. The two key limitations with existing approaches are that sampling properties vary with sample size and plotting positions are often biased, depending on the type of plot used. We propose two transformations, to ensure that diagnostics for multiple regions can be overlaid on a single plot whilst maintaining (near-)constant sampling properties; in this way, the significance of deviations can be visually assessed. The first diagnostic, referred to as the \textit{standardised tail plot}, is a detrended QQ plot, similar to the worm plot, but based on a transformation to standard exponential margins, placing greater emphasis on the upper tail of the distribution. The standardised tail plot also exploits unbiased plotting positions for the exponential distribution, discussed in Section~\ref{sec:exp-plot}. Exponential QQ plots are discussed in Section~\ref{sec:expQQ}, where we show that it is difficult to assess the significance of deviations for different sample sizes using this type of plot. To standardise for different sample sizes, we make use of an asymptotic property of the sampling distributions of exceedance probabilities, discussed in Section~\ref{sec:sampling_dist}. The standardised tail plot is then defined in Section~\ref{sec:standard_tail_plot}. The other diagnostic proposed here, which we term the \textit{normalised residual plot}, is introduced in Section~\ref{sec:normalised_residual}. The normalised residual plot provides a visualisation of the deviations between the model and observations over all probability levels, rather than just the upper tail.

The null hypothesis we wish to test is that the estimated model is equal to the true data-generating distribution: $\hat{F}_{Y|\bm{X}}(y|\bm{x}) = F_{Y|\bm{X}}(y|\bm{x})$ for every $y\in\mathcal{Y}$ and $\bm{x}\in\mathcal{X}$. Throughout this section, we assume that diagnostics are evaluated for an independent hold-out sample, that is not used for model estimation. Then, under the null hypothesis, any differences between the model and observations are due to random sampling effects only. This enables uncertainty bounds to be calculated in terms of the sampling properties of the observations. Confidence intervals for diagnostics applied to training data are smaller than those for out-of-sample data, since parameter optimisation finds the model that is `closest' to the observations, in the sense defined by the loss function, effectively `absorbing' some of the random variation. If an independent hold-out sample is not available, and the model is compared to the data used for parameter estimation, then alternative methods for calculating confidence intervals will be needed, such as a parametric bootstrap; however, this is not considered further here.

%%%%%%%%%%%%%%%%%%%%%%%%%%%%%%%%%
\subsection{Unbiased plotting positions for exponential quantiles}\label{sec:exp-plot}
In Section~\ref{sec:stationary_plots}, we discussed various ways of defining empirical exceedance probabilities as plotting positions for visual diagnostics. Here, we consider unbiased plotting positions for the exponential distribution. In the diagnostic plots proposed below, information for multiple regions of covariate space is overlaid in a single plot. It is therefore important that plotting bias is not conflated with model bias. 

Observations in bin $\mathcal{B}_b$ are assumed to be realisations of a sequence of independent and identically distributed (iid) pairs $\{(\bm{X}_i^{(b)},Y_i^{(b)})\}_{i=1}^{n_b}$. 
Define PIT residuals $P_i^{(b)} = F_{Y|\bm{X}}(Y_i^{(b)}| \bm{X}_i^{(b)})$ with $Q_i^{(b)} = 1 - P_i^{(b)}$, and $Z_i^{(b)} = -\log (Q_i^{(b)})$, for $i=1,\dots,n_b$. Since $P_i^{(b)} \sim \mbox{U}(0,1)$, it follows that $Z_i^{(b)} \sim \mbox{Exp}(1)$. Denote the ordered exceedance probabilities by $Q^{(b)}_{(1)} \le \cdots \le Q^{(b)}_{(n_b)}$ and corresponding exponential order statistics by $Z_{(1)}^{(b)} \ge \cdots \ge Z_{(n_b)}^{(b)}$. Note that we rank from most extreme to least extreme, so that $Z_{(1)}^{(b)} = -\log\big(Q_{(1)}^{(b)}\big)$ corresponds to the observation with lowest conditional exceedance probability in bin $\mathcal{B}_b$.

Since $Q_{(k)}^{(b)}$ is a random variable and, in typical applications, the data-generating distribution, $F_{Y|\bm{X}}$, is unknown, we must make some `best guess' of the sample probabilities. It is well-known that $Q_{(k)}^{(b)} \sim \mbox{Beta}(k,n_b-k+1)$\footnote{Strictly, as the distribution of $Q_{(k)}^{(b)}$ is dependent on both $k$ and $n_b$, we could use notation $Q_{(k,n_b)}^{(b)}$ to make this explicit. However, for simplicity we opt for the notation $Q_{(k)}^{(b)}$ and assume that the dependence on $n_b$ is understood implicitly.} \citep[see, e.g.,][]{David2003}. The expected values of the ranked exceedance probabilities are therefore $q_k^{(b)}\coloneq \mathbb{E}\left[Q_{(k)}^{(b)}\right] = k/(n_b+1)$. However, due to the nonlinear transformation {$Z_{(i)}^{(b)} = -\log (Q_{(i)}^{(b)})$}, we have that $\mathbb{E}\left[Z_{(k)}^{(b)}\right] \ne -\log\left(\mathbb{E}\left[Q_{(k)}^{(b)}\right]\right)$. A nice feature of the exponential distribution is that the expected values of the order statistics have the simple closed form expression \citep{David2003}
\begin{equation} \label{eq:EZk}
    z_k^{(b)} \coloneq \mathbb{E}\left[Z_{(k)}^{(b)}\right] = H_{n_b} - H_{k-1},
\end{equation}
where $H_k = \sum_{j=1}^k \frac{1}{j}$ is the $k$-th harmonic number, and we define $H_0 = 0$. Harmonic numbers have the well-known property
\begin{equation} \label{eq:Hn_approx}
    H_n = \log(n) + \gamma + O(1/n), 
    \end{equation}
where $\gamma\approx 0.57721$ is the Euler-Mascheroni constant. Therefore, the difference between the unbiased plotting position
$z_1^{(b)}$ and commonly-used plotting position $-\log\left(q_1^{(b)}\right)$ tends to $\gamma$ as $n_b\to\infty$. For example, for $n_b=100$ we have $z_1^{(b)} \approx 5.19$, whereas $-\log \left( q_1^{(b)} \right) \approx 4.62$. So, if diagnostics for multiple regions are overlaid, then biases caused by using $-\log \left( q_k^{(b)} \right)$, rather than $z_k^{(b)}$, as a plotting position may be wrongly-interpreted as model bias.

%%%%%%%%%%%%%%%%%%%%%%%%%%%%%%%%%%%%%%%%
\subsection{Exponential QQ plots revisited} \label{sec:expQQ}
For a model estimate of the conditional distribution function, $\hat{F}_{Y|\bm{X}}$, we denote estimated exceedance probabilities in bin $\mathcal{B}_b$ as $\hat{q}_k^{(b)} = 1 - \hat{F}_{Y|\bm{X}}(y_k^{(b)}|\bm{x}_k^{(b)})$. Denote the ordered values $\hat{q}_{(1)}^{(b)}\le \cdots \le \hat{q}_{(n_b)}^{(b)}$ and exponential order statistics from the model as $\hat{z}_k^{(b)} = -\log(\hat{q}_{(k)}^{(b)})$, $k=1,\dots,n_b$. We define the \textit{regional exponential QQ plot} to consist of lines joining the pairs 
\begin{equation}
    \{(\hat{z}_k^{(b)}, z_k^{(b)}) : k=1,\dots,n_b\},
\end{equation} 
for each bin $\mathcal{B}_1,\dots,\mathcal{B}_B$. Compared to the regionally-aggregated QQ plots described in Section~\ref{sec:challenges}, the exponential QQ plots are independent of the shape and scale of the upper tails of the underlying conditional distributions, and preserve information about non-stationarity. However, they are not invariant to the sample size, $n_b$, which may differ between bins. 

To illustrate this, the left plot in \autoref{fig:exp_QQ} shows a two-sided 95\% confidence interval (CI) on exponential order statistics for sample sizes of $n_b=10,10^2,10^3,10^4$. It is evident that what constitutes a large deviation in the exponential QQ plot for a small sample size is different to that for a larger sample size. The right hand plot in \autoref{fig:exp_QQ} shows an example with simulated data. In this case, data for $B=100$ regions have been generated, with each region having sample size $n_b$ realised from a random variable $N$, where $\log_{10}(N)$ is uniformly distributed in $[1,4]$. For each sample, the model values $\{\hat{q}_k\}_{k=1}^{n_b}$ are sampled from a uniform distribution over $[0,1]$, representing a perfectly calibrated model. Although there are no errors in the model, the individual QQ plots deviate from the 1:1 line (diagonal) and, due to the variably in the regional sample size $n_b$, it is not possible to visually assess where these deviations are significant. 

\begin{figure}[h!]
	\centering
	\includegraphics[width=0.49\textwidth]{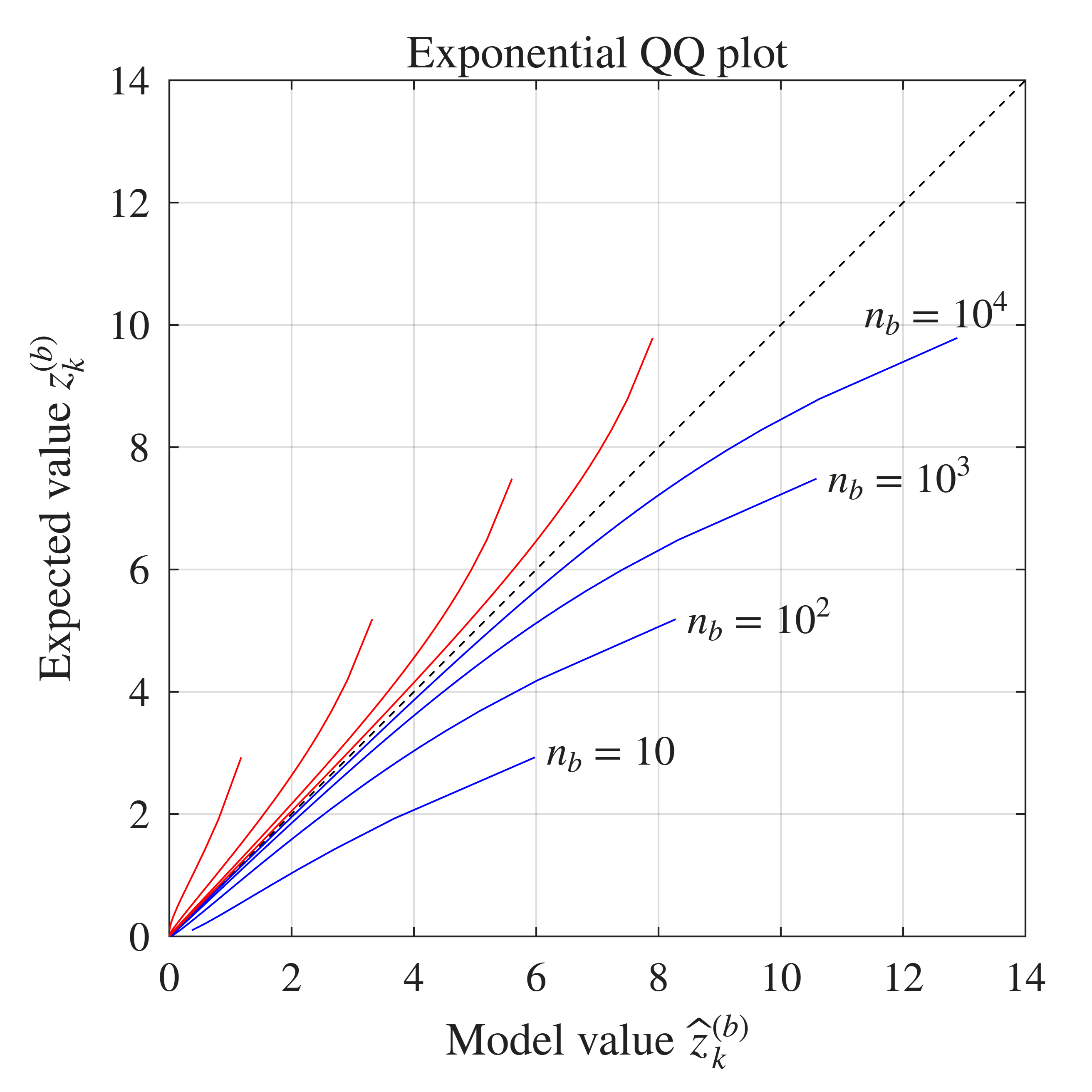}
    \includegraphics[width=0.49\textwidth]{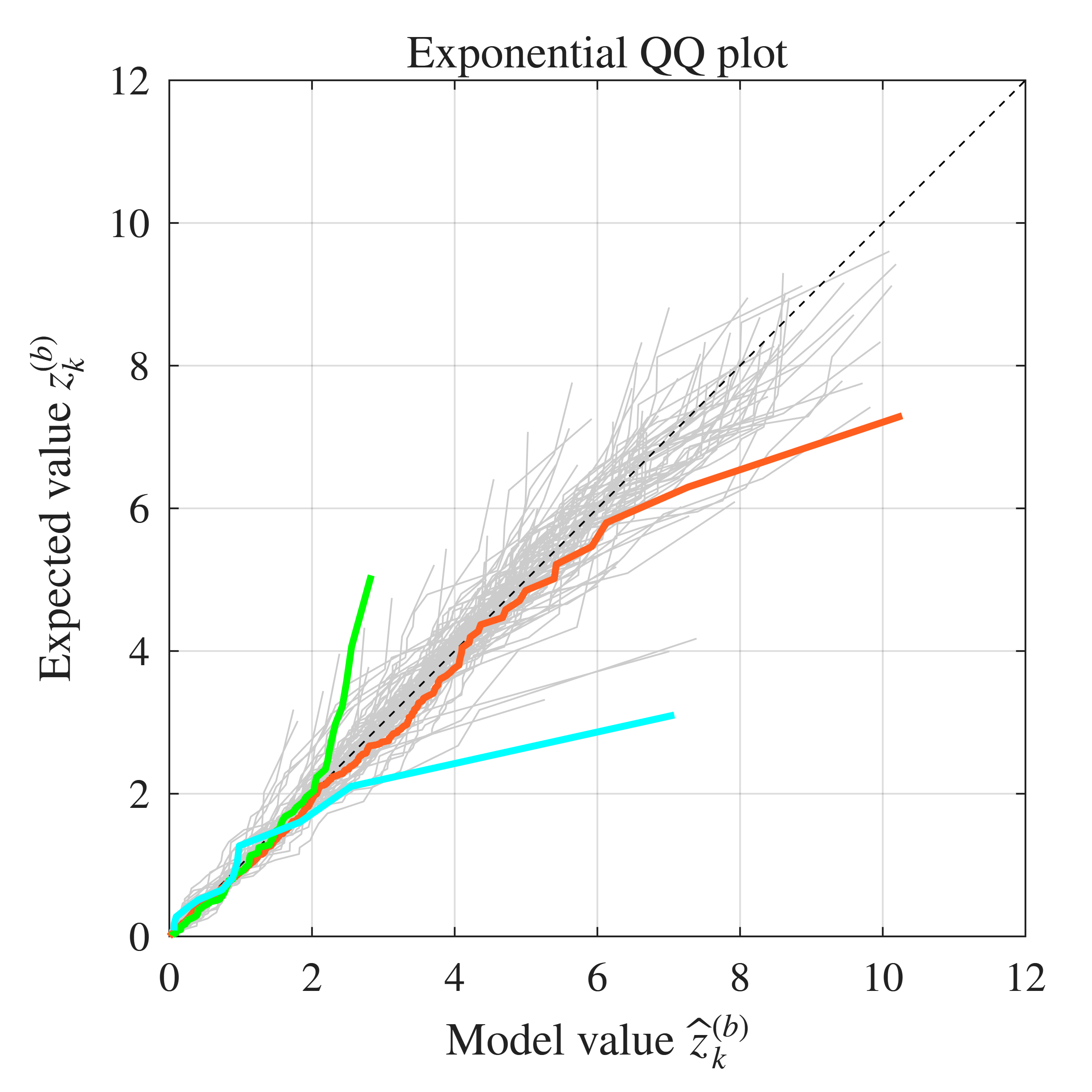}\\
    \includegraphics[width=0.49\textwidth]{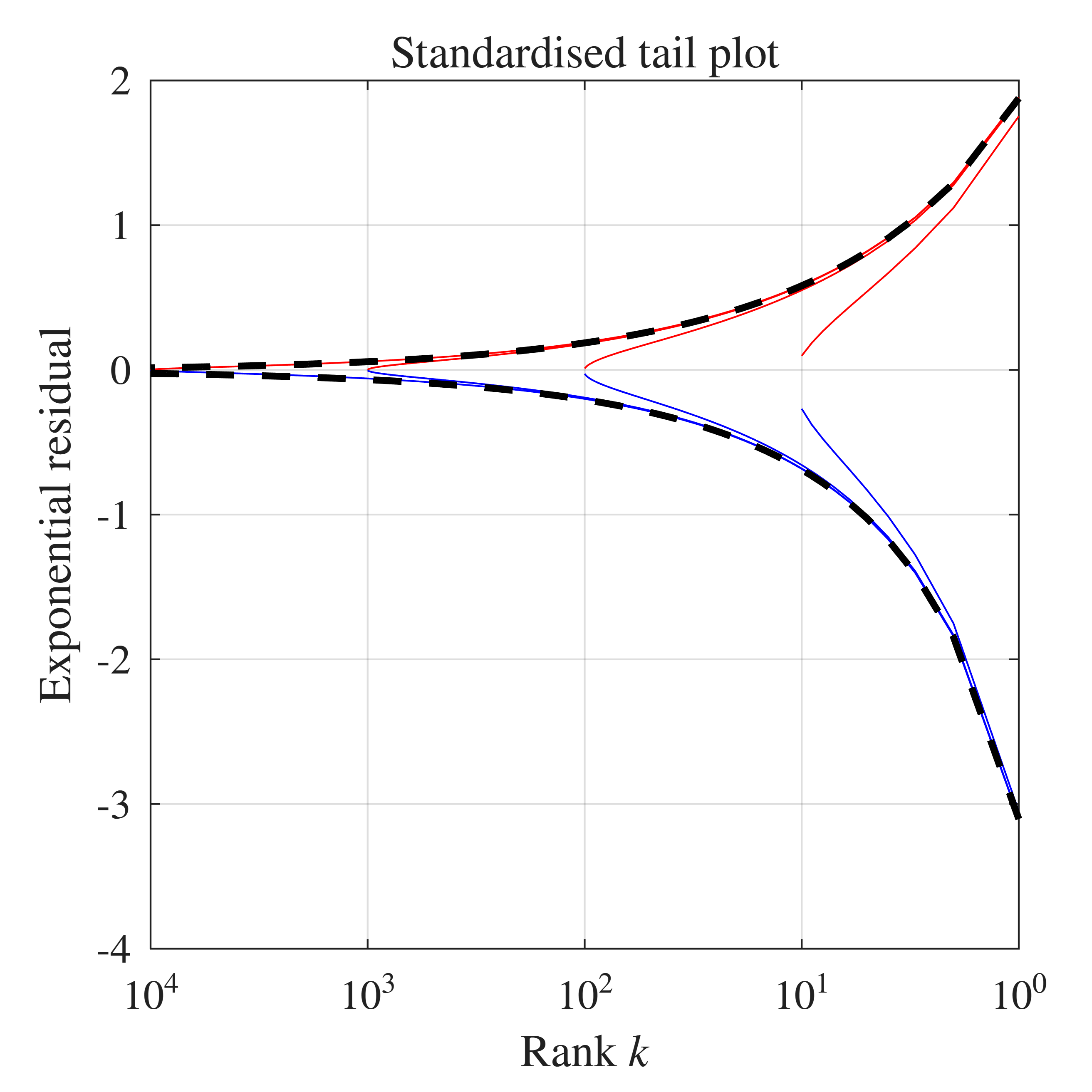}
    \includegraphics[width=0.49\textwidth]{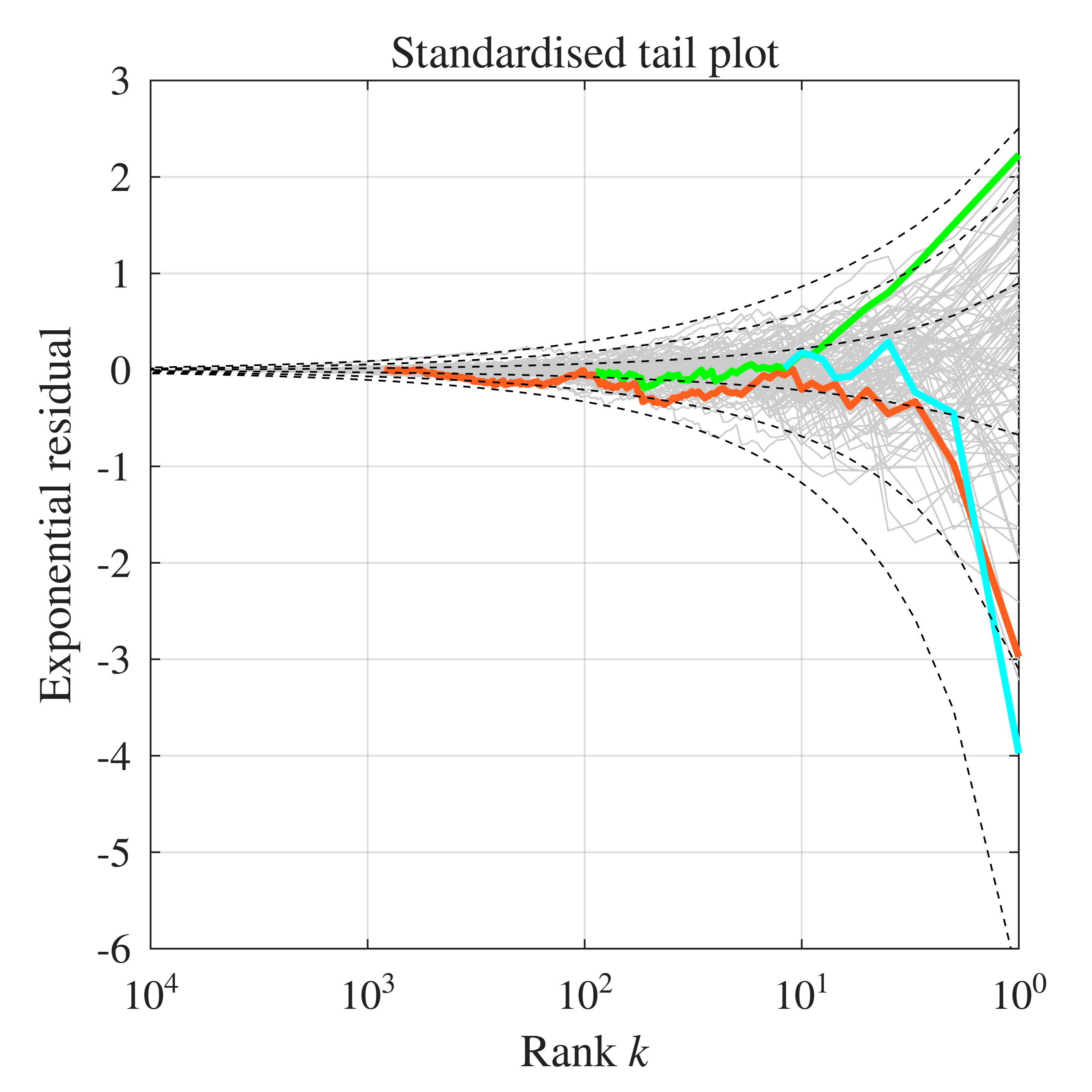}
	\caption{Top left: Two-sided 95\% CI on exponential quantiles for sample sizes of $n_b=10,10^2,10^3,10^4$. Bottom left: Two-sided 95\% CI on exponential residuals as a function of rank for the same sample sizes, together with 95\% CI for the asymptotic distribution (black dashed lines). Top right: Exponential QQ plot for simulated data for 100 samples of random size $N$, where $\log_{10}(N)\sim \mbox{U}\,(1,4)$. Bottom right: Standardised tail plot for the same data, together with quantiles of the asymptotic distribution at non-exceedance probabilities 0.001, 0.025, 0.25, 0.75, 0.975, and 0.999 (dashed lines). Highlighted lines correspond to the same samples in right hand column.}
	\label{fig:exp_QQ}
\end{figure}

%%%%%%%%%%%%%%%%%%%%%%%%%%%%%%%%%%%
\subsection{Asymptotic sampling distributions} \label{sec:sampling_dist}
To standardise the exponential QQ plots and enable a visual assessment of the significance of deviations, we make use of an asymptotic property of exceedance probabilities. The following theorem shows that the asymptotic distribution of the normalised exceedance probabilities associated with the most extreme observations is independent of the sample size, and depends only on the rank $k$.

\begin{theorem} \label{thm:normalised_exceedance}
Let $Q_{(1)} \le \cdots \le Q_{(n)}$ be the order statistics of iid $\mbox{U}(0,1)$ variables. For any fixed $k\in \mathbb{N}_{>0}$, the normalised variable $n Q_{(k)}$ converges in distribution, as $n\to\infty$, to a gamma random variable with shape parameter $k$ and unit scale. 
\end{theorem}

Theorem \ref{thm:normalised_exceedance} can be used to show that the sampling distribution of the difference between the $k$-th exponential order statistic and its expected value also converges to an asymptotic form that is independent of sample size, as stated in the corollary below.

\begin{corollary} \label{cor:Z_difference}
Let $Z_{(1)} \ge \cdots \ge Z_{(n)}$ be the order statistics of iid $\mbox{Exp}(1)$ variables. For any fixed $k\in \mathbb{N}_{>0}$, the difference $D_k = \mathbb{E}\big[Z_{(k)}\big] - Z_{(k)}$, between the expected value of the $k$-th exponential order statistic and the observed value, converges in distribution, as $n\to\infty$, to a log-gamma random variable, $D_{k,\infty}$, with shape parameter $k$, unit scale, and location $\mu_k=\gamma - H_{k-1}$. The density function of $D_{k,\infty}$ is
\begin{equation}
    f_{D_{k,\infty}} (x) = \frac{1}{\Gamma(k)} \exp\left[ k \left(x - \mu_k\right) - \exp\left(x - \mu_k\right)\right], \quad x\in\RR,
\end{equation}
where $\Gamma$ is the gamma function.
\end{corollary}

Both Theorem \ref{thm:normalised_exceedance} and Corollary \ref{cor:Z_difference} are special cases of the $k$-largest order statistic model \citep{weissman1978estimation, smith1986extreme}. Empirical assessment of the rate of convergence of $D_k$ is shown in \autoref{fig:density} for ranks $k=1$, 2, 10, and 100, and sample sizes $n/k=5$ and $n/k=10$. The density converges quickly to the asymptotic limit, showing that the asymptotic model is a reasonable approximation for relatively small values of $n/k$. We use this result to produce a diagnostic plot that is approximately independent of sample size, and thus can be used to compare model performance across regions with varying sample size $n_b$. The proofs of Theorem~\ref{thm:normalised_exceedance} and Corollary \ref{cor:Z_difference} are provided in Appendix~\ref{app:proof_thm:normalised_exceedance} and \ref{app:proof_cor:Z_difference}, respectively.

\begin{figure}[t]
	\centering
	\includegraphics[width=\textwidth]{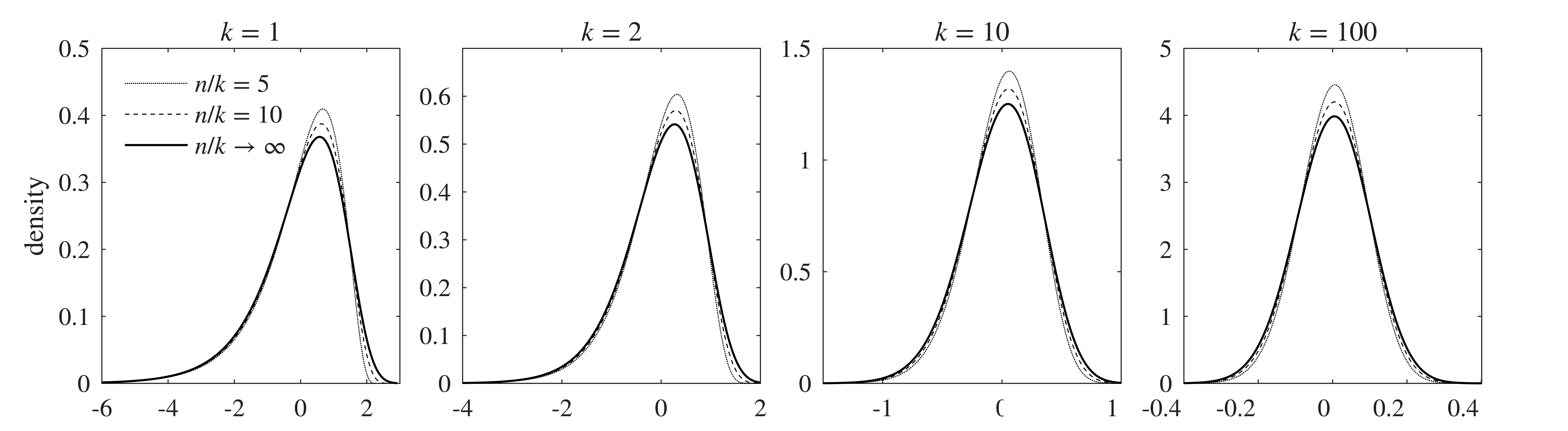}
	\caption{Densities of standardised exponential order statistics $D_k =\mathbb{E} \left[ Z_{(k)} \right] - Z_{(k)}$, for various ratios of ranks $k$ to sample sizes $n$.}
	\label{fig:density}
\end{figure}

%%%%%%%%%%%%%%%%%%%%%%%%%%%%%%%
\subsection{Standardised tail plots} \label{sec:standard_tail_plot}
Corollary \ref{cor:Z_difference} can be used to standardise the exponential QQ plots to have near-constant sampling variability across bins. Since the samples in each region are independent, the sampling distribution of $D_k=\mathbb{E}\left[Z_{(k)}^{(b)}\right] - Z_{(k)}^{(b)}$ is, asymptotically, dependent on rank only. Thus if we plot the difference $z_k^{(b)} - \hat{z}_k^{(b)}$ against rank $k$, under the null hypothesis, for a given $k$, the expected difference will be zero and, as $n_b\to\infty$, the differences will follow a log-gamma distribution. We refer to the difference $z_k^{(b)} - \hat{z}_k^{(b)}$ as the \textit{exponential residual}. We define the \textit{standardised tail plot} to be a plot of the exponential residuals against rank, consisting of lines joining the pairs
\begin{equation}
    \{(k, z_k^{(b)} - \hat{z}_k^{(b)}) : k=1,\dots,n_b\},
\end{equation}
for each bin $\mathcal{B}_1,\dots,\mathcal{B}_B$ in the covariate domain. The ranks are shown in reverse order, so that rank $k=1$ is on the right side of the plot, with rank increasing towards the left. In this way, the observations become more `extreme' towards the right, in terms of decreasing conditional exceedance probability within each bin. 

Asymptotic confidence intervals can be added to the standardised tail plots using the quantiles of the log-gamma random variable $D_{k,\infty}$. The quantile of $D_{k,\infty}$ at non-exceedance probability $\alpha\in[0,1]$, denoted $d_{k,\alpha}$, can be computed in terms of the quantiles of the gamma distribution as
\begin{equation*}
    d_{k,\alpha} = \log\left(a_{k,\alpha}\right) - H_{k-1} + \gamma,
\end{equation*}
where $a_{k,\alpha}$ is the $\alpha$-quantile of the gamma distribution with shape $k$ and unit scale. 

The lower left plot in \autoref{fig:exp_QQ} shows a two-sided 95\% CI for the exponential residuals for sample sizes $n_b=10,10^2,10^3,10^4$, together with a 95\% CI for the asymptotic log-gamma distribution. The finite sample CI lies within the asymptotic bounds. However, for ranks with $n/k\ge 10$ the agreement is close. The asymptotic confidence bounds are therefore a useful visual guide for what constitutes a significant deviation for the most extreme observations. The lower right plot in \autoref{fig:exp_QQ} shows an example of a standardised tail plot. To illustrate the transformation, the three samples highlighted in \autoref{fig:exp_QQ}, which have `large differences' in the upper tails in the exponential QQ plots, are also highlighted in the standardised tail plot. On the standardised scale, it can be seen that these differences in the upper tails are within the 99.8\% CI. Overall, although there are some samples with points outside the 95\% CI, all observations are within the 99.8\% CI. So, the immediate inference from this plot is that for any given sample, deviations between the model and expected values are within the range expected from sampling effects (as would be expected). However, this does not tell us about the overall performance of the model, that is whether it has a tendency to over- or under-predict on aggregate. This is discussed further in Section~\ref{sec:summary}.

%%%%%%%%%%%%%%%%%%%%%%%%%%%%%%%%
\subsection{Normalised residual plots} \label{sec:normalised_residual}

The standardised tail plot has the desirable feature that it is directly analogous to the commonly-used exponential QQ plots, whilst standardising the sampling properties over different sample sizes. However, as with exponential QQ plots, model discrepancy for non-extreme observations is difficult to assess visually. Another option for standardising the sampling properties of the order statistics is to apply a probability integral transform based on the exact sampling distribution of the exceedance probabilities. For a sample in bin $\mathcal{B}_b$, with ordered PITs $\hat{q}_{(1)}^{(b)} \le \cdots \le \hat{q}_{(n_b)}^{(b)}$, the p-value of the PIT $\hat{q}_{(k)}^{(b)}$ is $F_{\beta_k} \left( \hat{q}_{(k)}^{(b)} \right)$, where $F_{\beta_k}$ is the distribution function of $Q_{(k)}^{(b)} \sim \mbox{Beta}(k,n_b-k+1)$. Define  transformed p-values $\nu_{k}^{(b)} = \Phi^{-1} \left( F_{\beta_k} \left( \hat{q}_{(k)}^{(b)} \right)\right)$, $k=1,\dots,n_b$, where $\Phi$ is the standard normal CDF. We refer to $\nu_{k}^{(b)}$ as the \textit{normalised residual}, and define the \textit{normalised residual plot} to consist of lines joining the pairs
\begin{equation}
    \left\lbrace\left(k, \nu_k^{(b)}\right) : k=1,\dots,n_b\right\rbrace,
\end{equation}
for each bin $\mathcal{B}_1,\dots,\mathcal{B}_B$ in the covariate domain. The transformation to a standard normal scale enables a better visualisation of the significance of the differences than looking at the p-values themselves. Also, as we later discuss in Section~\ref{sec:RMS_normres}, under this transformation, the mean-square value of the normalised residuals, $\nu_1^{(b)}, \ldots, \nu_{n_b}^{(b)}$, is closely related to the Anderson-Darling goodness-of-fit statistic. 

The normalised residual plot shows the significance of the model deviations from the expected values (which are zero in the case of the standard normal distribution) as a function of rank. This is analogous to a PP plot (see Section \ref{sec:stationary_plots}), but gives an improved visualisation of the significance of the deviation at all probability levels, especially for large sample sizes. An alternative visualisation would be to plot $\nu_k^{(b)}$ against $k/(n_b+1)$, where the latter is the expected regional exceedance probability. The advantage to the alternative approach is that the horizontal scale is the same for any sample size, making it more akin to a probability plot. However, for consistency with the standardised tail plots, we use the rank on the horizontal axis.

\begin{figure}[t]
	\centering
    \includegraphics[width=0.5\textwidth]{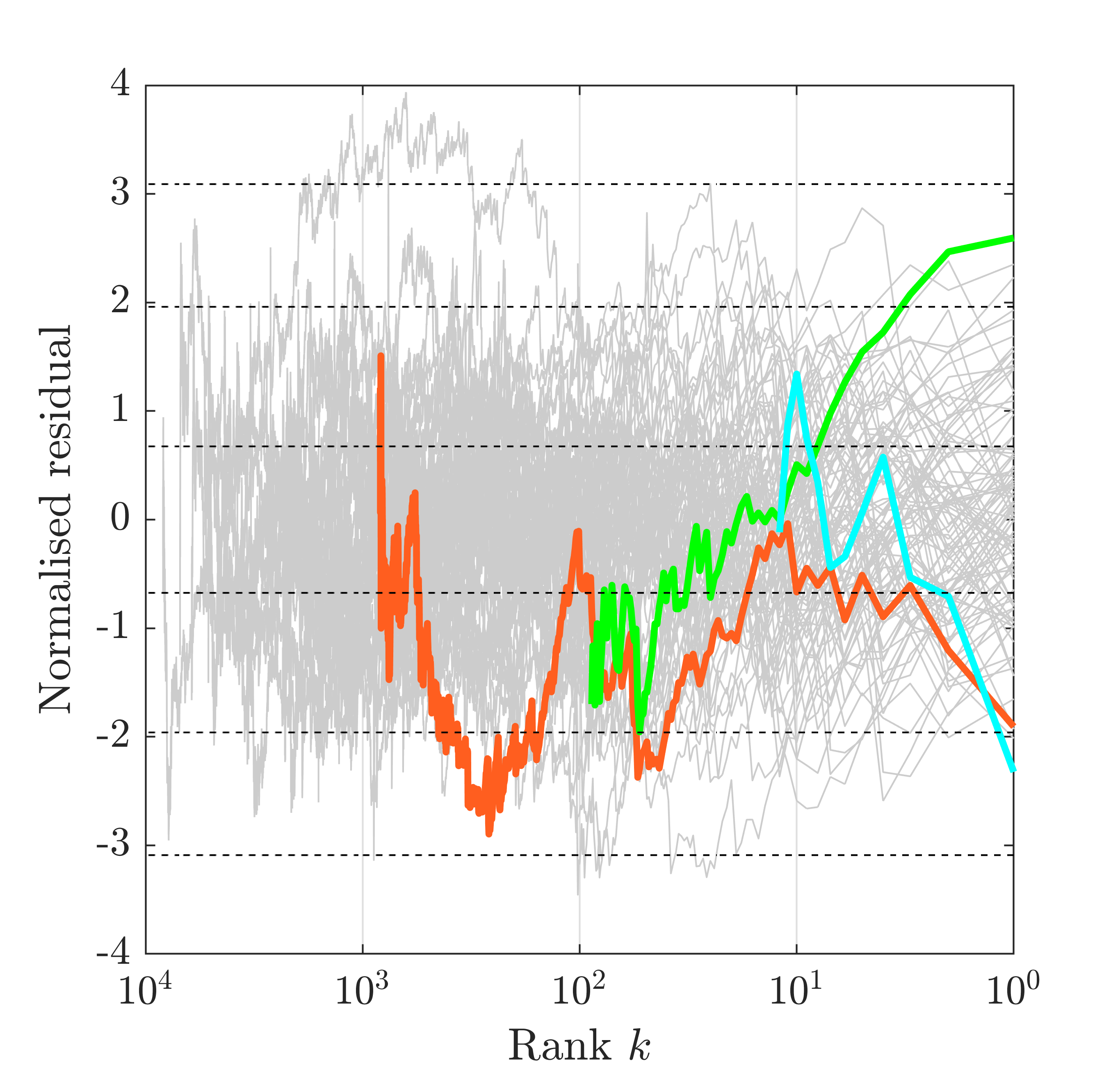}
	\caption{Normalised residual plot for the data shown in lower plots of \autoref{fig:exp_QQ}. Each line grey corresponds to an individual region. Dashed lines indicate normal quantiles at non-exceedance probabilities 0.001, 0.025, 0.25, 0.75, 0.975, and 0.999. Colours of highlighted lines correspond to the highlighted samples in \autoref{fig:exp_QQ}.}
	\label{fig:norm_res_plot}
\end{figure}

An example normalised residual plot is shown in \autoref{fig:norm_res_plot}, using the same simulated data used in \autoref{fig:exp_QQ}. There are some differences in excess of the 99.8\% CI which were not visible in the standardised tail plot (\autoref{fig:exp_QQ}), as they occur at higher ranks which have narrow sampling confidence intervals for the exponential quantiles. The normalised residual plot does not require any asymptotic assumptions, so they are applicable for any ratio $n/k$ of sample size to rank, whereas the asymptotic confidence bounds in the standardised tail plot are only a good approximation for $n/k \gtrsim 10$. However, the standardised tail plot places more emphasis on the tails of the distribution in each bin, which is perhaps more in the spirit of diagnostics for extreme value models.
    
%%%%%%%%%%%%%%%%%%%%%%%%%%%%%%%%
\section{Summarising regional and global model performance} \label{sec:summary}
\subsection{Outline}
The diagnostic plots introduced in Section \ref{sec:standard_regional} provide a visual check on model performance in each region. To assess the overall model performance, it is useful to aggregate and summarise this information. The information in the regional diagnostic plots can be summarised either by considering the performance of the model in each region (across all ranks), or by considering the model performance across all regions at a fixed rank $k$. For the regional summaries, we can apply standard tests for goodness-of-fit. As these summary statistics form an integral part of our proposed diagnostics, we start in Section~\ref{sec:GOF} by considering the suitability of these statistics for detecting lack of fit in the upper tail of the distribution. Of the many goodness-of-fit tests, we consider the widely-used Cram\'{e}r-von Mises (CvM) family (Section~\ref{sec:CvM}). To complement our new visual diagnostics, we also introduce two new test statistics, which are defined via the deviations visualised in the standardised tail plot (Section~\ref{sec:EMAD}) and normalised residual plots (Section~\ref{sec:RMS_normres}), and derive their asymptotic properties. The statistic defined in terms of the normalised residuals is shown to be asymptotically equivalent to the Anderson-Darling statistic. Section~\ref{sec:sensitivity} considers the asymptotic sensitivity of test statistics to deviations in various parts of the estimated distribution, and shows that the new statistic we propose in Section~\ref{sec:EMAD} has a greater sensitivity in the upper tail than the other test statistics considered. In Section~\ref{sec:summary_plot} we discuss how these summary statistics can be used to form regional and rank-based summary plots. Finally, we discuss global summary statistics in Section~\ref{sec:global_summary}. As in Section \ref{sec:standard_regional}, we assume that summary statistics are calculated for an independent hold-out sample, not used for model estimation. In this case, p-values for the summary statistics can be calculated quickly from pre-computed lookup tables. However, if an independent hold-out sample is not available, then p-values can be estimated using, e.g., a parametric bootstrap.

%%%%%%%%%%%%%%%%%%%%%%%%%%%%%%%%%%%%%
\subsection{Goodness-of-fit statistics} \label{sec:GOF}
\subsubsection{Cram\'{e}r-von Mises family} \label{sec:CvM}
The Cram\'{e}r-von Mises (CvM) family of goodness-of-fit statistics give a weighted measure of the uniformity of the PIT values in terms of their empirical distribution function (EDF) \citep{stephens1986}. For notational simplicity, we here consider suppose that $\hat{q}_{(1)}\le \cdots \le \hat{q}_{(n)}$ are the model global PIT values, but note that the following definitions can be adapted for regional PIT values. The EDF of the PITs is defined as 
\begin{equation} \label{eq:PIT_EDF}
    \mathbb{F}_{n}(u) = \frac{1}{n} \sum_{i=1}^{n} \mathbf{1}\left(\hat{q}_{(i)}\le u\right), \quad u\in[0,1].
\end{equation}
Then, the CvM family of test statistics is defined as
\begin{equation*}
    T_n^2 = n \int_{0}^1 \psi(u) \left(\mathbb{F}_n(u) - u\right)^2 \dd u,
\end{equation*}
where $\psi:[0,1]\to [0,\infty)$ is some non-negative weight function. When $\psi(u) = 1$, the statistic $T_n^2$ is the Cram\'{e}r-von Mises statistic, denoted $W_n^2$. When $\psi(u) = (u(1-u))^{-1},$ the statistic is the Anderson-Darling (AD) statistic, denoted $A_n^2$. When $\psi(u) = u^{-1}$, the statistic is the right-tail-weighted Anderson-Darling (ADR) statistic \citep{sinclair1990modified}, denoted $A_{R,n}^2$. Noting that $\mathbb{F}_{n}(u)=i/n$, $k=1,\dots,n-1$, on the interval $u \in \left[\hat{q}_{(i)},\hat{q}_{(i+1)}\right)$, the integrals can be evaluated explicitly in terms of $\hat{q}_{(i)}$, and the CvM, AD, and ADR test statistics can be computed as
\begin{align*}
    \text{CvM:} \quad & W^2_{n} = \frac{1}{12n} + \sum_{k=1}^{n} \left(q_{k} - \hat{q}_{(k)}\right)^2,\\
    \text{AD:} \quad & A_{n}^2 = -{n} - 2 \sum_{k=1}^{n} \left[q_{k}\,\log\left(\hat{q}_{(k)}\right) + \left(1-q_{k}\right)\,\log\left(1-\hat{q}_{(k)}\right)\right],\\
    \text{ADR:} \quad & A_{R,{n}}^2 = -\frac{3n}{2} - 2 \sum_{k=1}^{n} \left[q_{k}\,\log\left(\hat{q}_{(k)}\right) -\hat{q}_{(k)}\right],
\end{align*}
where $q_{k} = (k-0.5)/n$. To calculate the p-value of these statistics, we need to know their distributions under the null hypothesis. These do not admit a simple closed form solution, but can be computed via Monte Carlo simulation, by replacing $\hat{q}_{(k)}$ by $\text{U}(0,1)$ random variables (see Section SM1.1 of the Supplementary Material for details). Under the null hypothesis, the CvM family of statistics converge to an asymptotic distribution as $n\to\infty$. The scaled empirical process $\sqrt{n} (\mathbb{F}_n(u) - u)$ converges in distribution to $B(u)$, a Brownian bridge on $[0,1]$ \citep{shorack1986}. A Brownian bridge is a continuous Gaussian process with mean zero and covariance function $\Cov(B(s), B(t)) = \min(s,t)-st$, $s,t\in[0,1]$. Hence, under the null hypothesis,
\begin{equation} \label{eq:CvM_family_asymptotic}
    T_n^2 \xrightarrow{\hspace{2mm}d\hspace{2mm}} T^2 = \int_{0}^1 \psi(u) \left(B(u)\right)^2 \dd u.
\end{equation}
As discussed below, this can be used to compute an asymptotic sensitivity to deviations in the PIT distribution at various quantile levels $u$. 

%%%%%%%%%%%%%%%%%%%%%%%%%%%%%%%%%%%%%%%%%%%%
\subsubsection{MAD of exponential order statistics} \label{sec:EMAD}
We also consider a goodness-of-fit statistic based on the mean absolute deviation (MAD) of the exponential order statistics, defined as
\begin{equation*}
    S_{n} = \frac{1}{\sqrt{n}} \sum_{k=1}^{n} \left\lvert z_{k} - \hat{z}_k\right\rvert,
\end{equation*}
where $z_k = H_n-H_{k-1}$ and $\hat{z}_k=-\log(\hat{q}_{(k)}).$ Although the sum is normalised by $\sqrt{n}$ rather than $n$, $S_n$ can be viewed as the MAD of the scaled terms $\sqrt{n}\lvert z_{k} - \hat{z}_k\rvert$. We refer to $S_{n}$ as the MAD of the exponential order statistics, abbreviated as EMAD. The EMAD statistic can be interpreted as the MAD of the regional trajectories shown in the standardised tail plots. This and similar statistics have been used previously for threshold and model selection for extreme value regression analyses \citep[see, e.g.,][]{varty2021inference, richards2026regression}, although with the sum normalised by $n$ rather than $\sqrt{n}$. With the scaling $\sqrt{n}$, under the null hypothesis, the EMAD statistic converges to a similar asymptotic form to the CvM statistics, as given in the following proposition.

\begin{proposition}[Asymptotic distribution of EMAD] \label{prop:EMAD_limit}
Let $S_n$ be defined as above. Under the null hypothesis
\begin{equation*}
	S_n \xrightarrow{\hspace{2mm}d\hspace{2mm}} S = \int_0^1 \frac{|B(u)|}{u} \dd u, \quad n\to \infty,
\end{equation*}
where $B$ is a Brownian bridge on $[0,1]$. The asymptotic expectation and variance are $\E[S] = \sqrt{\pi/2}$ and $\Var(S) = 4 \log(2)-\tfrac{\pi}{2} - 1$, respectively.
\end{proposition}

As with the CvM test statistics, the EMAD does not have a simple closed form distribution for finite sample sizes, but the distribution under the null can be computed by Monte Carlo simulation in the same way, by replacing $\hat{z}_k^{(b)}$ by iid $\mbox{Exp}(1)$ random variables (see Section~SM1.1 of the Supplementary Material for details). The proof of Proposition~\ref{prop:EMAD_limit} is provided in Appendix~\ref{app:proof_prop:EMAD_limit}.

%%%%%%%%%%%%%%%%%%%%%%%%%%%%%%%%%%%%%%%%%%%
\subsubsection{Mean square of normalised residuals} \label{sec:RMS_normres}
The EMAD statistic is directly related to the trajectories shown in the standardised tail plots. Another goodness-of-fit statistic can be defined in terms of the trajectories shown in the normalised residual plots. Let $\nu_k$ be the normalised residuals defined in Section~\ref{sec:normalised_residual}, and define 
\begin{equation*}
    \mathcal{A}_{n}^2 = \frac{1}{n} \sum_{k=1}^{n} \nu_k^2.
\end{equation*}
The following proposition shows that, under the null hypothesis, this statistic has the same asymptotic distribution as the Anderson-Darling statistic, given in \eqref{eq:CvM_family_asymptotic}.

\begin{proposition} \label{prop:AD_equiv}
Let $\mathcal{A}_n^2$ be defined as above. Under the null hypothesis
\begin{equation*}
	\mathcal{A}_n^2 \xrightarrow{\hspace{2mm}d\hspace{2mm}}  \int_0^1 \frac{(B(u))^2}{u(1-u)} \dd u, \quad n\to \infty,
\end{equation*}
where $B$ is a Brownian bridge on $[0,1]$. 
\end{proposition}
The proof of Proposition~\ref{prop:AD_equiv} is provided in Appendix~\ref{app:proof_prop:AD_equiv}. 
Monte Carlo simulations show that the distribution of $\mathcal{A}_n^2$, for finite sample sizes, is also close to that of the AD statistic $A_n^2$ (see SM1.1 for details). Due to the similarity to the AD statistic, we do not consider $\mathcal{A}_n^2$ further. However, the relation between the two statistics provides a direct interpretation of the normalised residual plots as a visualisation of the deviations contributing to the AD statistic.

%%%%%%%%%%%%%%%%%%%%%%%%%%%%%%%%%%%%%%%%%%%
\subsubsection{Asymptotic sensitivity of test statistics} \label{sec:sensitivity}
Under the null hypothesis, the PIT values are uniformly distributed on $[0,1]$, with distribution function $F_0(u)=u$, $u \in[0,1]$. Suppose we test the distribution $F_0$ under the null hypothesis against a sequence of local alternatives that differ by $O(n^{-1/2})$, defined as $F_{\theta}(u) = u + \theta h(u)$, where $\theta =\lambda/\sqrt{n}$, $\lambda>0,$ and $h(u)$ is a local deviation function. Under the local alternatives, the empirical process $\sqrt{n}(\mathbb{F}_n(u)-u)$ acquires a deterministic drift, and converges in distribution to $B(u)+\lambda h(u)$ as $n\to\infty$ \citep{shorack1986}. Denote the CvM family test statistic under the null hypothesis as $T_{n,0}^2$ and the statistic under the local alternative $F_{\theta}$ as $T_{n,\theta}^2$. Then, from \eqref{eq:CvM_family_asymptotic}, the asymptotic bias converges, as $n\to\infty$, to 
\begin{equation}
\begin{aligned} \label{eq:CvM_sensitivity}
    \E[{T_{n,\theta}^2}] - \E[{T_{n,0}^2}] &\to \int_{0}^1 \psi(u) \E\left[\left(B(u)- \lambda h(u) \right)^2\right] \dd u - \int_{0}^1 \psi(u) \E\left[\left(B(u)\right)^2\right] \dd u\\
    &= \lambda^2 \int_{0}^1 \psi(u) h^2(u) \dd u.
\end{aligned}
\end{equation}
So, the influence of any deterministic local deviation $\lambda h(u)$ in the uniform distribution of the PIT values is weighted by $\psi(u)$. Therefore, the quantity 
\begin{equation*}
    \delta(u) \coloneq \frac{\psi(u)}{\sqrt{\Var\big(T_{0}^2\big)}},
\end{equation*}
defines a local signal-to-noise ratio for perturbations to the PIT distribution at exceedance probability $u$. For the CvM statistic, $\Var\big(W_{0}^2\big) = 1/45$, and $\delta(u) = 3\sqrt{5} \approx 6.71$ is constant with respect to $u$. For the AD statistic, $\Var\big(A_{0}^2\big) = 2\pi^2/3-6$ and $\delta(u) \approx 1.31/(u(1-u))$, whereas for the ADR statistic $\Var\big(A_{R,0}^2\big) = 1/6$ and $\delta(u) \approx 2.45/u$. So, although the AD and ADR statistics both have the same weights in the upper tail (asymptotically), the lower variance of the ADR statistic results in a greater sensitivity to perturbations in the upper tail of the distribution (at the expense of reduced sensitivity in the lower tail). 

A similar sensitivity analysis can be conducted for the EMAD statistic. We consider the same sequence of local alternatives as above, and denote the EMAD statistic under the null hypothesis as $S_{n,0}$ and the local alternative as $S_{n,\theta}$. The asymptotic sensitivity is given in the following proposition. 

\begin{proposition}[Asymptotic sensitivity of EMAD] \label{prop:EMAD_sensitivity}
Under the assumptions above, for small perturbations with $\lambda\to 0$, the asymptotic bias of the EMAD statistic for local alternatives converges, as $n\to\infty$, to
\begin{align*}
    \E[{S_{n,\theta}}] - \E[{S_{n,0}}] &\to \lambda^2 \int_{0}^1 \frac{h^2(u)}{u \sqrt{2\pi u (1-u)}} \dd u.
\end{align*}
\end{proposition}
The proof of Proposition~\ref{prop:EMAD_sensitivity} is provided in Appendix~\ref{app:proof_prop:EMAD_sensitivity}. 
The sensitivity for the EMAD statistic is of the same asymptotic form as the sensitivity for the CvM family \eqref{eq:CvM_sensitivity}, with a quadratic dependence on the perturbation $\lambda h (u)$, but with weighting function $\left(u \sqrt{2\pi u (1-u)}\right)^{-1}$. We can thus define a directly comparable asymptotic signal-to-noise ratio for the EMAD statistic as $\delta(u) = \left(u \sqrt{2\pi \Var(S_0) u (1-u)}\right)^{-1}$. So, the EMAD statistic has sensitivity $O(u^{-3/2})$ as exceedance probability $u\to0$, whereas the ADR statistic has sensitivity $O(u^{-1})$. The asymptotic signal-to-noise ratios are illustrated in \autoref{fig:test_stat_sensitivity}. The EMAD statistic is more sensitive than the ADR statistic for exceedance probabilities less than $\approx0.155$. This suggests that the EMAD may be more appropriate than the ADR as a goodness-of-fit statistic for models targetting the upper tails of the distribution.

For extreme value models the sample size is often small due to the scarcity of extreme observations, so this asymptotic analysis can only be taken as an indication. In Section~SM1.3 of the supplementary material, we present the results of a simulation study which shows that for finite sample sizes $n=25$, 50, and 100, the EMAD statistic is more sensitive than the ADR to perturbations in the upper tail of the distribution. 

\begin{figure}[t]
	\centering
    \includegraphics[width=0.6\textwidth]{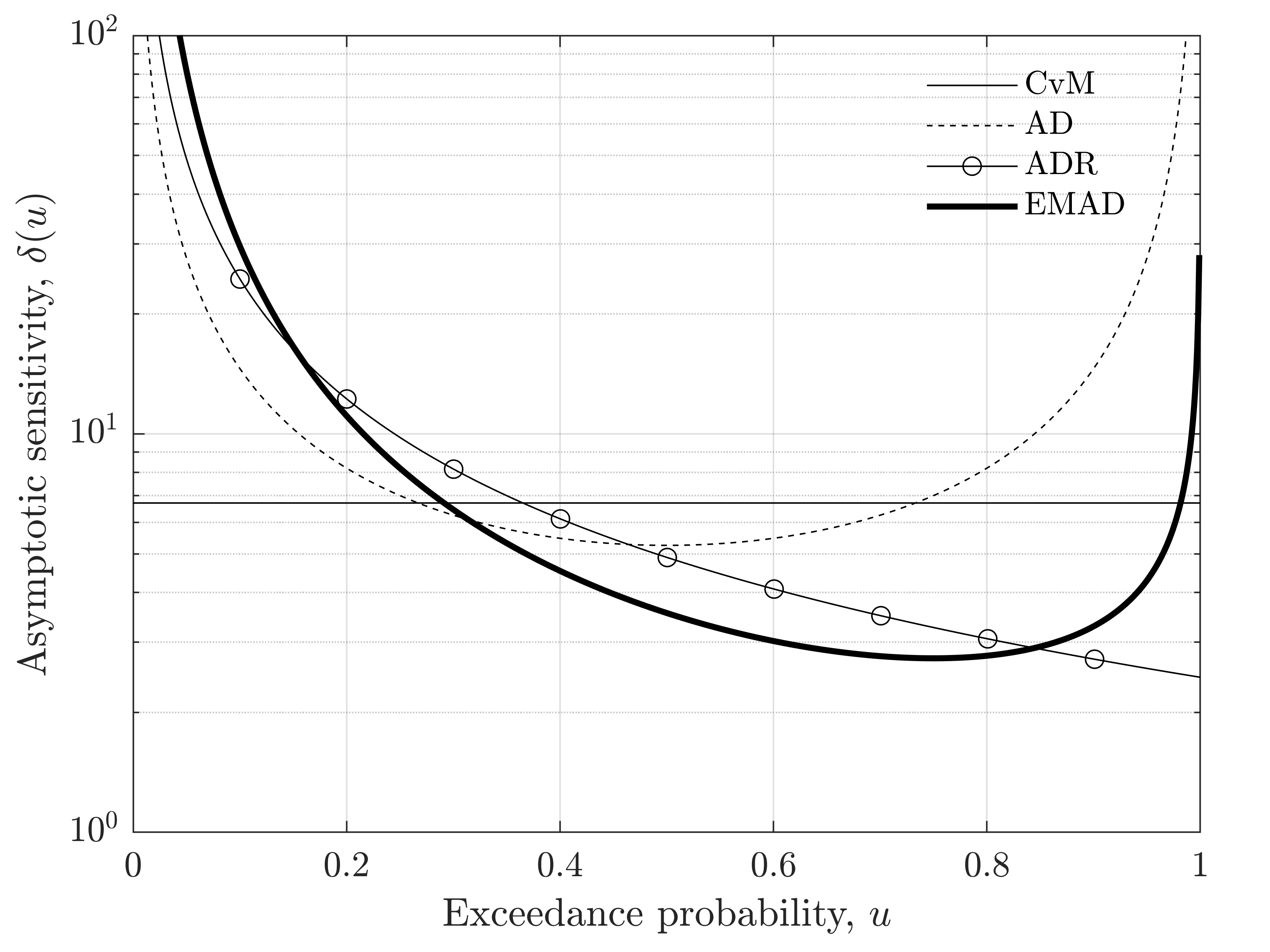}
	\caption{Asymptotic sensitivity of various goodness-of-fit statistics to perturbations in the PIT distribution at various exceedance probabilities.}
	\label{fig:test_stat_sensitivity}
\end{figure}

The choice of test statistic will influence judgement of model performance. A model for a given data sample may fail one test at a certain significance level but pass another test at the same level (see discussion in SM1.2). For the present purposes, we do not advocate the use of goodness-of-fit tests as pass/fail criteria, but use them only as a means of summarising model performance. 

%%%%%%%%%%%%%%%%%%%%%%%%%%%%%%%%%%%%%%%%%%%
\subsection{Performance summary plots} \label{sec:summary_plot}

\begin{figure}[t]
	\centering
    \includegraphics[width=0.8\textwidth]{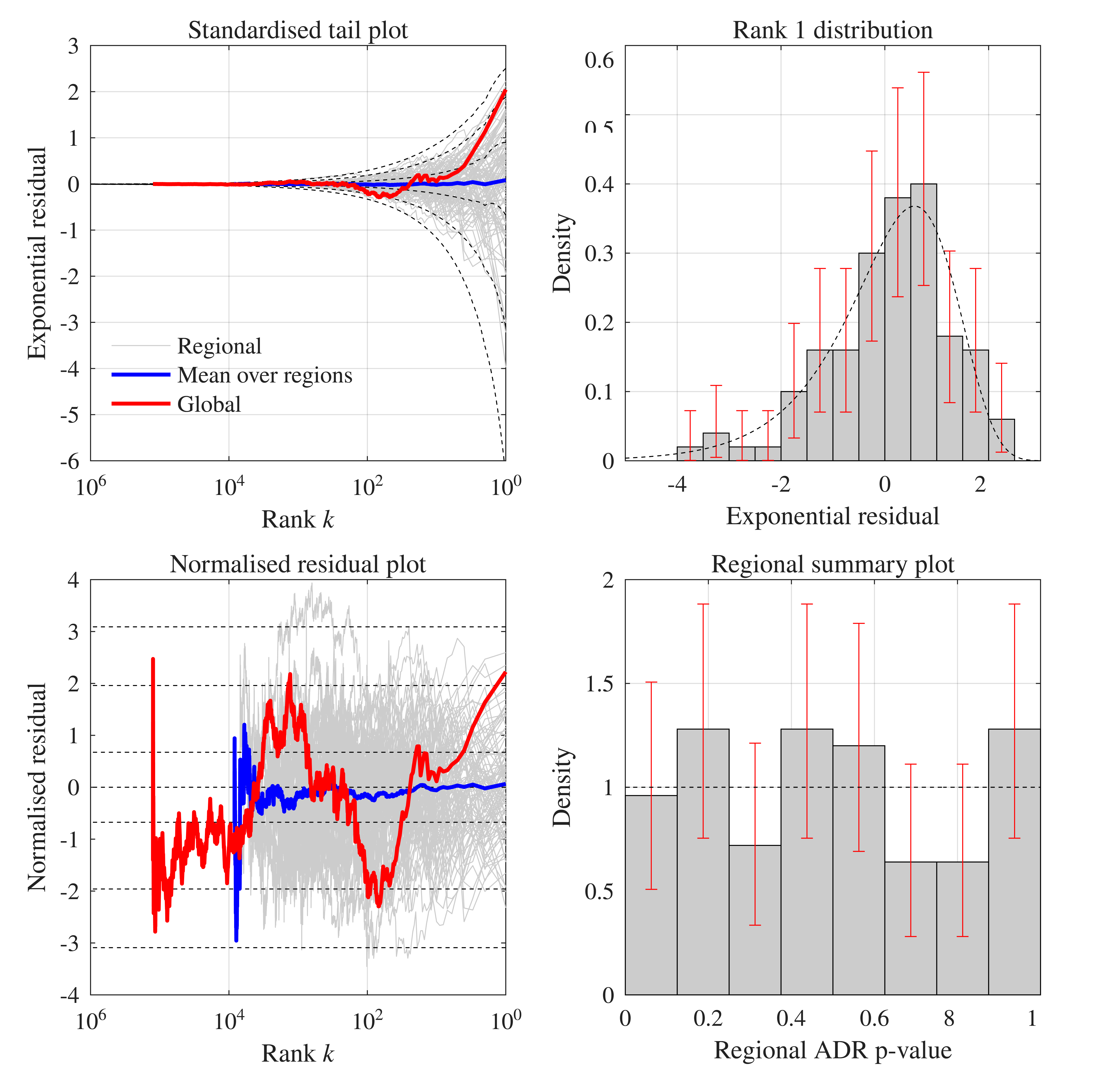}
	\caption{Example summary plots for the data used in Figures \ref{fig:exp_QQ} and \ref{fig:norm_res_plot}. Bold blue lines in standardised tail plot and normalised residual plot are the mean value at a given rank, over all regions. Red lines in these plots are the global diagnostics, obtained by pooling PIT values over all bins. Upper right plot shows a density histogram of the rank 1 exponential residuals over all regions, together with 95\% error bars. The asymptotic sampling distribution is shown as the dashed line. Lower right plot shows a density histogram of the regional ADR p-values, together with 95\% error bars.}
	\label{fig:summary_plots}
\end{figure}

\subsubsection{Regional summary plot}
We now reintroduce the regional binning of the covariate domain, for bins $b=1,\dots,B$. Suppose, for each bin, we calculate a goodness-of-fit statistic (either ADR or EMAD), and we denote the p-values of these statistics as $g_1,\dots,g_B$. Then, under the null hypothesis, these p-values are uniformly distributed in $[0,1]$. We define the \textit{regional summary plot} to be a density histogram of $g_1,\dots,g_B$. This can be used as a visual check on the uniformity of the regional goodness-of-fit p-values. A larger proportion of low p-values indicates a lack of fit overall. An example is shown in the lower right panel of \autoref{fig:summary_plots}, based on the simulated data used in \autoref{fig:exp_QQ} and \autoref{fig:norm_res_plot}. 

%%%%%%%%%%%%%%%%%%%%%%%%%%%%%%%%%%%%%%%%%%%
\subsubsection{Rank-based summaries}
We consider two rank-based summaries of model performance. Firstly, for a given rank $k$, we can calculate the mean value of either the exponential residuals ${d}_k^{(b)}\coloneq z_k^{(b)}-\hat{z}_k^{(b)}$ or normalised residuals $\nu_k^{(b)}$, over bins $b=1,\dots,B$, and add this information to the standardised tail and normalised residual plots. Examples of this are shown in \autoref{fig:summary_plots}, based on the simulated data shown in \autoref{fig:exp_QQ} and \autoref{fig:norm_res_plot}. Since, under the null hypothesis, both $d_k^{(b)}$ and $\nu_k^{(b)}$ have zero mean for any sample size, a shift in the observed mean away from zero indicates a bias in the model at a given rank. 

As well as the mean, we can investigate whether the distribution of $d_k^{(b)}$ or $\nu_k^{(b)}$ follows the expected sampling distribution. An example is shown in the upper right panel of \autoref{fig:summary_plots}, which shows a density histogram of the rank 1 exponential residuals (although we focus here on $k=1$, i.e., the most extreme observation in each bin, other ranks can also be considered). Although the theoretical sampling distribution is based on asymptotic arguments, all regions contain a minimum of 10 data points, so the asymptotic distribution is a close approximation in this case. 

%%%%%%%%%%%%%%%%%%%%%%%%%%%%%%%%%%%%%%%%%%%
\subsection{Global diagnostics} \label{sec:global_summary}
In some cases, we may wish to compare a large number of candidate models, without having to examine detailed visual diagnostics for each. In this case, it is useful to quantify the model performance in terms of a small number of summary values. Although uniformity of the global PIT values is not a sufficient condition for model consistency, it is a necessary one. Therefore, the p-values of the global ADR or EMAD statistics provide a useful summary of model performance. We can also form global diagnostic plots in the same way as described in Sections \ref{sec:standard_tail_plot} and \ref{sec:normalised_residual}, by pooling the model PIT values over the entire covariate domain, i.e., by setting the number of covariate bins to $B=1$. Examples of these global diagnostics are also shown in \autoref{fig:summary_plots}. Whilst the global diagnostics provide an indication of the overall model performance, if occurrences are concentrated in a small number of regions, then global diagnostics will tend to be dominated by quality of model fit in these regions. 

To quantify model performance over all regions, we can quantify the uniformity of the regional goodness-of-fit p-values $g_1,\ldots,g_B$ using an additional test. In this case, an unweighted test statistic, such as the CvM statistic, is appropriate, as we are not testing for specifically for deviations in the tails. We define the \textit{regional uniformity p-value} to be the p-value of the CvM statistic for $g_1,\ldots,g_B$. This can be interpreted as a quantification of the uniformity of the regional summary plot (lower right panel of \autoref{fig:summary_plots}).

%%%%%%%%%%%%%%%%%%%%%%%%%%%%%%%%
\section{Examples} \label{sec:example}
\subsection{Overview}
In this section, we consider two applications of peaks-over-threshold (POT) regression models. Section~\ref{sec:ex_SPAR} describes a multivariate extreme value model fitted to five-dimensional simulated data. In this setting, the multivariate data are expressed in polar coordinates, and a POT regression analysis is conducted for the radial component conditional on the angle. Section~\ref{sec:ex_surrogate} details a practical example of a probabilistic surrogate model for extreme responses of a floating offshore wind turbine to environmental loading, where the turbine response is modelled conditional on three environmental variables. 

Both examples use a deep learning-based inference procedure, in which artificial neural networks (ANNs) are used to represent parameter variation on the covariate domain; see \citet{Richards2024}. A brief overview of the inference procedure is provided in Appendix~\ref{app:POT}. Application of this approach requires the selection of various hyperparameters, such as the number of layers in the neural network, the number of nodes per layer, and the non-exceedance probability used to define the threshold. The approach we take to hyperparameter optimisation is to fit a large number of candidate models with different combinations of hyperparameters, and use our proposed diagnostics to select between the fitted models. 

Our diagnostics can equally be applied for other inference schemes. To illustrate this, in the second example, we compare the performance of the ANN model to a model that uses GAMs (generalised additive models) for the generalised Pareto threshold and parameter functions \citep[see, e.g.,][]{Youngman2019}.

%%%%%%%%%%%%%%%%%%%%%%%%%%%%%%%%
\subsection{SPAR model for multivariate extremes} \label{sec:ex_SPAR}
Let $\bm{X}=(X_1,\dots,X_d)\in\RR^d$ be a random vector with a multivariate normal copula with correlation matrix $\mathrm{S}\in\mathbb{R}^{d\times d}$, and standard Laplace margins. That is, for all $i=1,\dots,d$, $X_i\sim F_L(x)$ where $F_L(x) = \tfrac{1}{2} + \sgn(x) \big(1-\exp(-|x|)\big)$ is the standard Laplace distribution function. Then $\bm{X}$ has joint density function
\begin{equation}\label{eq:norm_laplace_dens}
    f_{\bm{X}}(\bm{x}) = 2^{-d} \lvert \mathrm{S} \rvert^{-1/2} \exp \left( - \|\bm{x}\|_1 - \frac{1}{2} \|\bm{z}\|^2_2 - \frac{1}{2} \bm{z}^\top \mathrm{S}^{-1} \bm{z} \right),
\end{equation}
where $\|\cdot\|_p$, $p\ge 1$, is the $L^p$ norm and $z_i = \Phi^{-1}\left(F_L(x_i)\right)$ for $i=1,\dots,d$. Define pseudo-polar coordinates $R=\|\bm{X}\|_2\in[0,\infty)$ and $\bm{W}=\bm{X}/R \in \dsphere$, where $\dsphere = \{\bm{x}\in\RR^d: \|\bm{x}\|_2 =1\}$ is the unit hypersphere in $\RR^d$. The conditional density of $R|\bm{W}=\bm{w}$ has the asymptotic form (see SM2.1)
\begin{equation} \label{eq:MVN_asymptotic}
    f_{R|\bm{W}}(r|\bm{w}) \propto [1+o(1)] r^{\zeta(\bm{w})- 1} \exp \left( - \frac{r}{\kappa(\bm{w})} \right), \quad r\to\infty, \; \bm{w} \in \dsphere.
\end{equation}
Expressions for $\zeta(\bm{w})$ and $\kappa(\bm{w})$ in terms of $\mathrm{S}$ are given in SM2.1. In this formulation, the conditional radial component converges to a truncated gamma distribution with shape $\zeta(\bm{w})$ and scale $\kappa(\bm{w})$. However, we note that $\zeta(\bm{w})$ can be negative in some cases, as discussed below.

The SPAR model \citep{mackay2025spar} assumes that the tail of the conditional radial distribution can be approximated by a generalised Pareto (GP) random variable, which has distribution function \begin{align*}
F_{\rm GP}(y|\sigma,\xi) := 
\begin{cases}
1-(1+\xi y /\sigma)^{-1/\xi}, \quad & \xi \neq 0, \\
1-\exp(-y/\sigma), \quad & \xi=0.
\end{cases}
\end{align*}
Specifically, the SPAR model assumes that, for some high quantile $u_\tau(\bm{w}):=F^{-1}_{R|\bm{W}}(\tau|\bm{w})$ with non-exceedance probability $\tau$ close to one, the conditional radial excesses $$(R-u_\tau(\bm{w}))|(\bm{W}=\bm{w})$$ follow a GP distribution with scale parameter $\sigma(\bm{w})>0$ and shape parameter $\xi(\bm{w})\in\mathbb{R}$ conditional on pseudo-angle $\bm{w}$. Since the conditional radial distribution converges to a gamma distribution in the case of Eq. \eqref{eq:MVN_asymptotic}, we make the simplifying assumption here that the tail of the radial distribution can be approximated by an exponential distribution (i.e., a GP distribution with shape parameter equal to zero, and hence assuming that $\zeta (\bm{w})=1$ everywhere). The SPAR approach also requires a model for the angular density, but here we concentrate on the conditional radial model; for further details, see \citet{mackay2026spar} and \citet{murphy2026exploring}.

In the present example, we consider $d=5$ with a randomly generated correlation matrix
\begin{equation*}
\mathrm{S} = 
\begin{bmatrix}
    \;\;\;1.0000 & -0.4387 & \;\;\;0.5946 & \;\;\;0.0758 & -0.2198\\
    -0.4387 & \;\;\;1.0000 & -0.5885 & \;\;\;0.0361 & \;\;\;0.3887\\
    \;\;\;0.5946 & -0.5885 & \;\;\;1.0000 & \;\;\;0.0778 & -0.2404\\
    \;\;\;0.0758 & \;\;\;0.0361 & \;\;\;0.0778 & \;\;\;1.0000 & -0.1047\\
    -0.2198 & \;\;\;0.3887 & -0.2404 & -0.1047 & \;\;\;1.0000
    \end{bmatrix}.
\end{equation*}
The range of tail shapes that this correlation matrix produces over the angular domain is illustrated in Figure SM4 of the Supplementary Material. Most values of $\zeta(\bm{w})$ are in the interval $[1,10]$, but some angles have $\zeta(\bm{w})<-100$ or $\zeta(\bm{w})>100$. This indicates that the accuracy of the approximation using an exponential distribution will vary across the domain, although convergence to the asymptotic form \eqref{eq:MVN_asymptotic} can be slow at some angles. The key point is that some differences between the model and observations are to be expected due to the sub-asymptotic approximation that we make with our simplified SPAR model.

One of the more challenging aspects when fitting a POT regression model is obtaining reasonable estimates in regions of sparse observations. As discussed in SM2.2, for this correlation matrix, the angular density varies by approximately two orders of magnitude over the angular domain, making inference challenging.

To infer the angle-dependent threshold and scale parameter, we use the inference scheme proposed by \citet{mackay2026spar}, where both the threshold function $u_\tau(\bm{w})$ and the subsequent exponential scale parameter $\sigma(\bm{w})$ are modelled by multilayer perceptrons. We train models (corresponding to different choices of pre-set model tuning parameters) using a sample of size $n=2\times 10^4$ for training and validation (with a standard 80/20 split), and have a hold-out test set of size $n=5\times 10^4$ for producing goodness-of-fit diagnostics. To apply our regional diagnostics, we first partition the angular domain by forming a Voronoi partition of the hypersphere, relative to a set of pseudo-regularly spaced reference angles (see Section SM3 of the Supplementary Material). One of the hyperparameters we wish to vary is the non-exceedance probability $\tau$ of the conditional quantile used to define the threshold. Varying $\tau$ results in different numbers of observations $n_b$ falling into each bin. To ensure that the number of bins used in the diagnostics does not vary with threshold level (i.e., all bins contain some observations exceeding the threshold $u_\tau(\bm{w})$), we iteratively refine the partition to ensure that all bins contain a minimum of $n_b=100$ observations, using the procedure described in Section SM3 of the Supplementary Material. This ensures that when, for example, $\tau=0.9$, there will be (approximately) 10 or more observations in each bin. The resulting partition has a total of $B=360$ bins.

We consider 5000 candidate deep SPAR models, with the architecture for each candidate model drawn randomly from a set of possible configurations. As training of deep models can be sensitive to the initialisation of their parameters, we also draw random initialisation weights for each candidate. The threshold non-exceedance probability $\tau$ is drawn randomly from a set of 36 equally-spaced candidate values: $\{0.1, 0.125, 0.15, \dots, 0.975\}$. For the neural network architectures that comprise $\sigma(\bm{w})$ and $u_\tau(\bm{w})$ (see Appendix~\ref{app:POT}), we consider multilayer perceptrons with either $L=2$ or $L=3$ layers, and with constant width $\eta$ of $2^4$, $2^5$, or $2^6$ nodes. The neural network models for $\sigma(\bm{w})$ and $u_\tau(\bm{w})$ are distinct and do not share parameters, and a different random configuration is used for each. This gives a total of $36\times(2\times3)^2 = 1296$ candidate model configurations. Therefore, our sample of 5000 candidate models will contain repeats with the same architecture but different random initialisations.

After estimating each of the candidate models, we compute global and regional visual diagnostics, and the summary statistics described in Sections~\ref{sec:standard_regional} and \ref{sec:summary}. We choose a threshold non-exceedance probability $\tau$ that provides both good global and good regional fit, which we quantify with (i) the global ADR p-value, and (ii) the regional uniformity p-value. The use of the ADR statistic for threshold selection has been discussed previously in \citet{solari2017peaks, alif2026extending}. \autoref{fig:ADR_threshold} presents box-plots of the estimated p-values, pooled by the threshold non-exceedance probability $\tau$ for each candidate model. In both plots, we observe a local maxima around $\tau = 0.825$, which we take to be our optimal threshold level for $u_\tau(\cdot)$. 

\begin{figure}[t!]
	\centering
    \includegraphics[width=\textwidth]{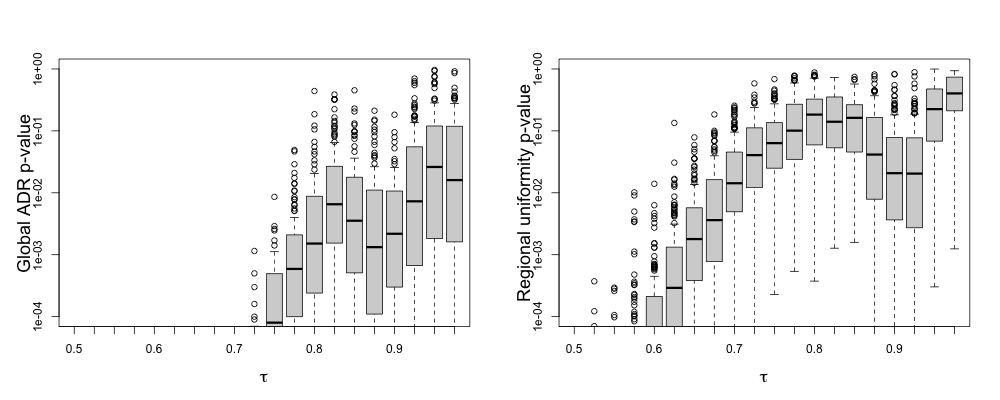}
	\caption{Box-plots of the global ADR p-values (left) and regional uniformity p-values (right), under 5000 estimated candidate models, as described in Section~\ref{sec:ex_SPAR}. The $x$-axis gives the candidate threshold non-exceedance probability~$\tau$. To enhance readability, the $y$-axis has been truncated at $1\times 10^{-4}$; all estimated $p$-values are less than $1\times 10^{-4}$ for $\tau<0.5$, and so the $x$-axis has also been truncated.}
	\label{fig:ADR_threshold}
\end{figure}

With $\tau = 0.825$ fixed, we estimate a further 2500 models with architectures randomly sampled and initialised from the 36 possible configurations, and investigate the summary statistics of their regional and global diagnostics. \autoref{fig:scatter} shows scatter plots of the p-values of three tests: (i) the global EMAD p-value, (ii) the global ADR p-value, and (iii) the regional uniformity p-value. We observe strong positive dependence between the two summary statistics for the global fits, i.e., (i) and (ii). However, we observe near-independence between global and regional summary statistics, suggesting that global goodness-of-fit diagnostics are not sufficient to diagnose goodness-of-fit at a regional or local level. Instead, practitioners should use both regional and global diagnostics in conjunction to determine goodness-of-fit for regression models. 

Of the 2500 candidate models, there are 122 models for which the global EMAD p-value exceeds 0.05, and 371 for which the global ADR p-value exceeds 0.05. This should be considered in light of the complexity of the models being estimated, and the relatively small sample used for training. In the most extremes cases, we have estimated models with approximately 8000 parameters from only 3500 training samples. Moreover, the stochastic gradient descent scheme used for parameter optimisation is not guaranteed to converge to a global optimum. Therefore, the relatively low proportion of models with `good' global diagnostics is to be expected for this `brute-force' strategy of hyperparameter optimisation.

\begin{figure}[t!]
	\centering
    \includegraphics[width=\linewidth]{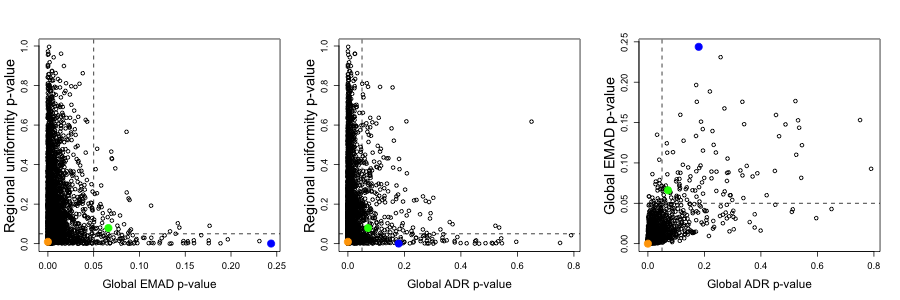}
	\caption{Scatter plots of the p-values associated with three test statistics for 2500 fitted models with non-exceedance probability $\tau=0.825$. Summary statistics include: global EMAD, ADR for the global fit, and the the CvM-based p-values for the uniformity of the regional ADR p-values. Dashed vertical and horizontal lines correspond to a p-value of 0.05. The three models that correspond with coloured points (Model 1 - blue, Model 2 - orange, Model 3 - green) are further investigated.}
	\label{fig:scatter}
\end{figure}

We now take a closer look at the visual diagnostics for individual estimates. Specifically, we focus on the models which provide the best (Model 1 - blue point on \autoref{fig:scatter}) and worst (Model 2 - orange) global fits with respect to the EMAD test, and a model (Model 3 - green) that exhibits good fit both regionally and globally, i.e., passes the global ADR and EMAD tests, and the regional uniformity test, at a 5\% significance level.

\autoref{fig:app1_example} provides visual diagnostics for Model 3. These plots provide further evidence in support of an excellent model fit at both the regional and global level. From the standardised tail (top centre), and normalised residual (bottom centre) plots, we observe good global fits (red lines) contained within 95\% confidence envelopes; the standard tail and normalised residual plots also illustrate good regional fits (grey lines), with regional estimates (grey lines) similarly contained within the 95\% confidence envelopes. The standardised tail and normalised residual plots also include a running empirical mean for each rank (blue), which appears close to the theoretical value (zero) in both cases. Goodness-of-fit for Model 3 is further exemplified at the regional level by the right column of \autoref{fig:app1_example}, which illustrates that, under Model 3, the rank 1 and rank 5 exponential residuals closely follow the theoretical asymptotic log-gamma sampling distribution (other ranks could also be considered, but are omitted here for brevity). Finally, the bottom left panel of \autoref{fig:app1_example} provides a histogram of the regional ADR p-values, which illustrate the expected uniformity (and which passes the CvM uniformity test at a 5\% significance level). 
\begin{figure}[t!]
	\centering
    \begin{minipage}{0.32\linewidth}
    \includegraphics[width=\linewidth]{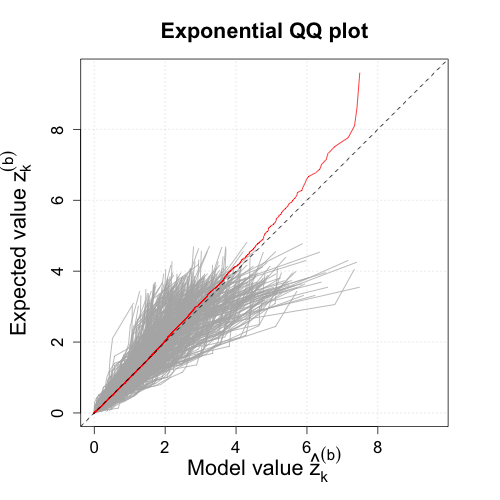}
    \end{minipage}
    \begin{minipage}{0.32\linewidth}
    \includegraphics[width=\linewidth]{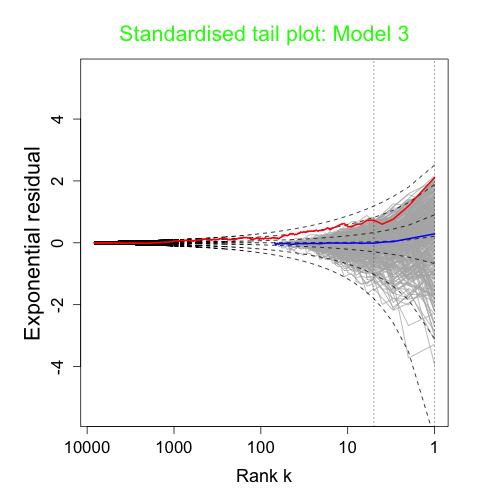}
    \end{minipage}
    \begin{minipage}{0.32\linewidth}
    \includegraphics[width=\linewidth]{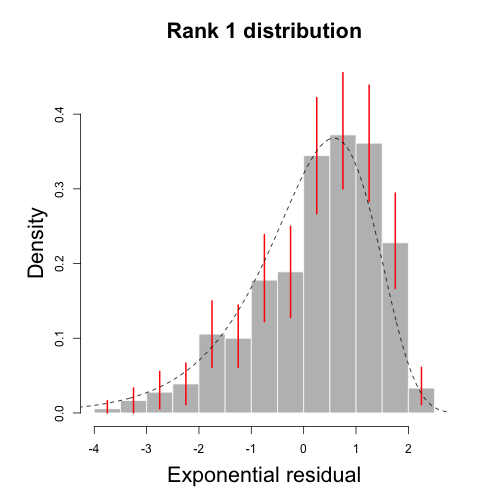}
    \end{minipage}
    \begin{minipage}{0.32\linewidth}
    \includegraphics[width=\linewidth]{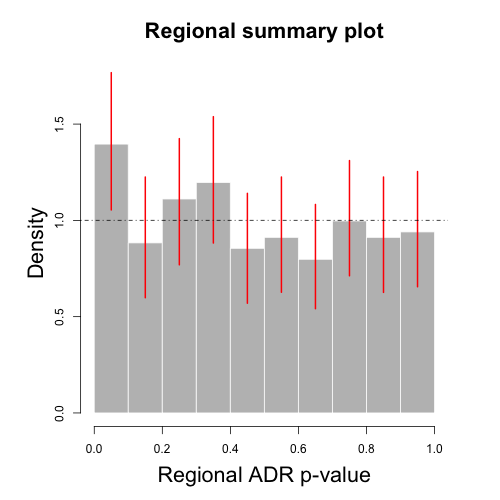}
    \end{minipage}
    \begin{minipage}{0.32\linewidth}
    \includegraphics[width=\linewidth]{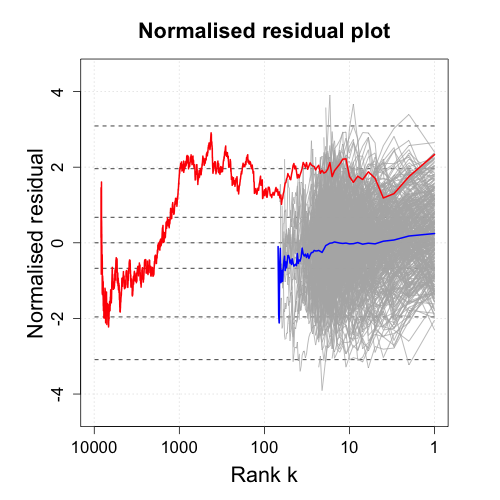}
    \end{minipage}
    \begin{minipage}{0.32\linewidth}
    \includegraphics[width=\linewidth]{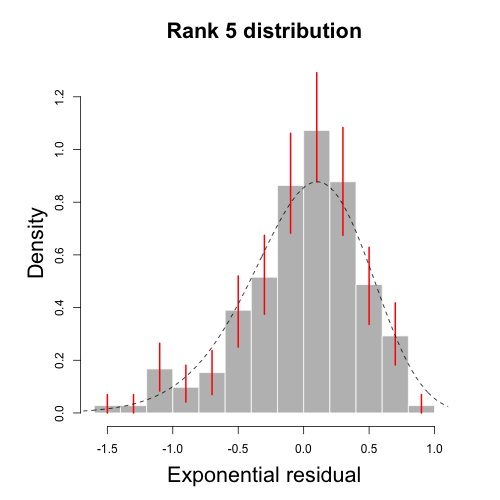}
    \end{minipage}
	\caption{Diagnostics for SPAR Model 3.  Top left and centre: Exponential QQ plot and Standardised tail plot, with grey and red curves denoting regional and global estimates, respectively. Top centre: Dashed curves denote  0.001, 0.025, 0.25, 0.5, 0.75, 0.975, and 0.999 quantiles from the theoretical log-gamma distribution (as a function of rank $k$). Right column: histogram of exponential residuals for ranks $k=1$ (top) and $k=5$ (bottom), with the theoretical density function (dashed line). Bottom left:  Histogram of the regional ADR p-values used in the CvM uniformity test. Bottom centre: Normalised residual for all data (red) and regional samples (grey), with dashed lines indicating normal quantiles at exceedance probabilities 0.001, 0.025, 0.25, 0.75, 0.975, and 0.999.
    Centre column: Blue lines denote running empirical means of regional estimates (in grey). For all histograms, red lines denote 95\% bootstrapped error bars. }
	\label{fig:app1_example}
\end{figure}

In contrast, deficiencies in the fits for Models 1 and 2 can be identified immediately from their own visual diagnostics. In the Supplementary Material, we provide analogues of \autoref{fig:app1_example} for Models 1 and 2; here, in \autoref{fig:app1_compare}, we focus only on a subset of the diagnostic plots. For Model 1, the standardised tail plot illustrates good fit at the global level, with the red curve contained within the 95\% confidence envelope. Good global fit of Model 1 is further supported by the large global ADR and EMAD p-value estimates (see \autoref{fig:scatter}). However, diagnosis of the fit at the regional level is more nuanced. The rank one exponential residuals appear to follow the theoretical asymptotic sampling distribution (top right panel; \autoref{fig:app1_compare}), but there is marked deviation away from uniformity for the regional ADR p-values (with estimated p-value of 0.0005). In particular, there is a significant excess in the number of small p-values, suggesting that the ADR goodness-of-fit test has failed for many regions. In this case, poor regional fits of the deep SPAR model are masked by the exemplary global fit.

Poor regional and global goodness-of-fit for Model 2 is immediately obvious from the visual diagnostics in \autoref{fig:app1_compare}. The global standardised tail plot markedly deviates from the horizontal, and we observe positive bias in the empirical standardised differences at rank $k=1$. While visual inspection of the histogram of regional ADR p-values (bottom left; Figure\ref{fig:app1_compare}) does suggest uniformity here, the p-value for the uniformity test is approximately zero (\autoref{fig:scatter}). This further supports the joint use of visual and quantitative tools for model assessment.
\begin{figure}[t!]
	\centering
        \begin{minipage}{0.3\linewidth}
    \includegraphics[width=\linewidth]{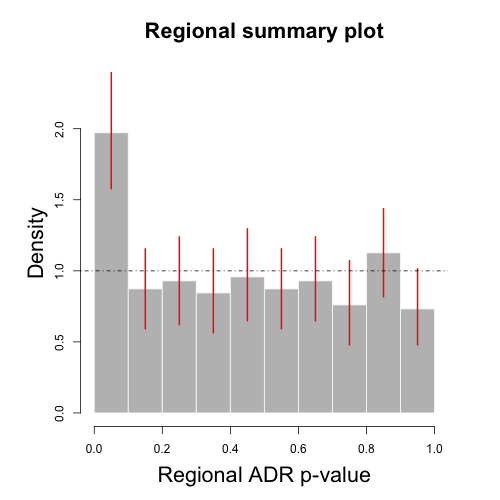}
    \end{minipage}
    \begin{minipage}{0.3\linewidth}
    \includegraphics[width=\linewidth]{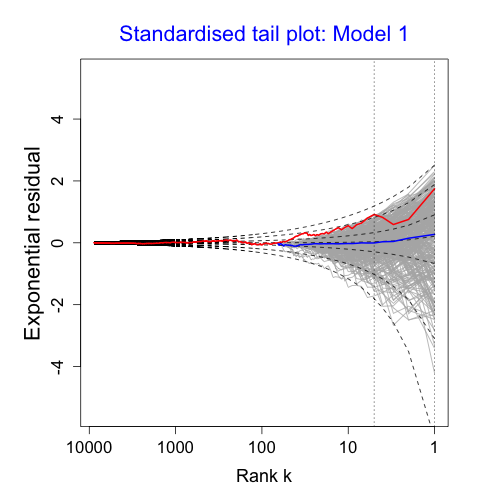}
    \end{minipage}
    \begin{minipage}{0.3\linewidth}
    \includegraphics[width=\linewidth]{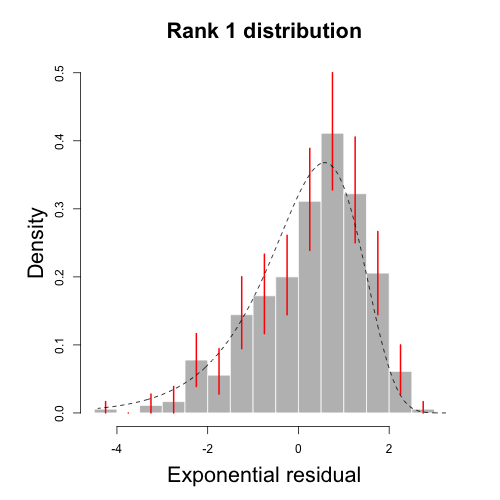}
    \end{minipage}
          \begin{minipage}{0.3\linewidth}
    \includegraphics[width=\linewidth]{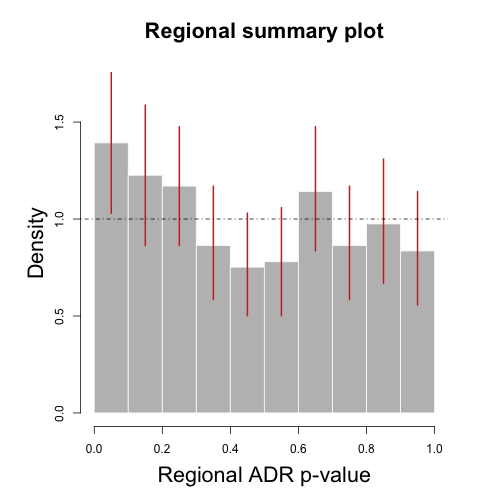}
    \end{minipage}
    \begin{minipage}{0.3\linewidth}
    \includegraphics[width=\linewidth]{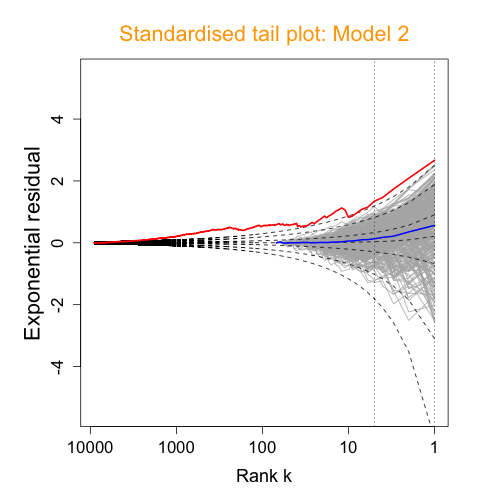}
    \end{minipage}
    \begin{minipage}{0.3\linewidth}
    \includegraphics[width=\linewidth]{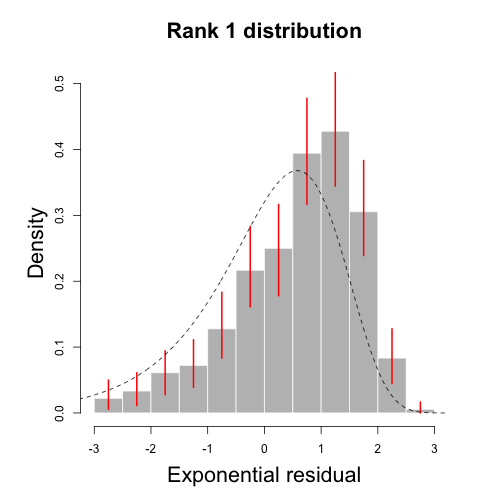}
    \end{minipage}
	\caption{Diagnostics for SPAR Models 1 (top) and 2 (bottom). Left: Histogram of the regional ADR p-values used in the CvM uniformity test. Centre: Standardised tail plot, with grey and red curves denoting regional and global estimates, respectively; Dashed curves denote  0.001, 0.025, 0.25, 0.5, 0.75, 0.975, and 0.999 quantiles from the theoretical log-gamma distribution (as a function of rank $k$). Right: histogram of empirical differences $D_k$ for $k=1$, with the theoretical density function (dashed line). For all histograms, red lines denote 95\% bootstrapped error bars. }
	\label{fig:app1_compare}
\end{figure}

This example illustrates the use of our proposed visual diagnostics in applications of extreme value regression models. Moreover, our regional and global summary statistics provide a simple and efficient approach to comparing fits across a large set of candidate models, where manual assessment is infeasible.

%%%%%%%%%%%%%%%%%%%%%%%%%%%%%%%%
\subsection{Surrogate model for extreme response of a wind turbine} \label{sec:ex_surrogate}
Here we consider the distribution of extreme tension in a mooring line for a floating offshore wind turbine, conditional on environmental condition. For a given environmental condition, described in terms of various wind, wave, and current variables, the responses of the system are stochastic, due to the inherent stochastic variability in the turbulent winds and random irregular waves. In the present example, we consider the distribution of peak mooring line tension $Y$ conditional on a vector of three environmental conditions $\bm{X}=(X_1,X_2,X_3)\in\mathcal{X}$, where $X_1$ is 1-hour mean wind speed, $X_2$ is significant wave height, and $X_3$ is peak wave period, and $\mathcal{X}$ is a bounded subset of $\mathbb{R}^3$. For model training, numerical simulations of the turbine dynamics over a duration of 1 hour were conducted for 851 combinations of $(X_1,X_2,X_3)$ on a regular grid, as described in \citet{mackay2026peak}. For each simulation, peaks in the time series of mooring line tension were defined as local maxima within a moving window of $\pm10$~s.
This yielded between 85 and 223 peaks per simulation, depending on environmental condition, and a total of 117388 peaks for training. For model testing, a further 219 response simulations were conducted at random uniformly distributed values of $(X_1,X_2,X_3)$ within the range used for model training. This yielded a total of 29067 peaks for testing. Unlike Example~1 in Section~\ref{sec:ex_SPAR}, this dataset permits a local (as opposed to regional) binning strategy for model assessment, as we observe multiple realisations of $Y|\bm{X}=\bm{x}$ for a single $\bm{x}$. We thus proceed by assigning each test covariate $\bm{x}_i$ to its own bin $\mathcal{B}_i$, such that we have $B=219$ bins. 

As in Section~\ref{sec:ex_SPAR}, we consider a peaks-over-threshold analysis using a GP regression model fitted to conditional excesses of $Y$ above some high non-stationary threshold $u_\tau(\bm{x}),$ which is the conditional $\tau$-quantile of $Y|\bm{X}=\bm{x}$. We consider two representations of the GP threshold and parameter functions, one using ANNs (see \autoref{app:POT}) and the second using GAMs. As in Section~\ref{sec:ex_SPAR}, we consider a large number of candidate ANN models (2000 in this example), with different architectures and random parameter initialisation, but all with threshold non-exceedance probability $\tau=0.8$. Candidate architectures are generated at random, separately for $u_\tau(\bm{x})$ and $(\sigma(\bm{x}),\xi(\bm{x}))$. We again consider $L=2$ or $L=3$ layers, but here, to reflect the simpler covariate domain, we also consider smaller widths of $\eta$ of $2^2$ and $2^3$. 

For each model, we estimate the p-values for the global test statistics - the EMAD and ADR - and provide their estimates in \autoref{fig:app2_scatter}. Previous applications of diagnostics for extreme value regression models have focused on these global goodness-of-fits, and so we hereafter focus on goodness-of-fit for three models which provide high p-values for the global ADR test and global EMAD test; see \autoref{fig:app2_scatter}. As with the previous example, only a small proportion (approx. 1\%) of the candidate models yield global EMAD and ADR p-values in excess of 0.05, which is to be expected given the random sampling of architectures and initialisations.

\begin{figure}[ht!]
	\centering
    \includegraphics[width=0.5\linewidth]{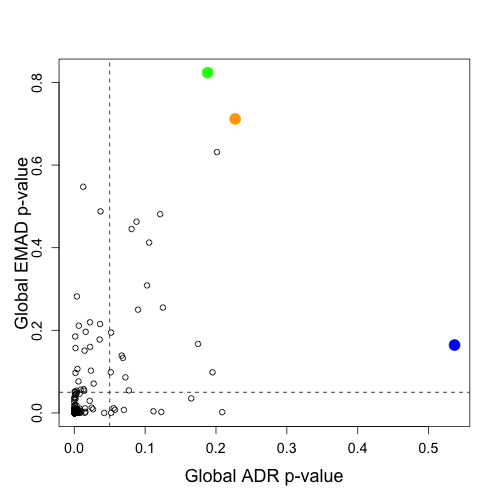}
	\caption{Scatter plot of the global EMAD and global ADR p-values for 2000 estimated deep GP regression models with non-exceedance probability $\tau=0.8$. Dashed vertical and horizontal lines correspond to a p-value of 0.05. The three models that correspond with coloured points (Model 1 - blue, Model 2 - orange, Model 3 - green) are further investigated.}
	\label{fig:app2_scatter}
\end{figure}

\begin{figure}[ht!]
	\centering
        \begin{minipage}{0.45\linewidth}
    \includegraphics[width=\linewidth]{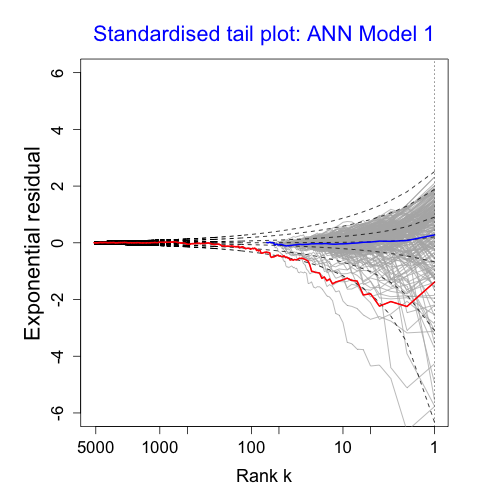}
    \end{minipage}
            \begin{minipage}{0.45\linewidth}
    \includegraphics[width=\linewidth]{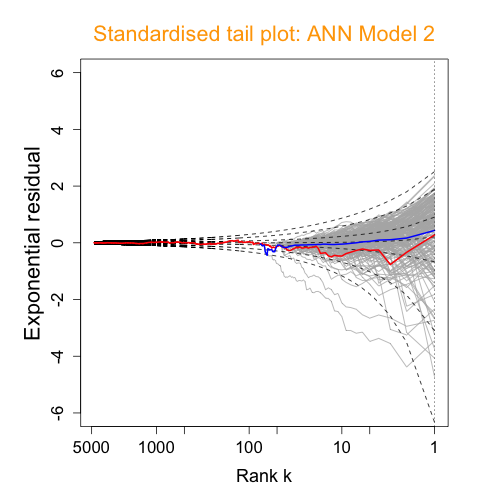}
    \end{minipage}
    \begin{minipage}{0.45\linewidth}
    \includegraphics[width=\linewidth]{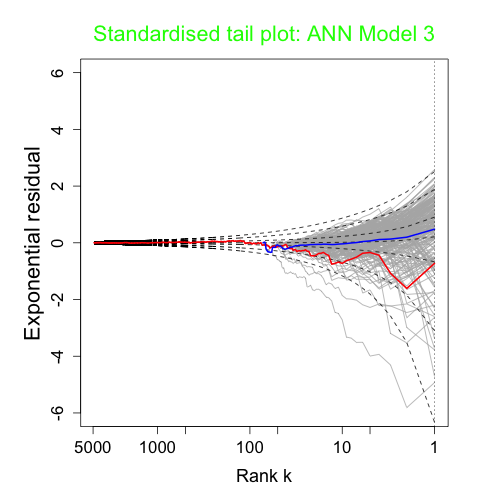}
    \end{minipage}
    \begin{minipage}{0.45\linewidth}
    \includegraphics[width=\linewidth]{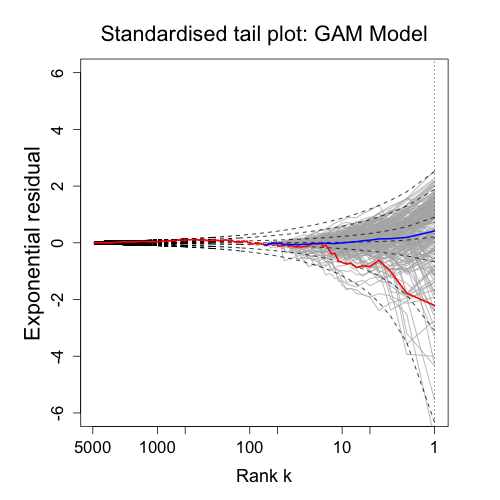}
    \end{minipage}
    
	\caption{Standardised tail plots for deep GP regression (top left, top right, and bottom left panels) and simplified GP-GAM (bottom right) model estimates. The colours of the title text in panels corresponds with the coloured p-values in \autoref{fig:app2_scatter}.  Grey and red curves denote regional and global estimates, respectively. Blue line is running mean over bins. Dashed curves denote 0.001, 0.025, 0.25, 0.5, 0.75, 0.975, and 0.999 quantiles from the theoretical log-gamma distribution (as a function of rank $k$).}
	\label{fig:app2_compare}
\end{figure}

\autoref{fig:app2_compare} shows standardised tail plots for the three models highlighted in \autoref{fig:app2_scatter}. Despite the large estimated p-values for the global ADR and EMAD tests, suggesting good global fits for all three models, we observe outliers in the the standardised tail plots. This suggests that all of the models suffer from poor fits in some regions of the covariate domain, which are not identifiable using global diagnostics. Further investigation revealed that, across the three chosen deep GP regression model fits, the outliers in the local estimates on the standardised tail plot (top of \autoref{fig:app2_compare}) correspond to observations at locations close to the boundary of the covariate domain. This suggests that the deep GP regression models may be experiencing overfitting, and motivates a follow-up analysis which considers a simpler formulation of the POT regression model. 

In a second analysis, we consider a simplified additive GP regression model with: i) the threshold $u_\tau(\cdot)$ and GP parameters $\sigma(\bm{x})$ and $\xi(\bm{x})$ modelled using generalised additive models (GAMs); ii) peak wave period $X_3$ not included in the model; and iii) covariate effects of $X_1$ and $X_2$ modelled using additive smoothing splines with only four knots. For this much simpler model, we estimate the quantile GAM using the qgam \texttt{R} package \citep{fasiolo2021qgam} and the GP-GAM  using the evgam \texttt{R} package \citep{youngman2022evgam}. The lower right plot in \autoref{fig:app2_compare} show the standardised tail plot for this simpler extreme value regression model. We immediately observe better local fits than the alternative deep model. Moreover, the simpler regression model retains a good global fit, as evidenced by the expected behaviour of the global estimate of the standardised tail plot.

This simple example provides a clear illustration of actionable insights that can be derived from our proposed extreme value regression diagnostics. By identifying regions or locations of the covariate space which exhibit poor fit, we can make pragmatic decisions about model tuning that improve the overall fit.

%%%%%%%%%%%%%%%%%%%%%%%%%%%%%%%%
\section{Conclusions} \label{sec:conclusion}
This paper has presented new visual tools and statistical tests for diagnosing goodness-of-fit for extreme value regression models. Using asymptotic theory for normalised exceedance probabilities and centred exponential order statistics, we developed two new visual tools, the standardised tail and normalised residual plots, that facilitate visualisation of goodness-of-fit consistently across multiple regions of the covariate space. We have further proposed the EMAD test statistic as the mean absolute deviation of exponential order statistics, and provided empirical and theoretical justification for its use quantifying goodness-of-fit for extreme value models. Alongside suggestions for summarising global and regional model performance, we have illustrated a practical workflow for assessing model fit in two examples using contemporary extreme value regression models. Our accessible tools can be used for quick and convenient assessment of goodness-of-fit for extreme value regression models, which facilitates model comparison (even with thousands of candidate models) and steers model design. 

The tools we have developed do not make assumptions about the form of the fitted distributional model nor its inference scheme, and are applicable in a wide range of settings. While we focus here on extreme value regression models, where deviance of extreme order statistics is of most interest, our diagnostics can equally be used in applications of traditional statistical regression models. For example, our normalised residual plot highlights poor fit in the bulk of the distribution, which is of general interest to practitioners.

A key assumption of our workflow is that diagnostics are evaluated on a hold-out sample independent of the data used for model fitting. In our examples, this assumption holds by design. However, for applications where strong dependence persists in data, like spatial environmental applications, bespoke algorithms may be required to produce independent hold-out sets \citep[see discussions by, e.g.,][]{roberts2017cross}. Future work may also consider bootstrap methods as an alternative way of producing calibrated confidence envelopes.

Regional goodness-of-fit assessment is contingent on a partition of the covariate domain. In our applications, the choice of partition comes quite naturally. In Section~\ref{sec:ex_SPAR}, since the covariates lie in a hypersphere, we considered regions described by a pseudo-Voronoi partition which assigns a minimum number of samples to bins; in Section~\ref{sec:ex_surrogate}, we considered individual locations of the covariate space, as our data included repeated sampling for a single covariate value. The partitioning should be specific to the application and the covariate domain, and it should be designed to ensure that a sufficient number of samples are available in each bin for reliable production of regional diagnostics. Sensitivity to the regionalisation could be assessed by repeated evaluation of the diagnostics for partitions of different sizes. Choice of the optimal partition, and the definition of optimal in this setting, can be considered in future work.

\section*{Acknowledgements}
The authors thank members of the Glasgow-Edinburgh Extremes Network (GLE$^2$N; \url{glen-scotland.github.io/glen/}) for helpful feedback. This work has made use of the resources provided by the Edinburgh Compute and Data Facility (ECDF) (\url{www.ecdf.ed.ac.uk/})

\section*{Statements and declarations}
EM was funded by the EPSRC Supergen Offshore Renewable Energy Hub, United Kingdom [grant no: EP/Y016297/1]. The authors have no competing interests to declare that are relevant to the content of this article.

\section*{Data availability}
R and MATLAB code for producing the diagnostics described in this paper is available at \url{https://github.com/edmackay/Diagnostics-for-extreme-value-regression}.

%%%%%%%%%%%%%%%%%%%%%%%%%%%%%%%%
\appendix
\section{Proofs} \label{app:proof}
\subsection{Proof of Theorem \ref{thm:normalised_exceedance}}
\label{app:proof_thm:normalised_exceedance}
This is a special case of Theorem 2.2.1 in \citet{leadbetter1983}. However, for completeness, we provide here an explicit derivation of the result. 
\begin{proof}
The density function of $Q_{(k)}$ is
\begin{equation*}
   f_{Q_{(k)}}(q)=\frac{\Gamma(n+1)}{\Gamma(n-k+1)\Gamma(k)}q^{k-1}(1-q)^{n-k}, \quad q\in[0,1], \,k=1,\dots,n,
\end{equation*}
where $\Gamma$ is the gamma function. Define normalised exceedance probabilities $A_{(k)} \coloneq n Q_{(k)}\in[0,n]$. Then the density function of $A_{(k)}$ is 
\begin{align*}
    f_{A_{(k)}}(a) = \frac{1}{n} f_{Q_{(k)}}\left(\frac{a}{n}\right) = g_n(k) \frac{a^{k-1}}{\Gamma(k)} \left(1-\frac{a}{n}\right)^{n-k},
\end{align*}
where $g_n(k)=n^{-k} \, \Gamma(n+1)\, /\, \Gamma(n-k+1)$. First, consider the behaviour of $(1-a/n)^{n-k}$ for fixed $a,k\in \mathbb{R},$ and $n\to \infty$. We have
\begin{align*}
    \left(1-\frac{a}{n}\right)^{n-k} 
    &= \exp \left((n-k) \log \left(1-\frac{a}{n}\right)\right) 
    = \exp \left((n-k) \left(-\frac{a}{n} + 	O\left(\frac{1}{n^2}\right)\right)\right) \\
    &= \exp \left( -a + O\left(\frac{1}{n}\right)\right)
    \to \exp(-a), \quad \text{as} \quad n\to\infty.
\end{align*}
For the limit of $g_n(k)$, we use Stirling's formula, $\Gamma(n+1) =n!= \sqrt{2\pi n} (n/e)^n (1+O(1/n))$, to obtain 
\begin{align*}
    g_n(k) = \frac{1+O(\frac{1}{n})}{1+O(\frac{1}{n-k})} \sqrt{\frac{n}{n-k}} \exp (h_n(k)),
\end{align*}
where $h_n(k) = (n-k)\log(n) - (n-k)\log(n-k) - k$. The first two terms in the expression above have unit limit for fixed $k$ and $n\to\infty$. It remains to show that $h_n(k)\to 0 $ for fixed $k$ and $n\to\infty$. For fixed $k\in\mathbb{R}$ we have
\begin{align*}
    h_n(k) &= (n-k)\log(n) - (n-k)\log\left[n \left(1-\frac{k}{n}\right)\right] - k\\
    &= - (n-k)\left[ -\frac{k}{n} + O\left(\frac{1}{n^2}\right)\right] - k = O\left(\frac{1}{n}\right).
\end{align*}
Hence $\exp(h_n(k))\to 1$ as $n\to\infty$, and therefore $g_n(k)\to 1$ as $n\to\infty$ also. So, for all $a\ge0$ and $k\in \mathbb{N}_{>0}$, we have
\begin{equation} \label{eq:limitdist}
    \lim_{n\to\infty} f_{A_{(k)}}(a) = \frac{a^{k-1}}{\Gamma(k)} \exp (-a).
\end{equation}
Hence, by Scheff{\'e}'s lemma, $A_{(k)}$ converges in distribution to a random variable with density \eqref{eq:limitdist}. This is a gamma density function with shape $k$ and unit scale. 
\end{proof}
%%%%%%%%%%%%%%%%%%%%%%%%%%%%%%%%%%%%%%%%%%%%%%%%%%%%%%%%%%%%%%%%%%
\subsection{Proof of Corollary \ref{cor:Z_difference}}
\label{app:proof_cor:Z_difference}
\begin{proof}
From Theorem \ref{thm:normalised_exceedance} we have $nQ_{(k)} \xrightarrow{\;d\;} A_{k,\infty} \sim \mbox{Gamma}(1, k)$, as $n\to\infty$. By definition, we have $Z_{(k)} = -\log\left(Q_{(k)}\right) = \log(n) -\log\left(n Q_{(k)}\right)$. Therefore, from Equations~\eqref{eq:EZk} and \eqref{eq:Hn_approx}, we have 
\begin{align*}
    D_k &= \mathbb{E}\left[Z_{(k)}\right] - Z_{(k)} = H_n-H_{k-1}  -\log(n) +\log\left(n Q_{(k)}\right) \\
    &= \log\left(n Q_{(k)}\right) - H_{k-1} + \gamma + O(1/n).
\end{align*}
Hence, by the continuous mapping theorem, 
\begin{equation*}
    D_k \xrightarrow{\;d\;} \log\left(A_{k,\infty}\right) - H_{k-1} + \gamma,\qquad n \rightarrow \infty.
\end{equation*}
The required density follows by applying the transformation to the density of $A_{k,\infty}$. 
\end{proof}
%%%%%%%%%%%%%%%%%%%%%%%%%%%%%%%%%%%%%%%%%%%%%%%%%%%%%%%%%%%%%%%%%%
\subsection{Preliminaries for results in Section~\ref{sec:summary}}
The proofs of Propositions \ref{prop:EMAD_limit}--\ref{prop:EMAD_sensitivity} require a number of preliminary results:

\begin{lemma} \label{lem:abs_mom}
Let $X$ and $Y$ be random variables with means $\mu_X$, $\mu_Y$ and variances $\sigma_X^2>0$, $\sigma_Y^2>0$, respectively. Then
\begin{enumerate}[label=(\alph*)]
   \item $\E\left[|X-\mu_X|\right] \le \sigma_X$,
   \item $|\Cov(|X|,|Y|)| \le \sigma_X\sigma_Y$.
\end{enumerate}
\end{lemma}
\begin{proof}
The first part follows immediately from Jensen's inequality. For the second part, the Cauchy–Schwarz inequality gives $|\Cov(X,Y)|^2 \le \sigma_X^2 \sigma_Y^2$. The result follows from noting that for any random variable $X$ with finite second moment $\Var(|X|)\le \Var(X)$.
\end{proof}

%%%%%%%%%%%%%%%%%%%%%%%%%%%%%%%%%%%%%%
\begin{lemma}[Absolute moments for normal distributions] \label{lem:norm_abs_mom}
Let $\Phi$ be the standard normal distribution function. For a bivariate normal pair $(X,Y)$ with means $\mu_X,\mu_Y\in\RR$, variances $\sigma_X^2, \sigma_Y^2>0$, and correlation $\rho\in(-1,1)$, we have:
\begin{enumerate}[label=(\alph*)]
    \item $\E[|X|]=\sigma_X \sqrt{\dfrac{2}{\pi}} \exp\left(-\dfrac{\mu_X^2}{2\sigma_X^2}\right) + \mu_X \left(1 - 2 \Phi\left(-\dfrac{\mu_X}{\sigma_X}\right)\right)$, 
    \item $\Cov(|X-\mu_X|,|Y-\mu_Y|) = \tfrac{2}{\pi} \sigma_X \sigma_Y \Bigl( \rho \arcsin \rho +\sqrt{1-\rho^2} - 1\Bigr)$.
\end{enumerate}
\end{lemma}

\begin{proof}
$|X|$ has a folded normal distribution. $\E[|X|]$ is given in \citet{leone1961folded}. The second result can be found in \citet{nabeya1951}.
\end{proof}

%%%%%%%%%%%%%%%%%%%%%%%%%%%%%%%%%%%%%%
\begin{lemma}[Variance of exponential order statistics] \label{lem:exp_var}
Let $Z_1 \ge \cdots \ge Z_n$ be the order statistics of $n$ independent $\mathrm{Exp}(1)$ random variables. Then, for $k=1\ldots,n$,
\begin{equation*}
    \Var(Z_k) = \sum_{i=k}^{n} \frac{1}{i^2} \le \frac{2}{k}.
\end{equation*}
\end{lemma}
\begin{proof}
Exponential order statistics can be expressed as a weighted sum of iid $\mathrm{Exp}(1)$ random variables $E_i$, $i=1,\dots,d$, as $Z_k \overset{d}{=} \sum_{i=k}^n (E_i/i)$ \citep{renyi1953}. The variance follows immediately from the independence of the summands. For the inequality, we apply the integral test to obtain:
\begin{equation*}
    \Var(Z_k) = \sum_{i=k}^{n}\frac{1}{i^2} <\sum_{i=k}^{\infty}\frac{1}{i^2} \le \frac{1}{k^2} + \int_{k}^{\infty}\frac{1}{x^2}\,\mathrm{d}x = \frac{1}{k^2} + \frac{1}{k} \le \frac{2}{k}.
\end{equation*}
\end{proof}

%%%%%%%%%%%%%%%%%%%%%%%%%%%%%%%%
The following result is due to \citet{bahadur1966} and \citet{kiefer1967} -- see, e.g., \citet[][\S10.2]{David2003}.

\begin{theorem}[Bahadur-Kiefer representation of central order statistics]\label{thm:Bahadur}
Consider iid $X_1, \ldots, X_n$ with distribution $F$ and density $f$, and let $X_{(k)}$ be the $k$-th (descending) order statistic. Let $0<p<1$, and assume $k=\lfloor np\rfloor$ and $0 < f(\xi_{p}) < \infty$, where $\xi_{p}$ is the population $(1-p)$-quantile of $F$ at exceedance probability $p$. Then 
\begin{equation*}
X_{(k)} - \xi_p = -\frac{\mathbb{F}_n(\xi_p)- p}{f(\xi_p)} + R_n,	
\end{equation*}
where $\mathbb{F}_n$ is the empirical distribution function and $R_n=O_p(n^{-3/4} (\log n)^{1/2}(\log \log n)^{1/2})$ as $n\to\infty$.
\end{theorem}

%%%%%%%%%%%%%%%%%%%%%%%%%%%%%%%%%%%%%%%%%%%%%%%%%%%%%%%%%%%%%%%%%%
\subsection{Proof of Proposition \ref{prop:EMAD_limit}}
\label{app:proof_prop:EMAD_limit}
\begin{proof}
Under the null hypothesis, the EMAD test statistic is 
\begin{equation*}
    S_n = \frac{1}{\sqrt{n}} \sum_{k=1}^{n} \left\lvert Z_{(k)} - z_k\right\rvert,
\end{equation*}
where $Z_{(1)} \ge \cdots \ge Z_{(n)}$ are the order statistics of $n$ independent $\mbox{Exp}(1)$ random variables, and $z_k=\E[Z_{(k)}]$. To establish the asymptotic distribution, we split $S_n$ into terms involving the lower and upper tails, and the central terms. Let $0<\varepsilon\ll 1$, and write 
$S_n =S_n^{low}(\varepsilon) + S_n^{cent}(\varepsilon) + S_n^{high}(\varepsilon)$, where
\begin{align*}
	S_n^{low}(\varepsilon) &= \frac{1}{\sqrt{n}} \sum_{k=1}^{\lfloor \varepsilon n \rfloor} \left\lvert Z_{(k)} - z_k \right\rvert,\\
	S_n^{cent}(\varepsilon) &= \frac{1}{\sqrt{n}} \sum_{k=\lfloor \varepsilon n \rfloor + 1}^{\lfloor (1-\varepsilon) n \rfloor} \left\lvert Z_{(k)} - z_k \right\rvert,\\
	S_n^{high}(\varepsilon) &= \frac{1}{\sqrt{n}} \sum_{k=\lfloor (1-\varepsilon) n \rfloor + 1}^{n} \left\lvert Z_{(k)} - z_k \right\rvert.
\end{align*}
We will show that $\E[S_n^{low}(\varepsilon)]$ and $\E[S_n^{high}(\varepsilon)]$ are $O(\sqrt{\varepsilon})$ as $n\to\infty$. Hence, by Markov's inequality, $S_n^{low}(\varepsilon)$ and $S_n^{high}(\varepsilon)$ converge in probability to zero as $\varepsilon \to 0$.
Using Lemmas~\ref{lem:abs_mom}(a) and \ref{lem:exp_var}, we have
\begin{align*}
    \E[S_n^{low}(\varepsilon)] &= \frac{1}{\sqrt{n}} \sum_{k=1}^{\lfloor \varepsilon n \rfloor} \E\left[\left\lvert Z_{(k)} - z_k \right\rvert\right] \le \frac{1}{\sqrt{n}} \sum_{k=1}^{\lfloor \varepsilon n \rfloor} \sqrt{\Var(Z_{(k)})} \le \frac{1}{\sqrt{n}} \sum_{k=1}^{\lfloor \varepsilon n \rfloor} \sqrt{\frac{2}{k}}.
\end{align*}
Note that $\sum_{k=1}^{n} k^{-1/2} \sim 2\sqrt{n}$ as $n\to\infty$, so $\E[S_n^{low}(\varepsilon)]$ is $O(\sqrt{\varepsilon})$ as $n\to\infty$. A similar argument shows the same bound applies to $S_n^{high}(\varepsilon)$ also, since the sum is over the same number of terms and each summand in $S_n^{high}(\varepsilon)$ is strictly less than that in $S_n^{low}(\varepsilon)$. 

For the central terms, we apply the Bahadur-Kiefer representation of central quantiles (Theorem \ref{thm:Bahadur}) with $k=\lfloor pn \rfloor$, for $\varepsilon<p<1-\varepsilon$, to obtain
\begin{equation} \label{eq:BK_exp}
	Z_{(k)} - \xi_p = - \frac{\mathbb{F}_n(\xi_p) - p}{p} + R_n,
\end{equation}
where $R_n=O(n^{-3/4} (\log n)^{1/2}(\log \log n)^{1/2})$ as $n\to\infty$. We have $z_k = H_n - H_{k-1} = \log(n) - \log(k-1) + O(1/n) = -\log(p) + O(1/n) = \xi_p + O(1/n)$. Therefore, we can replace $\xi_p$ with $z_k$ on the LHS of \eqref{eq:BK_exp}, with $R_n$ having the same order. From Donsker's theorem, we have $\sqrt{n}(\mathbb{F}_n(\xi_p) - p) \xrightarrow{\hspace{2mm}d\hspace{2mm}} B(p)$ uniformly on $(0,1)$, where $B$ is a standard Brownian bridge on $[0,1]$ \citep{shorack1986}. Therefore, expressing $S_n^{cent}(\varepsilon)$ as a Riemann sum, we have
\begin{equation*}
	S_n^{cent}(\varepsilon) = \frac{1}{n} \sum_{k=\lfloor \varepsilon n \rfloor + 1}^{\lfloor (1-\varepsilon) n \rfloor} \sqrt{n} \left\lvert Z_{(k)} - z_k \right\rvert \xrightarrow{\hspace{2mm}d\hspace{2mm}} \int_{\varepsilon}^{1-\varepsilon} \frac{\lvert B(p)\rvert}{p} \dd p, \quad n\to\infty.
\end{equation*}
Letting $\varepsilon\to 0$, we have $S_n^{cent}(\varepsilon) \xrightarrow{\hspace{2mm}d\hspace{2mm}} \int_{0}^{1} (\lvert B(p)\rvert/p) \dd p = S$, and $S_n^{low}(\varepsilon), S_n^{high}(\varepsilon) \xrightarrow{p} 0$. Hence, by Slutsky's theorem $S_n \xrightarrow{\hspace{2mm}d\hspace{2mm}} S$.

To calculate the expected value and variance, we note that $\{B(p): p\in(0,1)\}$ is a zero-mean Gaussian process with covariance $\Cov(B(p),B(q)) = \min(p,q) - pq$. Therefore, by Lemma \ref{lem:norm_abs_mom}(a), we have $\E[\lvert B(p)\rvert] =  \sqrt{2\Var{(B(p))}/\pi}=\sqrt{2p(1-p)/\pi}$. Hence
\begin{equation*}
	\E[S] = \int_0^1 \frac{\E[|B(p)|]}{p} \dd p = \sqrt{\frac{2}{\pi}} \int_0^1 \sqrt{\frac{1-p}{p}} \dd p = \sqrt{\frac{\pi}{2}}.
\end{equation*}
For the variance, applying Lemma \ref{lem:norm_abs_mom}(b) gives 
\begin{align*}
	\Var(S) = \int_0^1 \int_0^1 \frac{\Cov(|B(p)|,|B(q)|)}{pq} \dd p \dd q =	\frac{2}{\pi}\int_0^1\!I(q)\dd q,
\end{align*}
where
\begin{align*}
  I(q) = \int_0^1
  \sqrt{\frac{(1-p)(1-q)}{pq}}\Bigl(\rho(p,q)\arcsin\rho(p,q)+\sqrt{1-\{\rho(p,q)\}^2}-1\Bigr)
  \dd p,
\end{align*}
and where the correlation function $\rho(p,q)$ of the Brownian bridge can be written as 
\begin{equation*}
    \rho(p,q) = \sqrt{\frac{\min(p,q) (1-\max(p,q))}{\max(p,q) (1-\min(p,q))}}.
\end{equation*}
Evaluating this double integral gives the required solution, $\Var(S) = 4 \log(2)-\tfrac{\pi}{2} - 1$. As the calculations are standard but lengthy, we present this information in the Supplementary Material, SM5.
\end{proof}

%%%%%%%%%%%%%%%%%%%%%%%%%%%%%%%%%%%%%%%%%%%%%%%%%%%%%%%%%%%%%%%%%%
\subsection{Proof of Proposition \ref{prop:AD_equiv}}
\label{app:proof_prop:AD_equiv}
\begin{proof}

For simplicity of notation, we hereafter drop the `bin' index $b$ from notation, and replace $\nu_k^{(b)}$ and $n_b$ with $\nu_k$ and $n$, respectively. Under the null hypothesis:
\begin{equation*}
    \nu_k = \Phi^{-1}\left(F_k \left(U_{(k)}\right)\right), \quad k=1,\dots,n,
\end{equation*}
where $U_{(1)} \le \cdots \le U_{(n)}$ are order statistics of iid $\mbox{Uniform}(0,1)$ random variables, $F_k$ is the distribution function of the $\mbox{Beta}(k,n-k+1)$ distribution, and $\Phi^{-1}$ is the quantile function of the standard normal distribution. To calculate the limit distribution of $\mathcal{A}_n^2$, we follow a similar approach to the proof of Proposition~\ref{prop:EMAD_limit}, and set $0<\varepsilon\ll 1$ and split the sum into central and boundary terms. That is, we write $\mathcal{A}_n^2 = \mathcal{A}_n^{low}(\varepsilon) + \mathcal{A}_n^{cent}(\varepsilon) + \mathcal{A}_n^{high}(\varepsilon)$, where
\begin{align*}
    \mathcal{A}_n^{low} (\varepsilon) = \frac{1}{n} \sum_{k=1}^{\lfloor \varepsilon n \rfloor} \nu_k^2, 
    \quad \mathcal{A}_n^{cent} (\varepsilon) = \frac{1}{n} \sum_{k=\lfloor \varepsilon n \rfloor+1}^{\lfloor (1-\varepsilon) n \rfloor} \nu_k^2,
    \quad \mathcal{A}_n^{high} (\varepsilon) = \frac{1}{n} \sum_{k=\lfloor (1-\varepsilon) n \rfloor+1}^{n} \nu_k^2.
\end{align*}
Note that, under the null hypothesis, $\nu_k \sim N(0,1)$ for all $n\in\mathbb{N}$ and $k=1,\ldots,n$. Therefore, $\E[\nu_k^2]=1$ for all $k,n$. Hence, we have
\begin{align*}
    \E[\mathcal{A}_n^{low} (\varepsilon)] = \frac{1}{n} \sum_{k=1}^{\lfloor \varepsilon n \rfloor} \E[\nu_k^2] = \frac{\lfloor \varepsilon n \rfloor}{n} \le \varepsilon.
\end{align*}
So, Markov's inequality shows that $\mathcal{A}_n^{low} (\varepsilon) \xrightarrow{p} 0$ as $\varepsilon\to0$. A similar argument shows that $\mathcal{A}_n^{high} (\varepsilon) \xrightarrow{p} 0$ as $\varepsilon\to0$

For the central terms, we set $k=\lfloor tn \rfloor$ for $\varepsilon<t<1-\epsilon$. Applying the Bahadur–Kiefer representation (Theorem \ref{thm:Bahadur}) for uniform data, we have:
\begin{equation*}
U_{(k)} - t = -(\mathbb{F}_n(t)- t) + R_n,	
\end{equation*}
where $R_n=O_p(n^{-3/4} (\log n)^{1/2}(\log \log n)^{1/2})$ as $n\to\infty$. From Donsker's theorem, we have $\sqrt{n}(\mathbb{F}_n(t) - t) \xrightarrow{\hspace{2mm}d\hspace{2mm}} B(t)$ uniformly on $(0,1)$, where $B$ is a standard Brownian bridge on $[0,1]$ \citep{shorack1986}. Thus, 
\begin{equation*}
    U_{(k)} = t + \frac{B(t)}{\sqrt{n}} + o_p(n^{-1/2}).
\end{equation*}
Next, we derive a limiting expression for $F_k(t+u/\sqrt{n})$ for fixed $u\in\RR$ and $n\to\infty$. Note that the event $\{U_{(k)} \le x\}$ is equivalent to the event that at least $k$ out of the $n$ $\mbox{Uniform}(0,1)$ variables are less than or equal to $x$. Therefore, the CDF of $U_{(k)}$ evaluated at $x$ is equal to the probability that a Binomial random variable with $n$ trials and success probability $x$ yields $k$ or more successes. Let $p_n = t + u/\sqrt{n}$. We can rewrite the distribution function $F_k(p_n)$ as:
\begin{equation*}
    F_k(p_n) = \Pr(U_{(k)} \le p_n) = \Pr(Y_n \ge k)
\end{equation*}
where $Y_n \sim \mbox{Bin}(n, p_n)$. Because $Y_n$ is a sequence of $\mbox{Bin}(n, p_n)$ random variables we have
\begin{equation*}
  \frac{Y_n - \mathbb{E}[Y_n]}{\sqrt{\text{Var}(Y_n)}} \xrightarrow{\hspace{2mm}d\hspace{2mm}} Z \sim \mathcal{N}(0,1), \quad n\to\infty.
\end{equation*}
The mean and variance of $Y_n$ are
\begin{align*}
  \mathbb{E}[Y_n] &= np_n = n\left(t + \frac{u}{\sqrt{n}}\right) = nt + u\sqrt{n},\\
  \text{Var}(Y_n) &= np_n(1-p_n) = n\left(t + \frac{u}{\sqrt{n}}\right)\left(1 - t - \frac{u}{\sqrt{n}}\right) = nt(1-t) + O(\sqrt{n}).
\end{align*}
Therefore, the inequality for $Y_n$ is equivalent to
\begin{equation*}
    Y_n \ge k \iff \frac{Y_n - \mathbb{E}[Y_n]}{\sqrt{\text{Var}(Y_n)}} \ge \frac{\lfloor tn \rfloor- (nt + u\sqrt{n})}{\sqrt{nt(1-t) + O(\sqrt{n})}}.
\end{equation*}
The limit of the right-hand deterministic sequence as $n \to \infty$ is $-u/\sqrt{t(1-t)}$. Finally, taking the limit of the probability yields
\begin{equation*}
    \lim_{n \to \infty} \Pr(Y_n \ge k) = \Pr\left(Z \ge \frac{-u}{\sqrt{t(1-t)}}\right) = \Phi\left(\frac{u}{\sqrt{t(1-t)}}\right),
\end{equation*}
where the final equality follows from the symmetry of the standard normal distribution. Combining these yields 
\begin{equation*}
    \lim_{n \to \infty} \Phi^{-1}\left( F_k\left(U_{(k)}\right)\right) = \frac{B(t)}{\sqrt{t(1-t)}}.
\end{equation*}
Noting that $\mathcal{A}_n^{cent} (\varepsilon)$ is a Riemann sum, we have
\begin{equation*}
    \mathcal{A}_n^{cent} (\varepsilon) \xrightarrow{\hspace{2mm}d\hspace{2mm}} \int_\varepsilon^{1-\varepsilon} \frac{(B(t))^2}{t(1-t)} \dd t, \quad n \to \infty.
\end{equation*}
Finally, letting $\varepsilon\to 0$ gives the required result.
\end{proof}

%%%%%%%%%%%%%%%%%%%%%%%%%%%%%%%%%%%%%%%%%%%%%%%%%%%%%%%%%%%%%%%%%%
\subsection{Proof of Proposition \ref{prop:EMAD_sensitivity}}
\label{app:proof_prop:EMAD_sensitivity}

\begin{proof}
Let $k=\lfloor u n\rfloor$ for $u\in(0,1)$. Using the Bahadur-Kiefer representation in \eqref{eq:BK_exp} (and text below), we have
\begin{equation*}
	\hat{z}_k - z_k = - \frac{\mathbb{F}_n(\xi_u) - u}{u} + R_n,
\end{equation*}
where $R_n=O(n^{-3/4} (\log n)^{1/2}(\log \log n)^{1/2})$ as $n\to\infty$. Under the local alternative distribution, the empirical process $\sqrt{n}(\mathbb{F}_n(\xi_u) - u)$ has a deterministic drift, and converges in distribution to $B(u) + \lambda h(u)$ \citep[see \S4.2 of][]{shorack1986}. Hence,
\begin{equation*}
    \sqrt{n}(\hat{z}_k - z_k) \xrightarrow{\hspace{2mm}d\hspace{2mm}} -\frac{ B(u) + \lambda h(u)}{u}, \quad n\to\infty.
\end{equation*}
So we have
\begin{equation*}
    S_{n,\theta} \xrightarrow{\hspace{2mm}d\hspace{2mm}} \int_0^1 \frac{\lvert B(u) + \lambda h(u)\rvert}{u} \dd u, \quad n\to\infty.
\end{equation*}
The Brownian bridge is a zero-mean Gaussian process, with $\sigma^2 \coloneq \Var(B(u)) = u(1-u)$, so $\lvert B(u) + \lambda h(u)\rvert$ has a folded normal distribution, so by Lemma \ref{lem:norm_abs_mom}(a) 
\begin{equation*}
    \E\left[\lvert B(u) + \lambda h(u)\rvert\right] = \sigma \sqrt{\frac{2}{\pi}} \exp\left(-\frac{\lambda^2 h^2(u)}{2\sigma^2}\right) + \lambda h(u) \left(1 - 2 \Phi\left(-\frac{\lambda h(u)}{\sigma}\right)\right).
\end{equation*}
Expanding the exponential function and $\Phi$ as Taylor series about $0$ gives,
\begin{equation*}
    \E\left[\lvert B(u) + \lambda h(u)\rvert\right] \sim \sigma \sqrt{\frac{2}{\pi}} + \frac{\lambda^2 h^2(u)}{\sigma\sqrt{2\pi}}, \quad \lambda\to 0.
\end{equation*}
Taking expectations of the EMAD statistics under the null and local alternative distributions then gives
\begin{align*}
    \E[{S_{n,\theta}}] - \E[{S_{n,0}}] &\to \int_{0}^1 \frac{\E\left[\lvert B(u) + \lambda h(u)\rvert\right]}{u} \dd u - \int_{0}^1 \frac{\E\left[\lvert B(u)\right]}{u} \dd u\\
    &=\lambda^2 \int_{0}^1 \frac{h^2(u)}{u \sqrt{2\pi u (1-u)}} \dd u.
\end{align*}
\end{proof}

%%%%%%%%%%%%%%%%%%%%%%%%%%%%%%%%
\section{Inference for deep Generalised Pareto regression models} \label{app:POT}
Following \citet{Richards2024}, we construct deep Generalised Pareto regression models using multilayered perceptrons (MLPs). We construct  MLPs as parametric functions $\mathbf{g}_{\boldsymbol{\psi}}:\mathbb{R}^d \mapsto \mathbb{R}^p$ that map an input vector of covariates $\bm{x}\in\mathbb{R}^d$ to a $p$-dimensional output through estimable parameters contained in $\boldsymbol{\psi}$. For the deep SPAR model in Section~\ref{sec:ex_SPAR}, we consider two MLPs with $p=1$: one for the exceedance threshold $u_\tau(\cdot)$ (i.e., where $u_\tau(\bm{w}):=\mathbf{g}_{\boldsymbol{\psi}}(\bm{w})$) and one for the exponential scale parameter  (i.e., where $ \sigma(\bm{w}):=\mathbf{g}_{\boldsymbol{\psi}}(\bm{w})$). For the surrogate modelling in Section~\ref{sec:ex_surrogate}, we further consider an MLP with $p=2$; here $(\sigma(\bm{x}),\xi(\bm{x})):=\mathbf{g}_{\boldsymbol{\psi}}(\bm{x}).$

In this work, we construct the neural networks $\mathbf{g}_{\boldsymbol{\psi}}$ as a composition of $L\in\mathbb{N}$ equal-width {hidden layers}, $\mathbf{g}^{(l)}$ for $l=1,\dots,L,$ and an {output} layer, $\mathbf{g}^{(L+1)},$ such that $\mathbf{g}_{\boldsymbol{\psi}}(\cdot) := \mathbf{g}^{(L+1)}\circ\dots \circ \mathbf{g}^{(1)}(\cdot)$. Each hidden layer, $l=1,\dots,L$, has fixed width $\eta>0$ and ReLU hidden activation function, such that  
\begin{equation}
\label{Eq:MLP}
\bm{x}^{(l)}:=\mathbf{g}^{(l)}(\bm{x}^{(l-1)})=\rm{ReLU}\left( \bm{W}^{(l)}\bm{x}^{(l-1)}+\bm{b}^{(l)}\right)\in \mathbb{R}^\eta,
\end{equation}
where $\rm{ReLU}(\bm{x})=(\max(x_1,0),\max(x_2,0),\dots)$, $\bm{x}^{(0)}:=\bm{x}$ is the input covariates with dimension $d$ and $\bm{x}^{(L+1)}$ is the output with dimension $p$. The parameter set ${\boldsymbol{\psi}}$ comprises estimable weights and biases, $ \bm{W}^{(l)}$ and $ \bm{b}^{(l)},l=1,\dots,L$, with their dimensions determined by the relative input and output of layer $l$. The output layer $L+1$ takes a similar form to \eqref{Eq:MLP}, but with the ReLU function replaced with an appropriate link function to constrain $\bm{x}^{(L+1)}$: when $\bm{x}^{(L+1)}$ corresponds to $u_\tau$ or $\sigma$, we use the exponential function to ensure strict positivity; where the shape parameter $\xi$ is required, a compressed hyperbolic tangent function is used to constrain the output of the MLP to $(-0.5,0.1)$.

Estimation of ${\boldsymbol{\psi}}$ proceeds via minimisation of an empirical loss function using the ADAM algorithm \citep{kingma2014adam}, a variant of stochastic gradient descent. The loss function differs with the target parameter: for $u_\tau$ we use the quantile loss and, otherwise, we use the negative log-likelihood implied by the considered model. To mitigate overfitting during training, we partition data into training (80\%) and validation (20\%) sets with the latter used to check for parameter convergence. In particular, at each iteration of ADAM, we evaluate the loss on the validation data and perform early-stopping \citep{prechelt2002early} of the training scheme if the validation loss has not decreased in five iterations. Note that, for model training, we do not use the test data required for producing the goodness-of-fit diagnostics.

%%%%%%%%%%%%%%%%%%%%%%%%%%%%%%%%%%%%%%%%%%%%%%%%%%%%%%%%%%%%%%%%%%

\baselineskip=14pt
\begingroup

\begin{refcontext}[sorting=nyt]
    \printbibliography
\end{refcontext}

\endgroup
\clearpage

%%%%%%%%%%%%%%%%%%%%%%%%%%%%%%%%%%%%%%%%%%%%%%%%%%%%%%%%%%%
%%%%%%%%%%%%%%%%%  Supplementary %%%%%%%%%%%%%%%%%%%%%
%%%%%%%%%%%%%%%%%%%%%%%%%%%%%%%%%%%%%%%%%%%%%%%%%%%%%%%%%%%
% Redefine numbering format
\renewcommand{\thesection}{SM\arabic{section}}
\renewcommand{\thefigure}{SM\arabic{figure}}
\renewcommand{\thetable}{SM\arabic{table}}
\renewcommand{\theequation}{SM\arabic{equation}}
\setcounter{section}{0}
\setcounter{page}{1}
\setcounter{footnote}{0}
\setcounter{figure}{0}
\setcounter{equation}{0}

%%%%%%%%%%%%%%%%%%%%%%%%%%%%%%%%%%%%%%%%%%%%%%%%%%%%%%%%%%%
\bookmarksetup{startatroot} 
\pdfbookmark[0]{Supplementary Material}{supplement_bookmark}

\clearpage
\begin{center}
	% --- SUPPLEMENT TITLE ---
	{\LARGE \bfseries Supplementary material for ``Diagnostic tools for extreme value regression models''} \\[1.5em]
	
	% --- DATE ---
	{\large \today} \\[1.5em]
	
	% --- AUTHORS WITH SUPERSCRIPTS ---
	{\large E. Mackay$^1$, J. Richards$^2$, and P. Jonathan$^3$}\\[0.8em]
	
	% --- AFFILIATIONS & MAILS ---
	{\small $^1$Department of Engineering, University of Exeter, Penryn, TR10 9FE, UK. \\ 
		\texttt{e.mackay@exeter.ac.uk}. ORCID: 0000-0001-7121-4231} \\[0.8em]
	
	{\small $^2$School of Mathematics and Maxwell Institute for Mathematical Sciences, \\ 
		University of Edinburgh, Edinburgh, EH9 3FD, UK. \\ 
		\texttt{jordan.richards@ed.ac.uk}. ORCID: 0000-0002-0697-2551} \\[0.8em]
	
	{\small $^3$School of Mathematical Sciences, Lancaster University, Lancaster, LA1 4YF, UK. \\ 
		\texttt{p.jonathan@lancaster.ac.uk}. ORCID: 0000-0001-7651-9181} 
\end{center}

\vspace{2.5em} % 

%%%%%%%%%%%%%%%%%%%%%%%%%%%%%%%%
\setcounter{section}{0}

\section{Goodness-of-fit tests} 

\subsection{Sampling distributions}
As noted in Section 4 of the main text, the sampling distributions of the EMAD and Cram\'{e}r-von Mises (CvM) family of goodness-of-fit statistics do not have closed-form expressions, but converge to an asymptotic form. \autoref{fig:test_stat_dist} shows the exceedance probability of $W_n$ (CvM statistic), $A_n$ (Anderson-Darling statistic), $A_{R,n}$ (right-tail-weighted Anderson-Darling statistic), and $S_n$ (EMAD) for sample sizes $n=10,30,100,300$, and 1000, estimated by Monte Carlo simulation using $10^8$ trials at each value of $n$. To avoid overlap, with the EMAD distribution, the distributions of $W_n$, $A_n$, and $A_{R,n}$ are shown, rather than those of $W_n^2$, $A_n^2$, and $A_{R,n}^2$. For the CvM family with sample sizes $n\ge 10$, the distribution is relatively constant for lower exceedance probabilities, with the largest changes in the upper tail. The sampling distribution of the EMAD statistic exhibits larger changes with sample size and is slower to converge to the asymptotic form. 

\begin{figure}[h!]
	\centering    
	\includegraphics[scale=0.7]{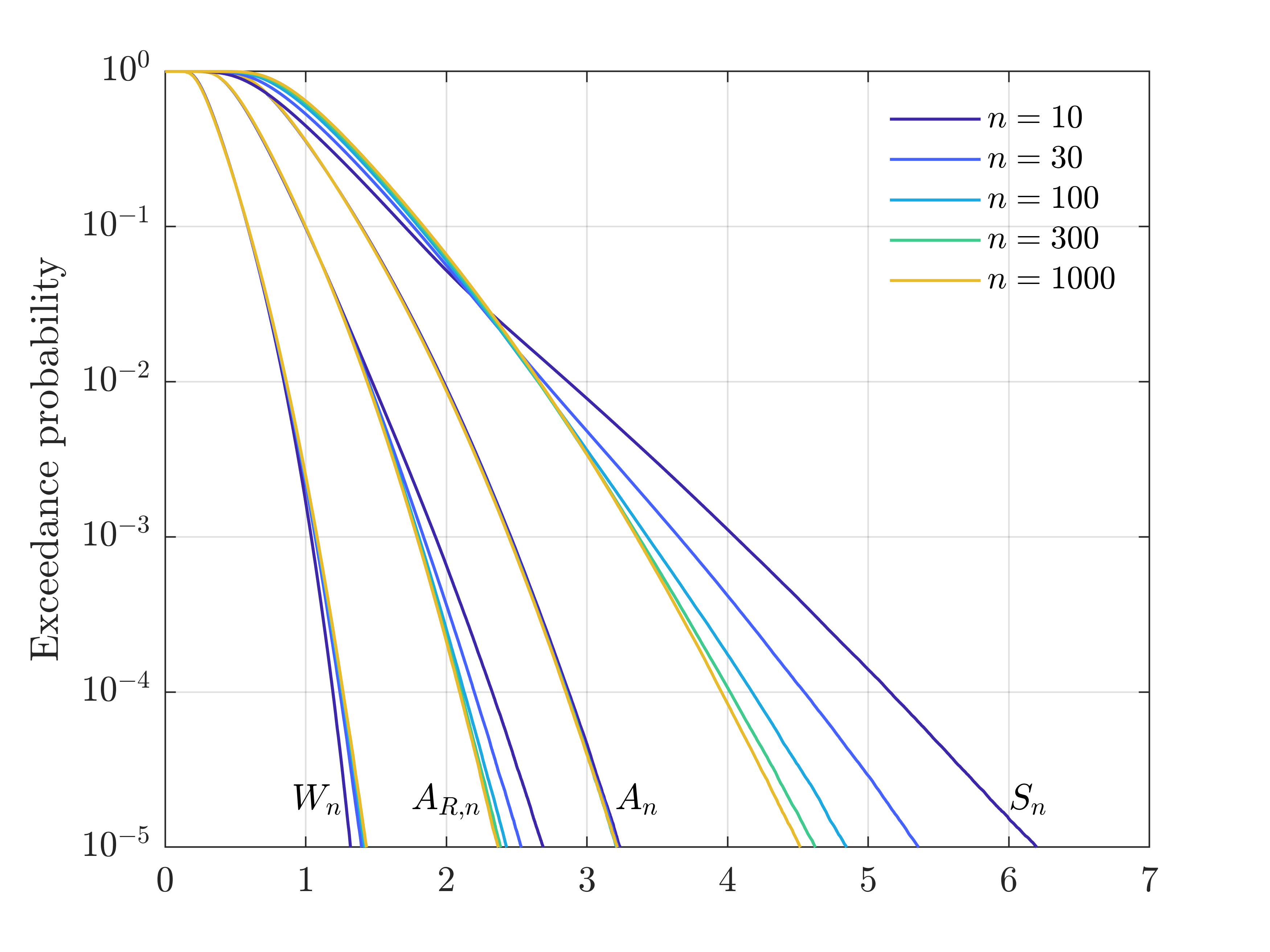}
	\caption{Exceedance probability for goodness-of-fit test statistics, $W_n$, $A_n$, $A_{R,n}$, and $S_n$ under the null hypothesis, for various sample sizes $n$.}
	\label{fig:test_stat_dist}
\end{figure}

Section 4.1.3 also showed that, under the null hypothesis, the distribution of $\mathcal{A}_n^2$ (the mean-square value of the normalised residuals $\nu_k$) converges to the same asymptotic form as that of the AD statistic. \autoref{fig:AD_vs_gaussian} compares the sampling distributions of the two statistics for finite sample sizes of $n=10$, 30, and 100. The distributions are in close agreement for relatively small sample sizes of $n=30$.

\begin{figure}[h!]
	\centering    
	\includegraphics[width=\textwidth]{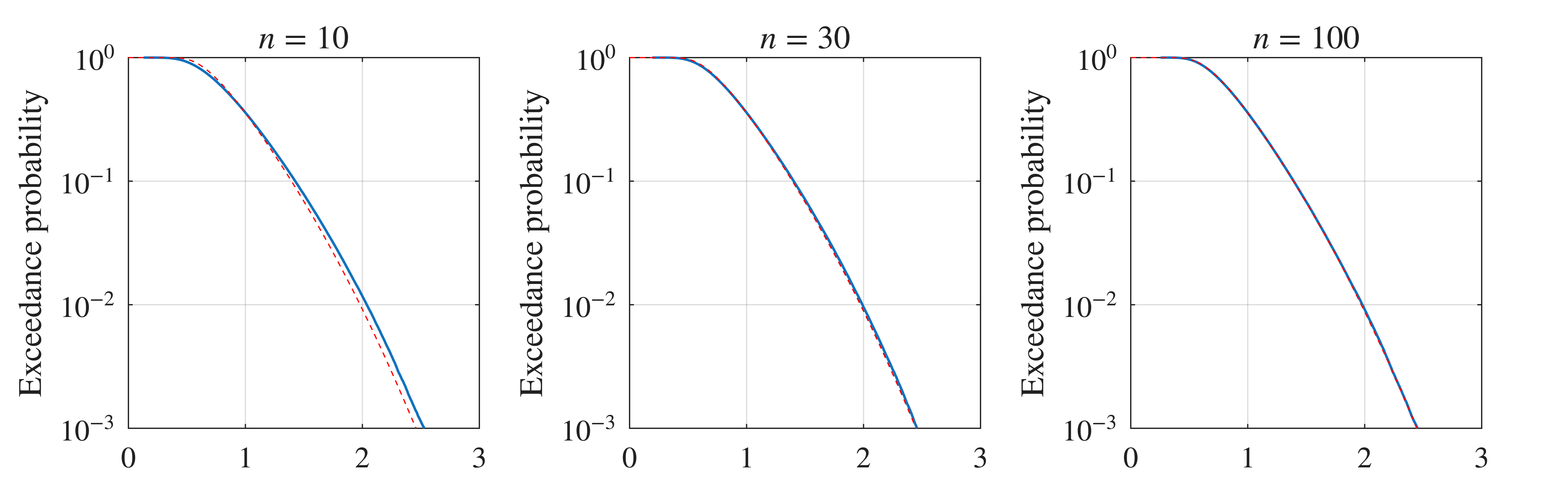}
	\caption{Exceedance probability for RMS value of normalised residuals $\mathcal{A}_n$ (solid lines) and AD statistic $A_n$ (dashed lines), for various sample sizes $n$.}
	\label{fig:AD_vs_gaussian}
\end{figure}
%%%%%%%%%%%%%%%%%%%%%%%%%%%%%%%%
\subsection{Comparison of test statistics}
The choice of test statistic used to quantify goodness of fit influences judgements about the quality of a given model. Due to the different weights that test statistics place on various parts of the distribution, a model for a given sample may fail one test at a certain significance level, but pass another test at the same level. We can examine the agreement between test statistics under the null hypothesis by generating samples from the uniform distribution on $[0,1]$, and comparing p-values for each sample. \autoref{fig:test_stat_joint_dens} shows empirical joint densities of $10^6$ test statistic p-values estimated from samples of size $n=100$. Whilst some pairs of test statistics exhibit reasonably strong correlation (e.g. CvM and AD, or ADR and EMAD), other pairs of statistics have much weaker correlation and exhibit a large degree of scatter. 

\begin{figure}[h!]
	\centering    
	\includegraphics[scale=0.7]{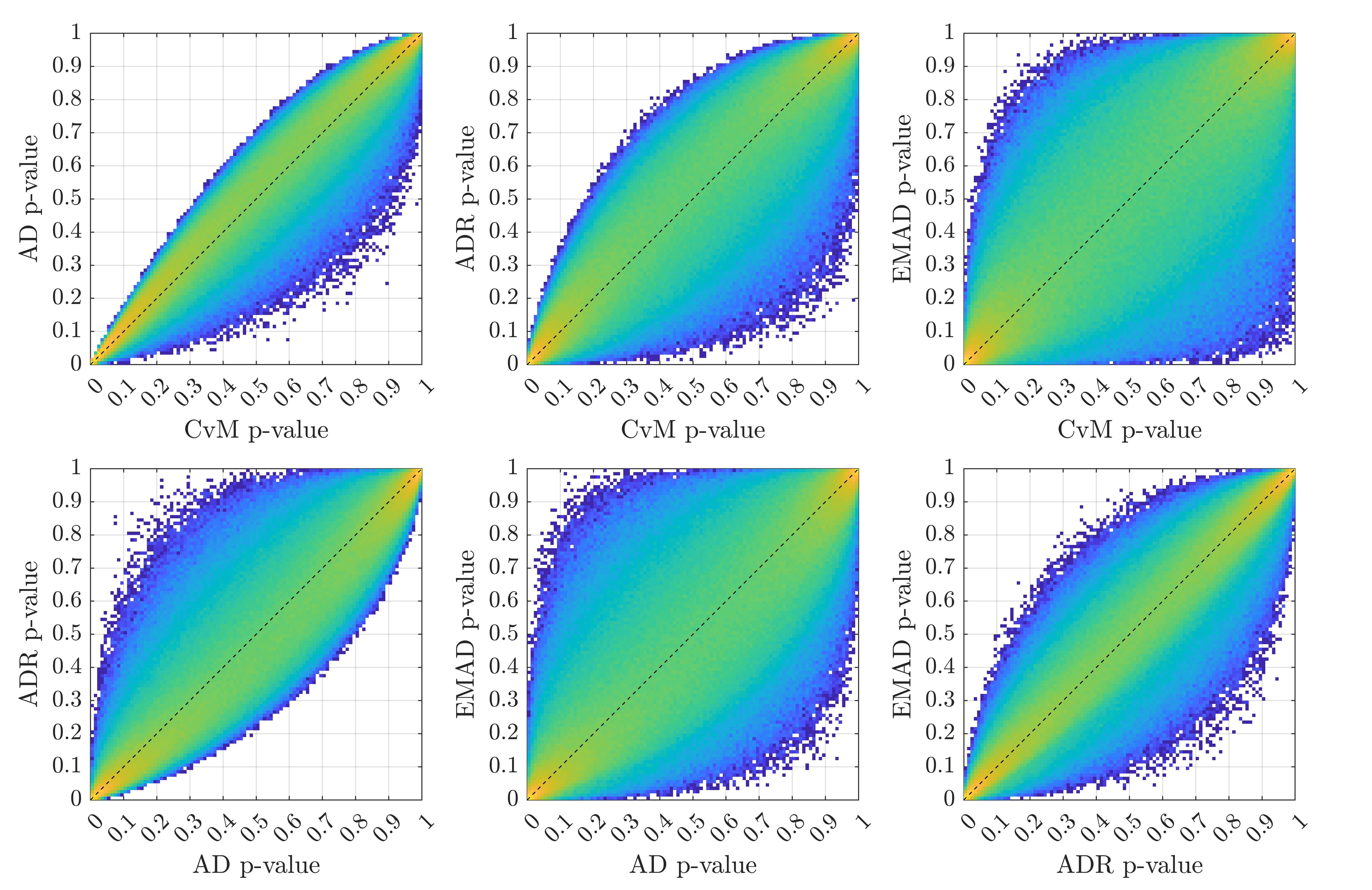}
	\caption{Empirical joint densities of  $10^6$ test statistic p-values estimated from samples of size $n=100$ from a uniform distribution. Colour scale indicates joint density on a logarithmic scale.}
	\label{fig:test_stat_joint_dens}
\end{figure}

\begin{figure}[h!]
	\centering    
	\includegraphics[width=0.6\textwidth]{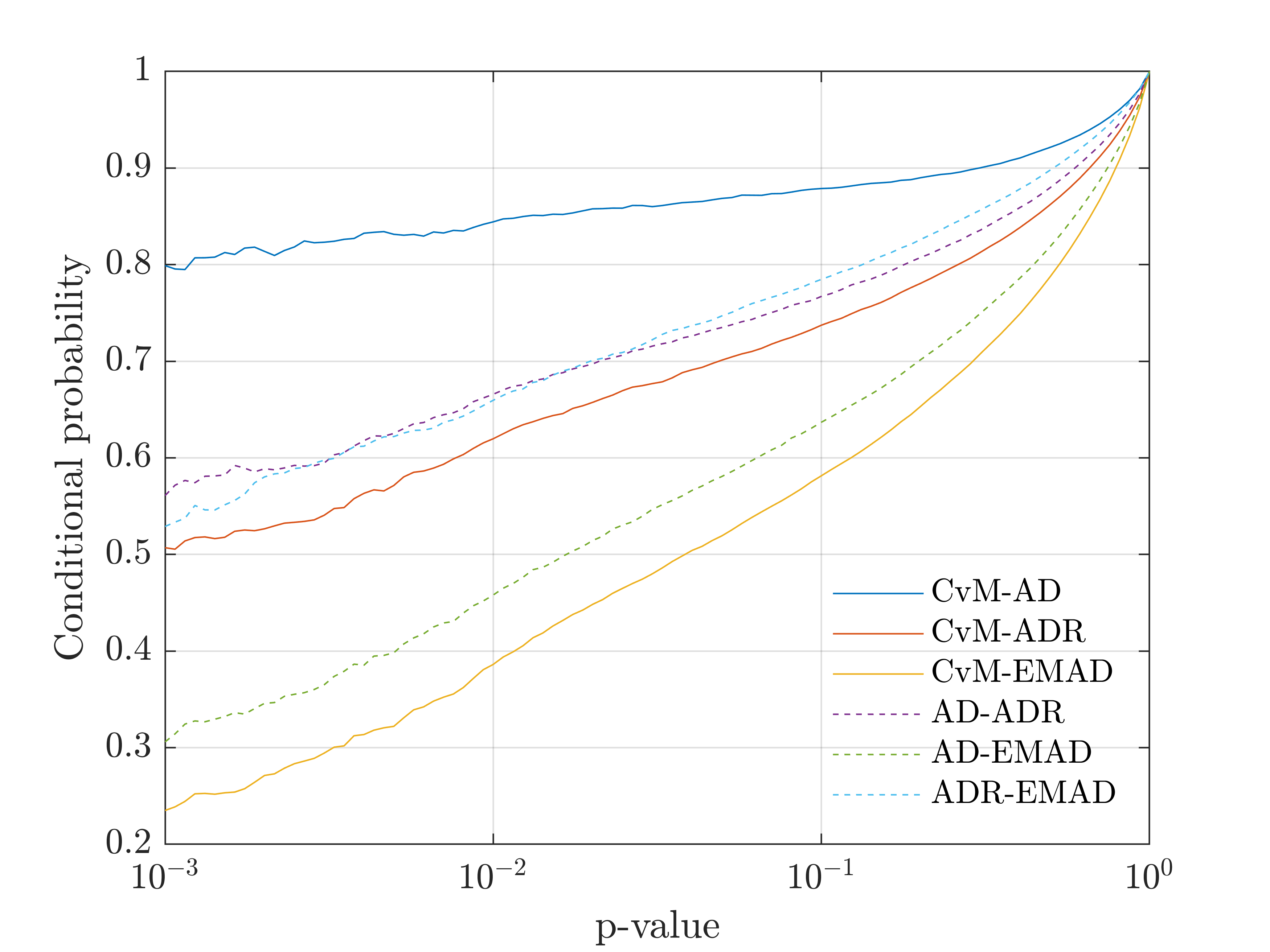}
	\caption{Probability that both test statistics have p-value less than a given level, conditional on one test statistic having a p-value less than this level. This is the conditional probability that a pair of test statistics both reject a sample at a given p-value.}
	\label{fig:test_stat_comp}
\end{figure}

\autoref{fig:test_stat_comp} shows the probability that pairs of test statistics both have p-value less than a given level, conditional on one test statistic having a p-value less than a given level. This is the conditional probability that a pair of test statistics both reject a sample at a given p-value. The CvM and AD statistics have the strongest joint rejection probability, which is in excess of 0.8 for p-values above $10^{-3}$. In contrast, the CvM and EMAD statistics have the lowest joint rejection probability, at less than 0.4 for a p-value of 0.01. 

\subsection{Sensitivity to perturbations}
In Section 4.1.3 of the main text, we consider the asymptotic sensitivity of various test statistics to perturbations in the uniform distribution of PIT values. Here we consider finite sample size sensitivity, for a specific example. For a sample size $n$, we generate IID $\mbox{Uniform}(0,1)$ variables $U_1,\ldots,U_n$, representing the PIT values under the null hypothesis. We then apply a perturbation to these values, representing an error in a model. As the interest is in extreme value models, we apply the perturbation to the upper tail of the distribution for $u\le u_0$, where $u$ is the exceedance probability and $0<u_0\le 1$. For $a>0$, we define
\begin{equation*}
	V_k = 
	\begin{cases}
		U_k, & U_k>u_0,\\
		u_0^{1-a} U_k^a, & U_k\le u_0,
	\end{cases}
	\qquad k=1,\ldots,n.
\end{equation*}
An example of the relationship between $V_k$ and $U_k$ is shown in \autoref{fig:PIT_perturbation} for the case $u_0=0.5$ and $a=0.5$, 0.75, 1, 1.25 and 1.5. Values of $a>1$ correspond to a negative bias in the model, i.e. the model predicts that a given quantile is less likely to be exceeded. Similarly, values of $a<1$ correspond to positive model bias.

\begin{figure}[h!]
	\centering    
	\includegraphics[width=0.6\textwidth]{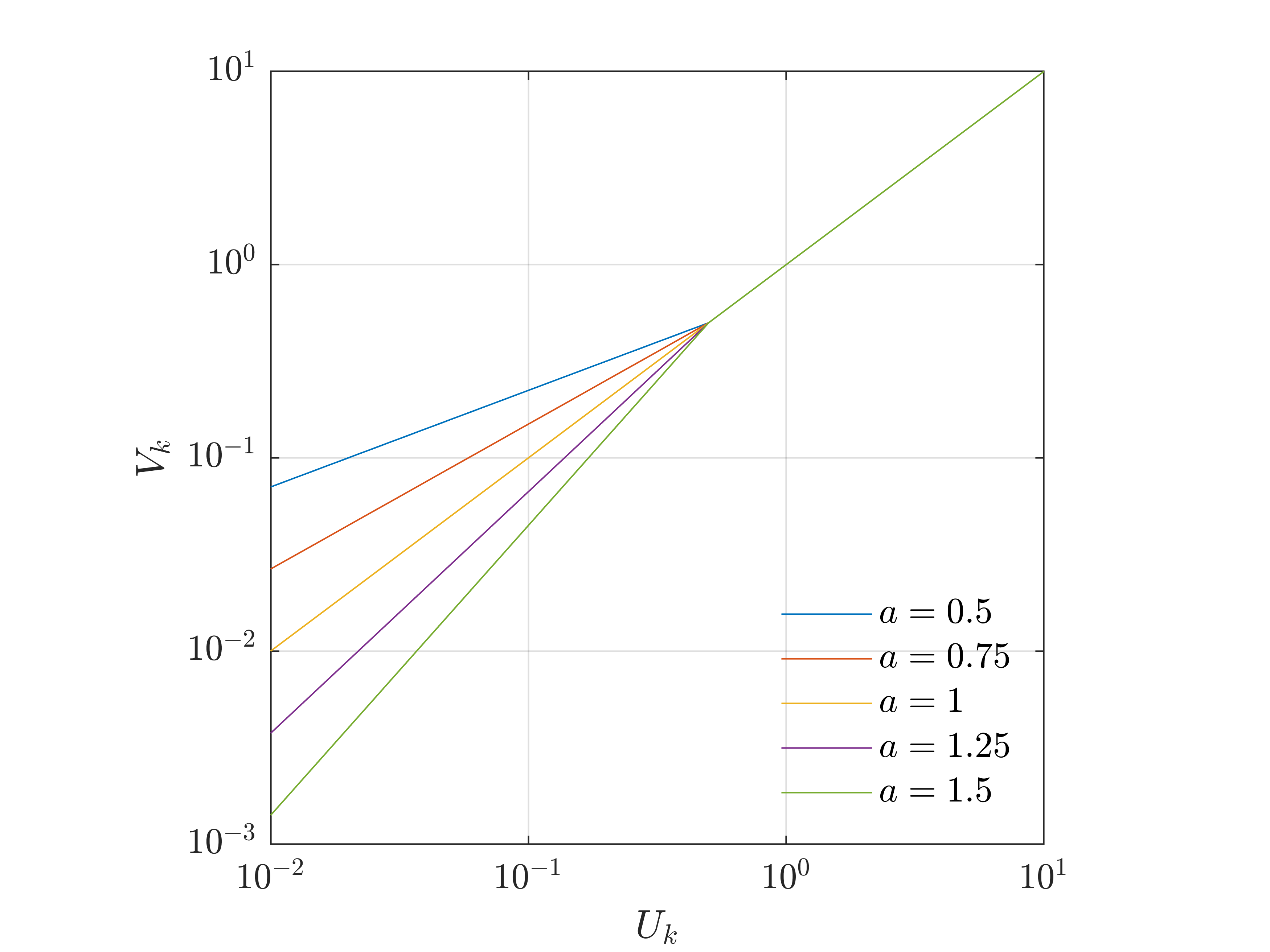}
	\caption{Relation between perturbed PIT values $V_k$ and uniform PIT values $U_k$ used in simulation study, for $u_0=0.5$ and various values of perturbation parameter $a$.}
	\label{fig:PIT_perturbation}
\end{figure}

Simulation studies were conducted for perturbation parameters in the range $0.25\le a\le 2$ and $0.3\le u_0 \le 1$ and sample sizes $n=25$, 50 and 100. For each pair of parameter values and sample size, $10^5$ samples were generated and the p-values of the ADR and EMAD statistics were calculated. \autoref{fig:EMAD_vs_ADR} shows contour plots of the mean p-values for each sample size, as functions of $u_0$ and $a$, as well as the ratios of the mean p-values. The change in p-value is not symmetric about $a=1$. For example, for $n=25$, there is a small region of values of $(u_0,a)$ where the mean p-values of the ADR and EMAD statistics increase slightly above 0.5. These occur for values of $a$ just less than 1, and lower values of $u_0$. In the simulation study, the maximum p-values were found to be only slightly larger than 0.5, with a maximum of 0.52 for the EMAD statistic for a sample size $n=25$ and parameters $u_0=0.3$ and $a=0.8$, and a maximum of 0.51 for the ADR statistic for a sample size $n=25$ and parameters $u_0=0.3$ and $a=0.85$. 

\begin{figure}[h!]
	\centering    
	\includegraphics[width=1\textwidth]{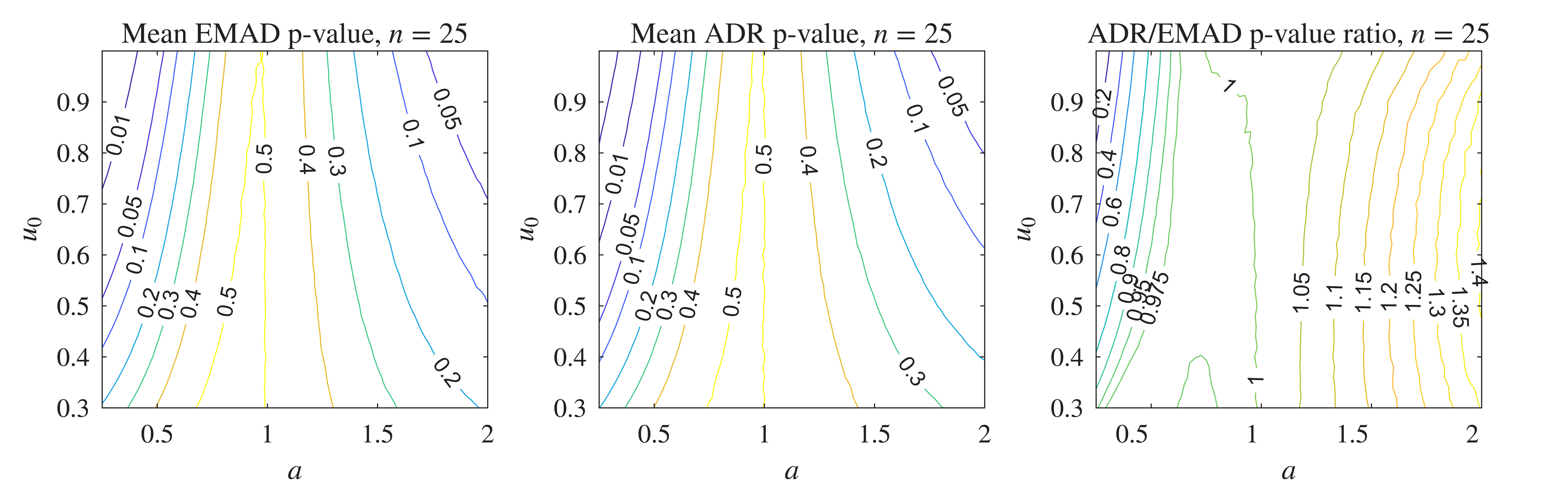}\\
	\includegraphics[width=1\textwidth]{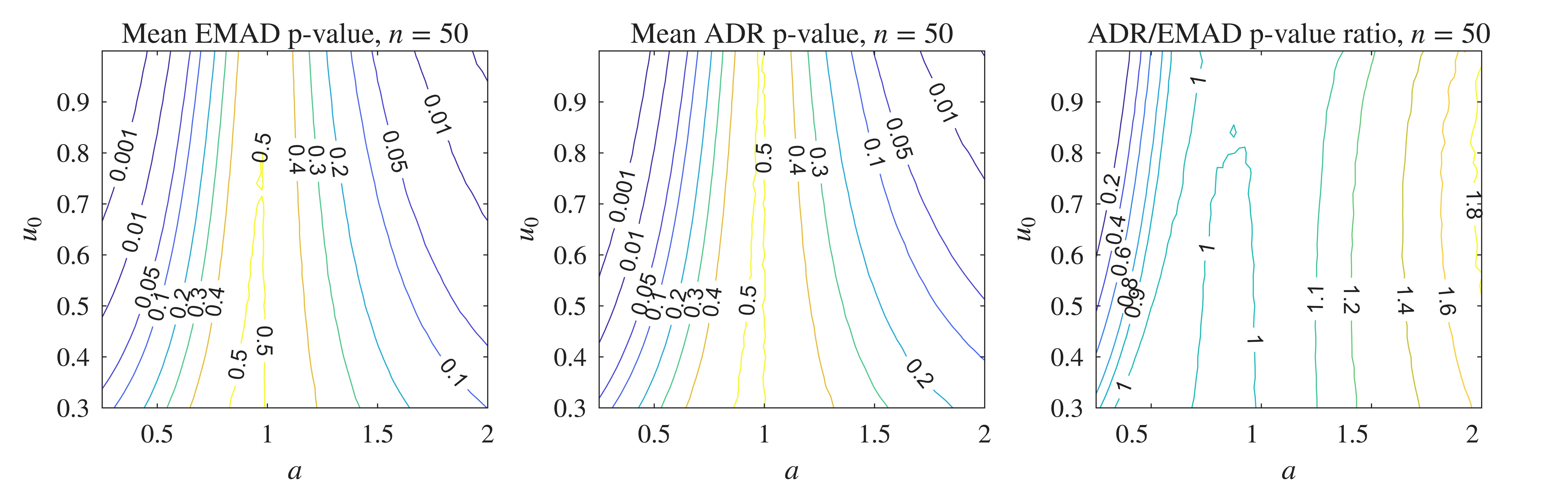}\\
	\includegraphics[width=1\textwidth]{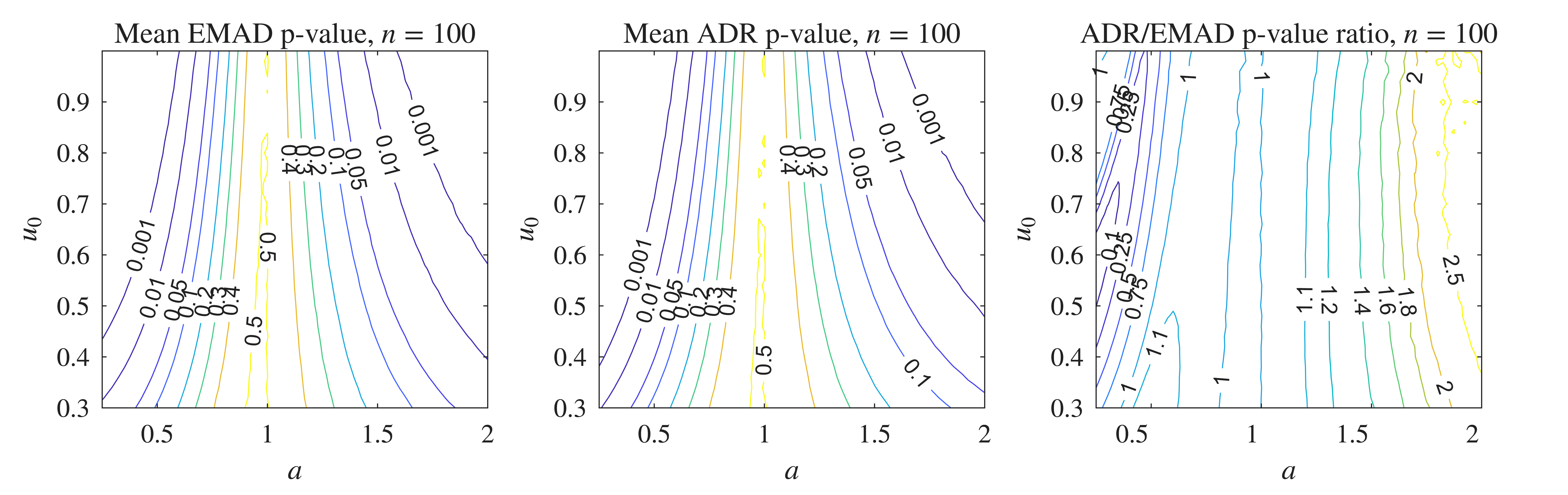}
	\caption{Mean p-values for EMAD (left) and ADR (middle) statistics, for perturbed PIT distributions with various perturbation parameters of $u_0$ and $a$. Right column shows ratio of mean p-values for ADR and EMAD statistics. Results shown for sample sizes $n=25$ (top), $n=50$ (middle), and $n=100$ (bottom).}
	\label{fig:EMAD_vs_ADR}
\end{figure}

For the rest of the parameter space investigated here, the mean p-value is less than 0.5 for the perturbed PIT values. For values of $a>1$ (negative model bias), the ratio of the mean ADR to EMAD p-values is positive for the three sample sizes and parameter ranges tested, indicating that the EMAD statistic is more sensitive to detecting these errors. For values of $a<1$ (positive model bias), the ADR statistic is more sensitive for smaller values of $a$. However, for $0.6<a<1$ the performance is similar, with the p-value ratio (ADR/EMAD) above 0.975 for the three sample sizes considered. These results suggest that the EMAD statistic is potentially more useful than the ADR statistic for detecting errors in the upper tails of a distribution, since it is more likely to flag models with negative bias as having a poor fit. Negative errors are potentially more important, as these correspond to models which predict large values are less likely than they in fact are.

%%%%%%%%%%%%%%%%%%%%%%%%%%%%%%%%%%%
\section{Multivariate normal copula on Laplace margins}
\subsection{Asymptotic density}
The asymptotic form of the density of a multivariate normal copula on exponential margins was derived by \citet{wadsworth2024statistical}. Here we consider the case of Laplace margins. Let $\bm{u}=(u_1,...u_d)\in(0,1)^d$ and $\bm{z} = \left(\Phi^{-1}(u_1), ..., \Phi^{-1}(u_d)\right)$, where $\Phi$ is the standard normal CDF. The copula density function for the multivariate normal distribution with correlation matrix $\mathrm{S}$ is
\begin{equation*}
	c_{\mathrm{S}}(\bm{u}) = \lvert \mathrm{S} \rvert^{-1/2} \exp \left( - \frac{1}{2} \|\bm{z}\|^2_2 - \frac{1}{2} \bm{z}^\top \mathrm{S}^{-1} \bm{z} \right),
\end{equation*}
where $\|\bm{z}\|_p = \left(\sum_{i=1}^d \lvert z_i \rvert^p\right)^{1/p}$ is the $L^p$ norm. Let $\bm{X}=(X_1,...,X_d)\in\RR^d$ be a random variable with copula density $c_{\mathrm{S}}$ and standard Laplace margins. Then $\bm{X}$ has joint density function
\begin{equation}\label{eq:norm_laplace_dens_SM}
	f_{\bm{X}}(\bm{x}) = 2^{-d} \lvert \mathrm{S} \rvert^{-1/2} \exp \left( - \|\bm{x}\|_1 - \frac{1}{2} \|\bm{z}\|^2_2 - \frac{1}{2} \bm{z}^\top \mathrm{S}^{-1} \bm{z} \right),
\end{equation}
where $z_i = \Phi^{-1}\left(F_L(x_i)\right)$ for $i=1,...,d$, and $F_L(x) = \tfrac{1}{2} + \sgn(x) \big(1-\exp(-|x|)\big)$ is the CDF of the standard Laplace distribution. From the symmetry of the normal and Laplace distribution we can write $z_i = -\sgn (x_i) \Phi^{-1}\left( \tfrac{1}{2} \exp\left(-\lvert x_i\rvert\right) \right)$ for $x_i\in \RR$, $i=1,...,d$. From the asymptotic properties of the inverse complementary error function \citep[][\S7.17(iii)]{NIST:DLMF}, we have for $u\to 0$,
\begin{equation*}
	\Phi^{-1}(u) = - \sqrt{-2\log(2u)} + \frac{\log(-\pi\log(2u))}{2\sqrt{-2\log(2u)}} + O\left(\frac{\log(-\log(2u))}{(-\log(2u))^{3/2}}\right).
\end{equation*}
Therefore, for $\lvert x_i\rvert \to \infty$ we have
\begin{equation*}
	z_i = \sgn (x_i) \left[ \sqrt{2\lvert x_i\rvert} - \frac{\log(\pi\lvert x_i\rvert )}{2\sqrt{2\lvert x_i\rvert}} \right] + O\left(\frac{\log(\lvert x_i\rvert)}{\lvert x_i\rvert^{3/2}}\right).
\end{equation*}
From here onwards it is useful to work in pseudo-polar coordinates $\bm{x}=r\bm{w}$, where $\bm{w}=(w_1,...,w_d)^\top$. For $r\to\infty$ we have
\begin{equation*}
	z_i^2 = \begin{cases}
		2r\lvert w_i\rvert - \log(\pi r \lvert w_i\rvert ) + O\left(\frac{\log(r)}{r}\right), & w_i\ne 0,\\
		0, &w_i=0.
	\end{cases}
\end{equation*}
Let $\mathcal{I}=\{i\in\{1,...,d\} : w_i\ne 0\}$. The first two terms in the exponential in \eqref{eq:norm_laplace_dens_SM} can then be written as
\begin{align*}
	\exp \left( - \|\bm{x}\|_1 - \frac{1}{2} \|\bm{z}\|^2_2  \right) &=  \exp \left( - \frac{1}{2} \sum_{i\in\mathcal{I}}  \log(\pi r \lvert w_i\rvert ) + O\left(\frac{\log(r)}{r}\right) \right)\\ 
	&= \left[1 +o(1)\right] r^{-d/2} \prod_{i\in\mathcal{I}} (\pi \lvert w_i\rvert )^{-1/2}.
\end{align*}
Now consider the terms $\bm{z}^\top \mathrm{S}^{-1} \bm{z} = \sum_{i=1}^d \sum_{i=1}^d \mathrm{S}^{-1}_{ij} z_i z_j$. For $i,j\in\mathcal{I}$ we have
\begin{align*}
	z_i z_j &= \sgn (w_i) \sgn (w_j) \left[2r\sqrt{\lvert w_i w_j \rvert} - \frac{1}{2}\sqrt{\frac{\lvert w_i\rvert}{\lvert w_j\rvert}} \log(\pi r \lvert w_j \rvert)  - \frac{1}{2}\sqrt{\frac{\lvert w_j\rvert}{\lvert w_i\rvert}} \log(\pi r \lvert w_i \rvert)\right] \\
	& \qquad + O\left(\frac{\log(r)}{r}\right).
\end{align*}
Therefore, the quadratic form can be written as
\begin{equation*}
	\frac{1}{2} \bm{z}^\top \mathrm{S}^{-1} \bm{z} = r\,\bm{a}^\top \mathrm{S}^{-1} \bm{a} - \frac{1}{2} \sum_{i\in\mathcal{I}} \sum_{j\in\mathcal{I}} \mathrm{S}^{-1}_{ij} \sgn (w_i) \sgn (w_j) \sqrt{\frac{\lvert w_i\rvert}{\lvert w_j\rvert}} \log(\pi r \lvert w_j \rvert) + O\left(\frac{\log(r)}{r}\right),
\end{equation*}
where $\bm{a}=(a_1,...,a_d)^\top$ and $a_i=\sgn(w_i) \sqrt{\lvert w_i\rvert}$, $i=1,...,d$. Collecting terms we arrive at the following asymptotic expression for the density
\begin{equation*}
	f_{\bm{X}}(r\bm{w}) \propto [1+o(1)] r^{\zeta(\bm{w})} \exp \left( - r \lambda(\bm{w}) \right), \quad r\to\infty,
\end{equation*}
where $\lambda(\bm{w}) = \bm{a}^\top \mathrm{S}^{-1} \bm{a}$ and $\zeta(\bm{w}) = -\tfrac{d}{2}+\tfrac{1}{2}\bm{a}^\top \mathrm{S}^{-1} \bm{b}$, with $b_i = \sgn(w_i) /\sqrt{\lvert w_i\rvert}$ for $i\in\mathcal{I}$ and $b_i = 0$ otherwise.

Defining pseudo-polar coordinates $R=\|\bm{X}\|_2$ and $\bm{W}=\bm{X}/R$, where $\|\cdot\|_2$ is the $L^2$ norm, the joint density of $(R,\bm{W})$ is given by 
\begin{equation*}
	f_{R,\bm{W}}(r,\bm{w}) = r^{d-1} f_{\bm{X}}(r\bm{w}) \propto [1+o(1)] r^{\zeta(\bm{w})+d-1} \exp \left( - r \lambda(\bm{w}) \right), \quad r\to\infty.
\end{equation*}
In this formulation, the radial component converges to a truncated gamma distribution with shape $\zeta(\bm{w})+d+1$ and scale $\kappa(\bm{w}) \coloneqq 1/\lambda(\bm{w})$. 

\subsection{Example}
In the example considered in the paper, we consider a five-dimensional case with correlation matrix
\begin{equation*}
	\mathrm{S} = 
	\begin{bmatrix}
		\;\;\;1.0000 & -0.4387 & \;\;\;0.5946 & \;\;\;0.0758 & -0.2198\\
		-0.4387 & \;\;\;1.0000 & -0.5885 & \;\;\;0.0361 & \;\;\;0.3887\\
		\;\;\;0.5946 & -0.5885 & \;\;\;1.0000 & \;\;\;0.0778 & -0.2404\\
		\;\;\;0.0758 & \;\;\;0.0361 & \;\;\;0.0778 & \;\;\;1.0000 & -0.1047\\
		-0.2198 & \;\;\;0.3887 & -0.2404 & -0.1047 & \;\;\;1.0000
	\end{bmatrix}.
\end{equation*}
To illustrate the range of tail shapes that this correlation matrix produces over the domain, we generate a sample of size $10^6$ and calculate the values of $\zeta(\bm{w})$ and $\lambda(\bm{w})$ at each point. The results are shown in \autoref{fig:MVN_characteristics}. The range of scales is relatively narrow, with most observations falling in the interval $[0.2, 1.2]$. Approximately 95\% of the values of $\zeta(\bm{w})$ are in the interval $[1,10]$, but the distribution has long tails (not shown), with some angles having $\zeta(\mathbf{w})<-100$ or $\zeta(\mathbf{w})>100$.

\begin{figure}[h!]
	\centering
	\includegraphics[scale=0.7]{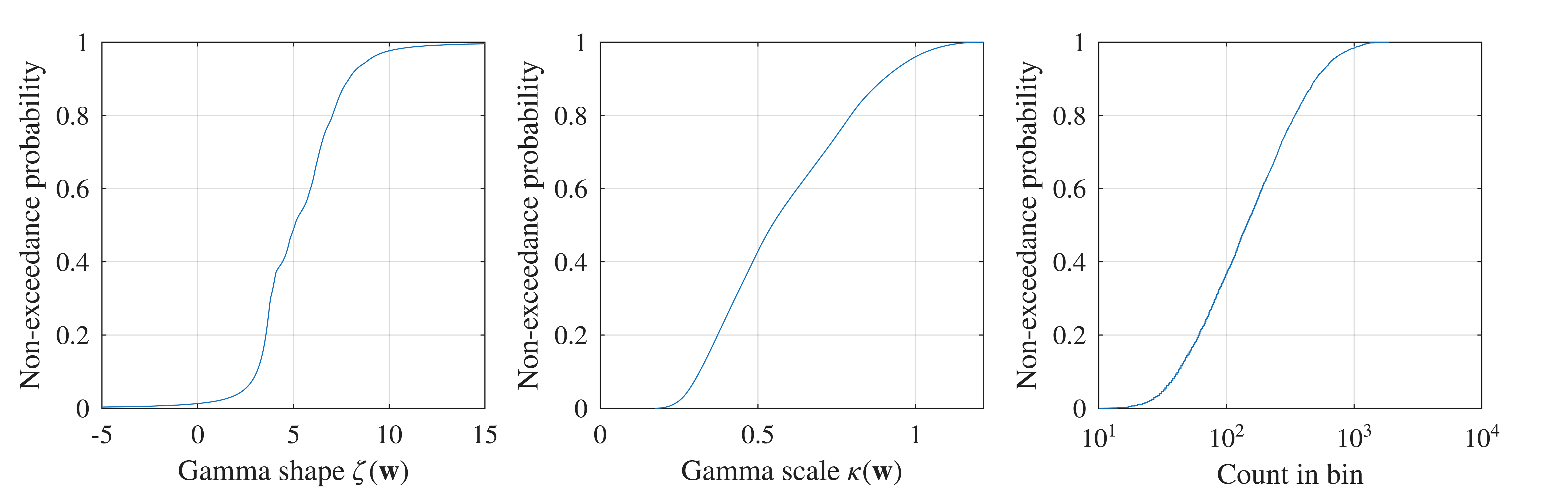}
	\caption{Empirical distributions of values of gamma shape parameter $\zeta(\bm{w})$ (left) and gamma scale parameter $\kappa(\bm{w})$ (middle) over the hypersphere, for the multivariate normal copula on Laplace margins. Right: Number of observations falling within $\pi/18$ radians of a set of 5890 pseudo-regularly spaced angles on $\mathbb{S}^4$, based on a sample of size $10^6$ from density \eqref{eq:norm_laplace_dens_SM}. This indicates the variation in the angular density over the hypersphere.}
	\label{fig:MVN_characteristics}
\end{figure}

To give an indication of the variation of the angular density over the covariate domain, we generate a set of 5890 pseudo-regularly spaced reference angles on $\mathbb{S}^4$ (see Section \ref{sec:partition}), then count the number of observations falling within $\pi/18$ radians of each reference angles, from a sample of size $10^6$. The resulting counts are proportional to an estimate of the angular density. The results are shown in the right hand plot of \autoref{fig:MVN_characteristics}. The angular density varies by approximately two orders of magnitude over the covariate domain, indicating that some regions contain very few observations relative to others.

%%%%%%%%%%%%%%%%%%%%%%%%%%%%%%%%

\section{Partitioning of the hypersphere} \label{sec:partition}
The method for partitioning the hypersphere, used in Example 1, is based on forming a Voronoi partition relative to a set of pseudo-regularly spaced reference directions vectors. The reference directions are defined using the method proposed in \citet{mackay2023diform}. The first step is to create a regular grid of points $\bm{u}_1=(u_1,...,u_d)\in [-1,1]^d$, with spacing $1/m$, $m\in\mathbb{N}_{>0}$, and $mu_i \in\{-m,...,m\}$ for $i=1,\dots,d$. We then keep only points that lie on the $L^1$ unit sphere, such that $\sum_{i=1}^d|u_i|=1$. Finally, these points are projected onto the $L^2$ unit sphere, by defining $\bm{u}_2 = \bm{u}_1/\|\bm{u}_1\|_2$. We denote the set of pseudo-regularly spaced reference directions as $\mathcal{U}$. 

\begin{figure}[h!]
	\centering
	\includegraphics[scale=0.6]{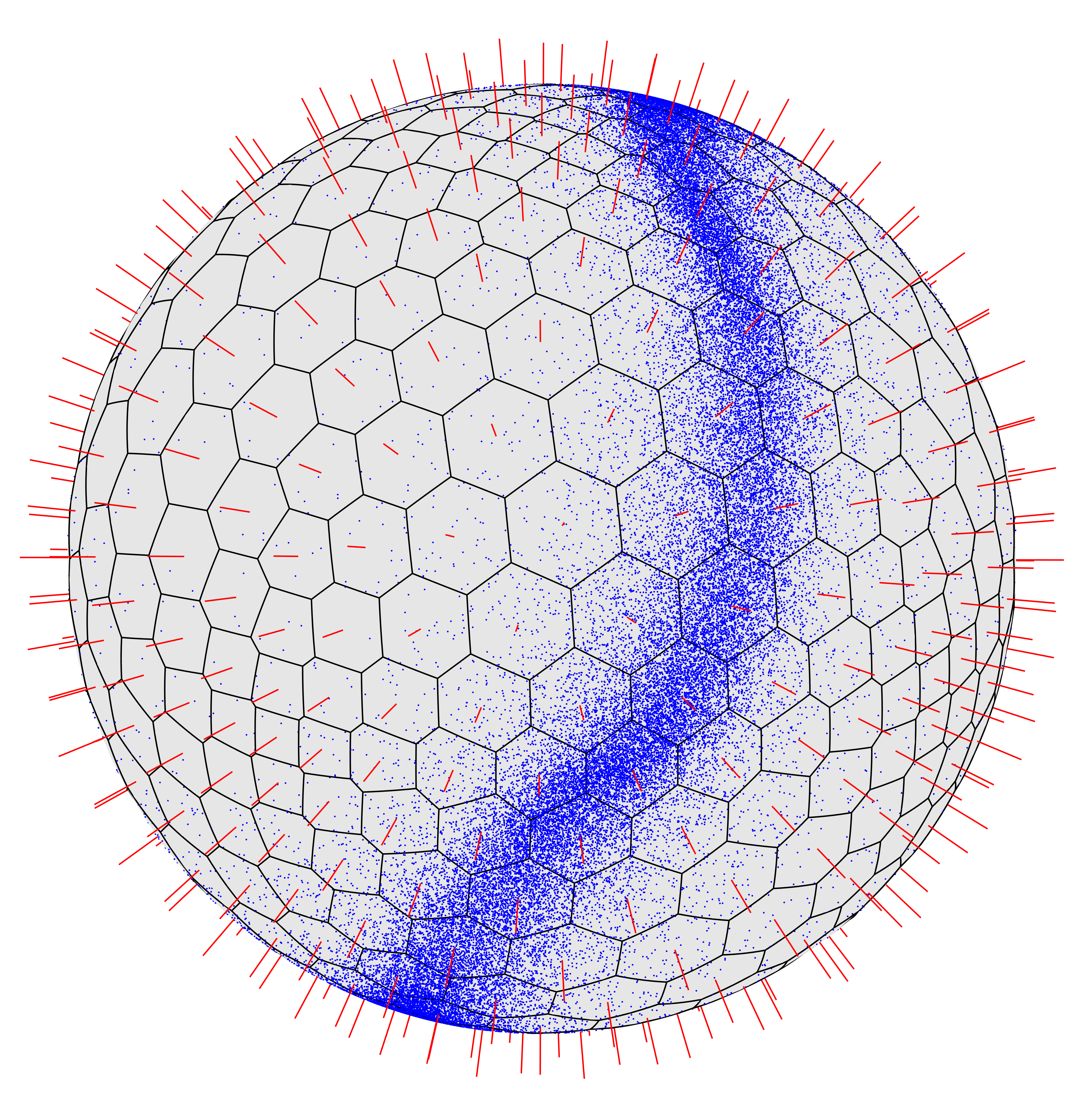}
	\includegraphics[scale=0.6]{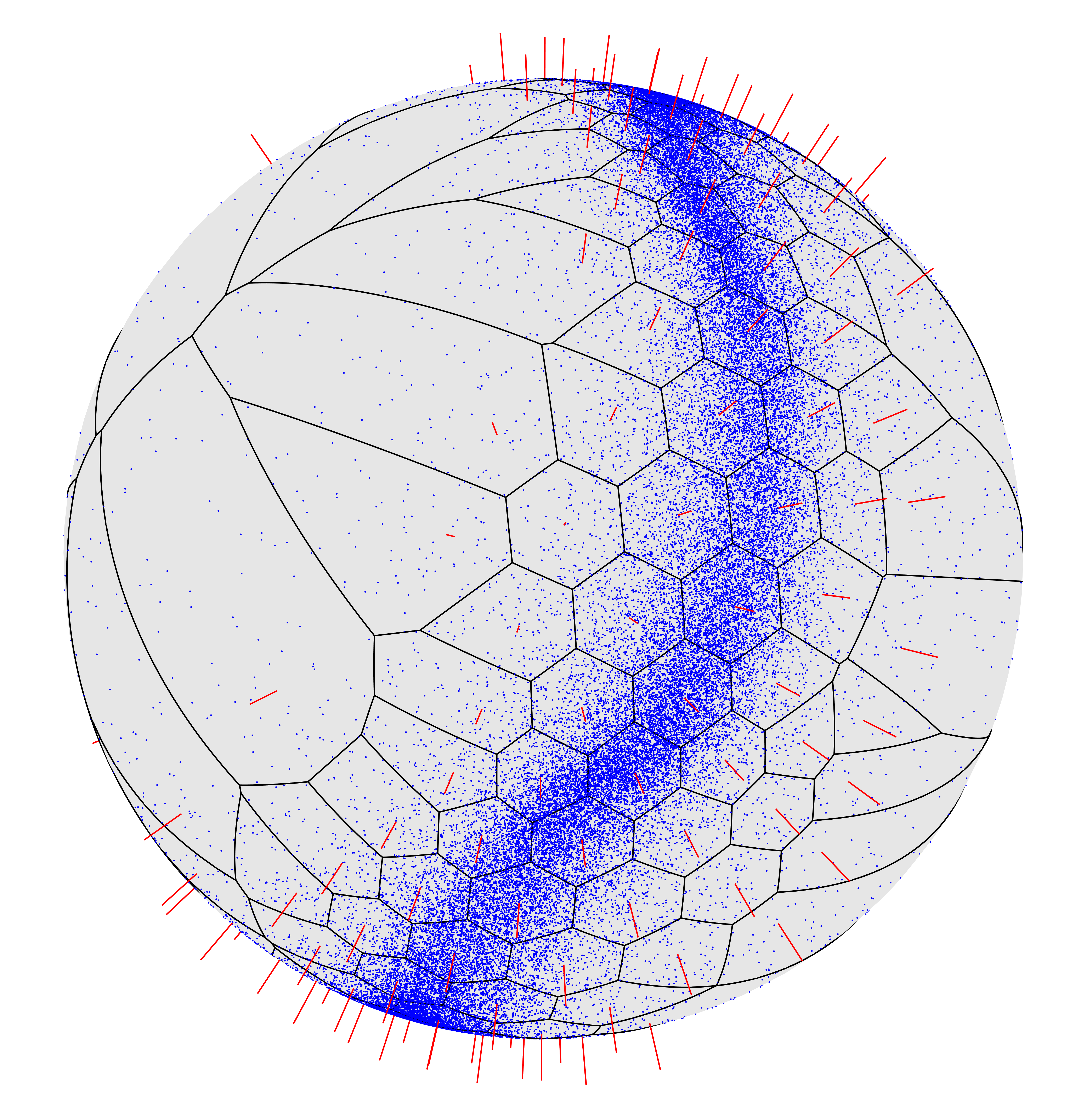}
	\caption{Illustration of initial (left) and final (right) Voronoi partitions of the sphere, relative to a set of pseudo-regularly spaced direction vectors (red lines). The angular component of a sample of $10^5$ points from a multivariate normal copula with standard Laplace margins is shown in blue. The partition on the right has been iteratively refined so that there are a minimum of $100$ observations in each bin.}
	\label{fig:sphere_voronoi}
	\centering
	\includegraphics[scale=0.7]{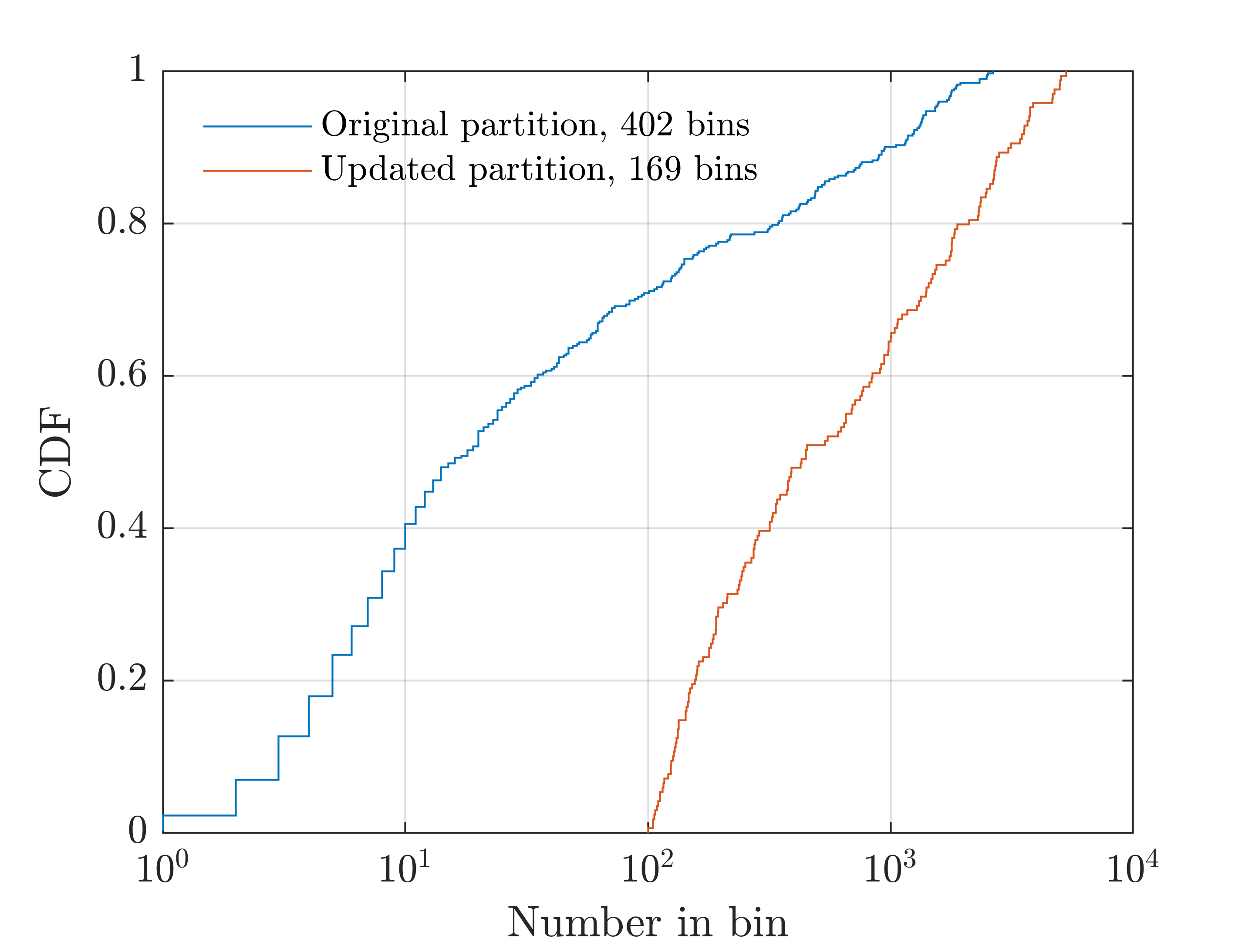}
	\caption{Empirical distributions of the number of observations in each bin for the Voronoi partitions shown in \autoref{fig:sphere_voronoi}.}
	\label{fig:sphere_voronoi_dist}
\end{figure}

For a given set of angles $\{\bm{w}_i\}_{1:n}$, a Voronoi partition, relative to $\mathcal{U}$, can be formed by assigning each angle $\bm{w}_i$ to the closest reference angle $\bm{u}_j\in\mathcal{U}$, i.e. the reference angle that minimises $\arccos(\bm{u}_j \cdot \bm{w}_i)$. To enforce a minimum number of observations, $n_0$, associated with each reference angle, we iteratively remove reference angles corresponding to bins with less than $n_0$, as follows:
\begin{enumerate}
	\item If any bins contain less than $n_0$ observations, remove direction vector corresponding to bin with fewest observations. Otherwise finish.
	\item Reassign each observation in the removed bin to the bin for the nearest remaining reference direction.
	\item Go back to step 1.
\end{enumerate}
An example of the initial and final partitions created using this method is shown in \autoref{fig:sphere_voronoi} for a three-dimensional example, with a sample of $10^5$ points from joint density \eqref{eq:norm_laplace_dens_SM} with $d-3$ and partial correlations $(S_{12}, S_{13}, S_{23}) = (0.8,0.2,-0.4)$. The initial set of reference directions is created using a grid with $m=10$, yielding 402 bins for the initial partition. After iterative refinement with a minimum bin size of $n_0=100$, the final partition has 169 bins. The empirical distributions of the number of observations per bin for the initial and final distributions are shown in \autoref{fig:sphere_voronoi_dist}.

This partitioning method is different from clustering algorithms on the hypersphere \citep[e.g.,][]{banerjee2005clustering, hornik2012spherical}. The method described above does not attempt to quantify any sort of `similarity' within bins. Instead, we just create a partitioning of the hypersphere into `bins' and refine this until we have a specified minimum number of observations in each bin. The advantage of this approach is that it is much faster to run -- the example above takes a few seconds to run on a laptop. As the diagnostics proposed do not assume any sort of stationarity across bins, similarity within bins does not matter in our application.

%%%%%%%%%%%%%%%%%%%%%%%%%%%%%%%%
\section{Supplementary figures}
This section includes figures that are supplementary to the main text.

\begin{figure}[t!]
	\centering
	\begin{minipage}{0.3\linewidth}
		\includegraphics[width=\linewidth]{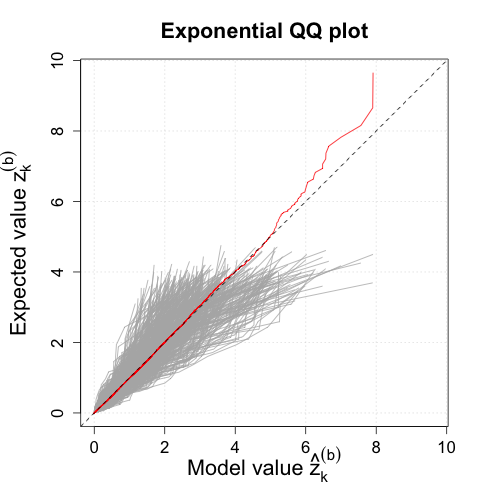}
	\end{minipage}
	\begin{minipage}{0.3\linewidth}
		\includegraphics[width=\linewidth]{plot_7497_2.png}
	\end{minipage}
	\begin{minipage}{0.3\linewidth}
		\includegraphics[width=\linewidth]{plot_7497_3.png}
	\end{minipage}
	\begin{minipage}{0.3\linewidth}
		\includegraphics[width=\linewidth]{plot_7497_4.png}
	\end{minipage}
	\begin{minipage}{0.3\linewidth}
		\includegraphics[width=\linewidth]{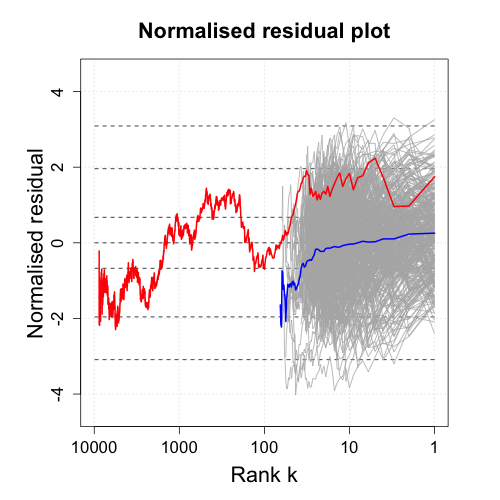}
	\end{minipage}
	\begin{minipage}{0.3\linewidth}
		\includegraphics[width=\linewidth]{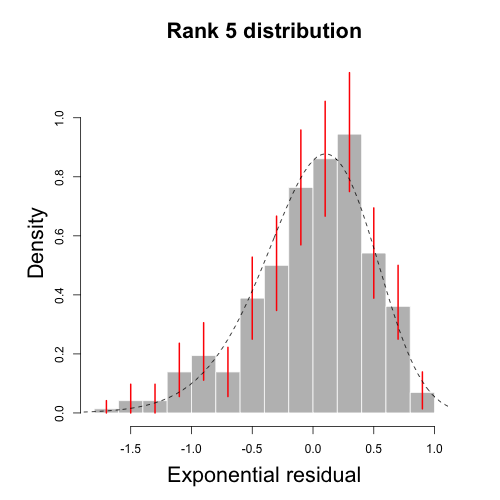}
	\end{minipage}
	\caption{Diagnostics for the fit of SPAR Model 1. Top left and centre: Exponential QQ plot and Standardised tail plot, with grey and red curves denoting regional and global estimates, respectively. Top centre: Dashed curves denote  0.001, 0.025, 0.25, 0.5, 0.75, 0.975, and 0.999 quantiles from the theoretical log-gamma distribution (as a function of rank $k$). Right column: histogram of empirical differences $D_k$ for $k=1$ (top) and $k=5$ (bottom), with the corresponding density function $f_{D_{k,\infty}}$ (dashed line). Bottom left:  Histogram of the regional ADR p-values used in the CvM uniformity test. Bottom centre: Normalised residual for all data (red) and regional samples (grey), with dashed lines indicating normal quantiles at exceedance probabilities 0.001, 0.025, 0.25, 0.75, 0.975, and 0.999. }
\end{figure}
\begin{figure}[t!]
	\centering
	\begin{minipage}{0.3\linewidth}
		\includegraphics[width=\linewidth]{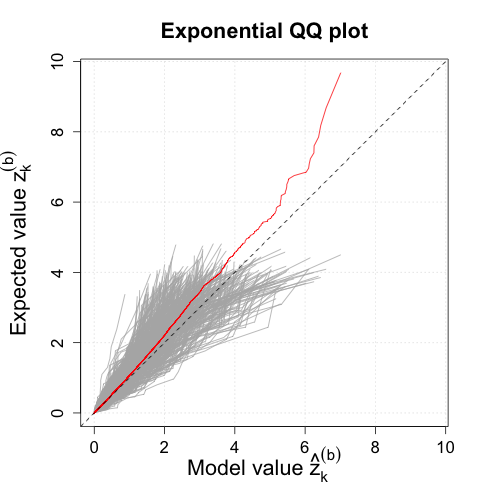}
	\end{minipage}
	\begin{minipage}{0.3\linewidth}
		\includegraphics[width=\linewidth]{plot_6272_2.png}
	\end{minipage}
	\begin{minipage}{0.3\linewidth}
		\includegraphics[width=\linewidth]{plot_6272_3.png}
	\end{minipage}
	\begin{minipage}{0.3\linewidth}
		\includegraphics[width=\linewidth]{plot_6272_4.png}
	\end{minipage}
	\begin{minipage}{0.3\linewidth}
		\includegraphics[width=\linewidth]{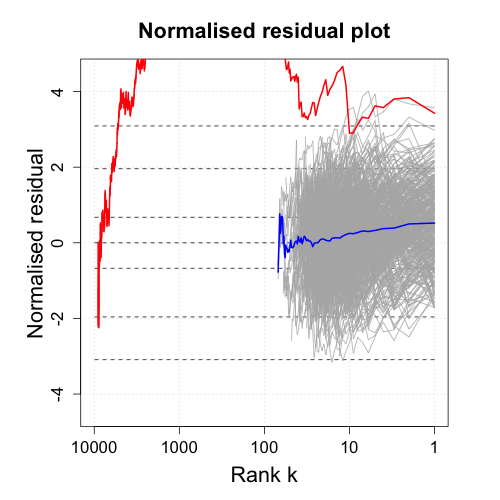}
	\end{minipage}
	\begin{minipage}{0.3\linewidth}
		\includegraphics[width=\linewidth]{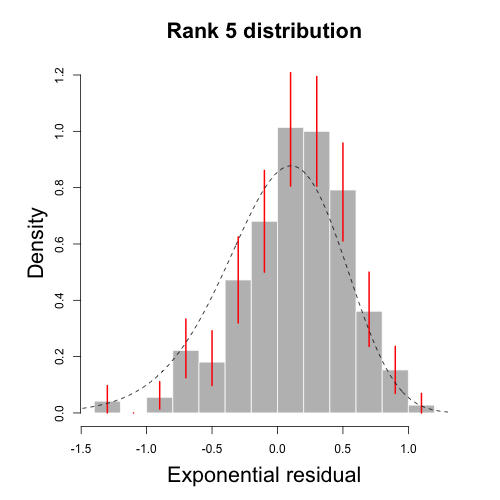}
	\end{minipage}
	\caption{Diagnostics for the fit of SPAR Model 2. Top left and centre: Exponential QQ plot and Standardised tail plot, with grey and red curves denoting regional and global estimates, respectively. Top centre: Dashed curves denote  0.001, 0.025, 0.25, 0.5, 0.75, 0.975, and 0.999 quantiles from the theoretical log-gamma distribution (as a function of rank $k$). Right column: histogram of empirical differences $D_k$ for $k=1$ (top) and $k=5$ (bottom), with the corresponding density function $f_{D_{k,\infty}}$ (dashed line). Bottom left:  Histogram of the regional ADR p-values used in the CvM uniformity test. Bottom centre: Normalised residual for all data (red) and regional samples (grey), with dashed lines indicating normal quantiles at exceedance probabilities 0.001, 0.025, 0.25, 0.75, 0.975, and 0.999. Centre column: Blue lines denote running empirical means of regional estimates (in grey). For all histograms, red lines denote 95\% bootstrapped error bars for each bin. }
\end{figure}

%%%%%%%%%%%%%%%%%%%%%%%%%%%%%%%%
\section{Integrals for EMAD asymptotic variance}
In the proof of Proposition 4.1, it was shown that, under the null hypothesis, the asymptotic variance of the EMAD statistic is given by $\Var(S) = \tfrac{2}{\pi}\int_0^1\!I(q)\dd q$, where
\begin{align*}
	I(q) &= \int_0^1
	\sqrt{\frac{(1-p)(1-q)}{pq}}\Bigl(\rho(p,q)\arcsin\rho(p,q)+\sqrt{1-\{\rho(p,q)\}^2}-1\Bigr)
	\dd p, \\
	\rho(p,q) &= \sqrt{\frac{\min(p,q) (1-\max(p,q))}{\max(p,q) (1-\min(p,q))}}.
\end{align*}
For notational simplicity, we hereafter drop the dependency of $(p,q)$ on the correlation function $\rho$. To simplify the calculation, we split the inner integral into three parts with $p<q$ and $p>q$, and define $I(q) = I_0(q) + I_1(q)-I_2(q)$, where
\begin{align*}
	I_0(q) &= \int_0^q
	\sqrt{\frac{(1-p)(1-q)}{pq}}\Bigl(\rho_0 \arcsin\rho_0+\sqrt{1-\rho_0^2}\Bigr)
	\dd p,\\
	I_1(q) &= \int_q^1
	\sqrt{\frac{(1-p)(1-q)}{pq}}\Bigl(\rho_1 \arcsin\rho_1+\sqrt{1-\rho_1^2}\Bigr)
	\dd p,\\
	I_2(q) &= \int_0^1
	\sqrt{\frac{(1-p)(1-q)}{pq}}
	\dd p = \frac{\pi}{2}\sqrt{\frac{1-q}{q}},
\end{align*}
and $\rho_0^2 = p(1-q)/(q(1-p))$, $\rho_1^2 = q(1-p)/(p(1-q))$. Multiplying through, we obtain
\begin{align*}
	I_0(q) &= \frac{1-q}{q} \int_0^q
	\arcsin \left(\sqrt{\frac{p(1-q)}{q(1-p)}}\right)
	\dd p + \frac{\sqrt{1-q}}{q}
	\int_0^q
	\sqrt{\frac{q-p}{p}}
	\dd p,\\
	I_1(q) &= \int_q^1
	\frac{1-p}{p} \arcsin \left(\sqrt{\frac{q(1-p)}{p(1-q)}} \right)
	\dd p
	+ \frac{1}{\sqrt{q}} \int_q^1
	\frac{\sqrt{(1-p)(p-q)}}{p} \dd p.
\end{align*}
Next, we evaluate the four individual integrals above.

%%%%%%%%%%%%%%%%%%%%%%%%%%%%
\subsection{Integral A}
To calculate the integral 
\begin{equation*}
	I_A = \int_0^q \arcsin\left(\sqrt{\frac{p(1-q)}{q(1-p)}}\right) \dd p, 
\end{equation*}
we substitute $u = \arcsin\left( \sqrt{\frac{p(1-q)}{q(1-p)}}\right)$. Solving for $p$ gives
\begin{equation*}
	p := p(u) = \frac{q \sin^2(u)}{(1-q) + q\sin^2(u)} = 1 - \frac{1-q}{1 - q \cos^2(u)}.
\end{equation*}
So, we can write the integral as $I_A = \int_0^{\pi/2} u \dd p(u)$. Integrating by parts, 
\begin{equation*}
	I_A = \left[ u \cdot p(u) \right]_0^{\pi/2} - \int_0^{\pi/2} p(u) \dd u = \frac{\pi}{2}(q-1) + (1 - q) J, 
\end{equation*}
where $J := \int_0^{\pi/2} (1 - q \cos^2(u))^{-1} \dd u$. To evaluate $J$, we divide the numerator and denominator by $\cos^2(u)$ and substitute $t = \tan(u)$, $\dd dt = \sec^2(u) \dd u$:
\begin{align*}
	J &= \int_0^{\pi/2} \frac{\sec^2(u)}{\sec^2(u) - q} \dd u = \int_0^{\pi/2} \frac{\sec^2(u)}{1 + \tan^2(u) - q} \dd u \\
	&= \int_0^\infty \frac{dt}{t^2 + (1-q)} = \left[ \frac{1}{\sqrt{1-q}} \arctan\left(\frac{t}{\sqrt{1-q}}\right) \right]_0^\infty = \frac{1}{\sqrt{1-q}} \cdot \frac{\pi}{2}
\end{align*}
Substituting this back gives
\begin{equation*}
	I_A = \frac{\pi}{2} \left( q - 1 + \sqrt{1-q} \right).
\end{equation*}

%%%%%%%%%%%%%%%%%%%%%%%%%%%%
\subsection{Integral B}
To calculate the integral 
\begin{equation*}
	I_B = \int_0^q \sqrt{\frac{q-p}{p}} dp,
\end{equation*} 
we substitute $p = q \sin^2(t)$ and $\dd p = 2q \sin(t) \cos(t) \dd t$. This gives
\begin{equation*}
	I_B = 2q \int_0^{\pi/2} \sqrt{\frac{1-\sin^2(t)}{\sin^2(t)}} \sin(t) \cos(t) \dd t = 2q \int_0^{\pi/2} \cos^2 (t) \dd t = \frac{\pi q}{2}.
\end{equation*} 

%%%%%%%%%%%%%%%%%%%%%%%%%%%%
\subsection{Integral C}
To calculate the integral 
\begin{equation*}
	I_C = \int_q^1 \frac{1-p}{p} \arcsin \sqrt{\frac{q(1-p)}{p(1-q)}}
	\dd p,
\end{equation*}
we start by substituting $x = \sqrt{(1-p)/(cp)}$, where $c = (1-q)/q$. The differential is
\begin{equation*}
	\dd p = \frac{-2cx}{(1 + cx^2)^2} \dd x
\end{equation*}
Substituting these elements back into the original integral gives
\begin{equation*}
	I_C = \int_0^1 \arcsin(x) \frac{2c^2 x^3}{(1 + cx^2)^2} \dd x.
\end{equation*}
Next, we apply integration by parts, with $u = \arcsin(x)$, $\dd u = \left(\sqrt{1-x^2}\right)^{-1} \dd x$ and $dv = 2c^2 x^3 (1 + cx^2)^{-2} \dd x$, to give
\begin{equation*}
	I_C = \left[u v\right]_0^1 - \int_0^1 v \dd u.
\end{equation*}
To find $v = \int \dd v$, let $y = 1 + cx^2$, so that $\dd y = 2cx \dd x$ and $x^2 = (y-1)/c$:
\begin{equation*}
	v = \int \frac{c \left(\frac{y-1}{c}\right)}{y^2} \dd y = \int \frac{y-1}{y^2} \dd y = \int \left(\frac{1}{y} - \frac{1}{y^2}\right) \dd y = \log|y| + \frac{1}{y} = \log(1 + cx^2) + \frac{1}{1 + cx^2}.
\end{equation*}
Therefore, the boundary term is
\begin{equation*}
	\left[u v\right]_0^1 = \left[\arcsin(x) \left(\log(1 + cx^2) + \frac{1}{1 + cx^2}\right)\right]_0^1 = \frac{\pi}{2} \left[q - \log(q)\right].
\end{equation*}
For the remaining integral we substitute $x = \sin\theta$ to give
\begin{align*}
	\int_0^1 v \dd u &= \int_0^1 \frac{\log(1 + cx^2) + \frac{1}{1 + cx^2}}{\sqrt{1-x^2}} \dd x \\
	&= \int_0^{\pi/2} \frac{\log(1 + c\sin^2\theta) + \frac{1}{1 + c\sin^2\theta}}{\sqrt{1 - \sin^2 \theta}} (\cos\theta \dd \theta) \\
	&= \int_0^{\pi/2} \frac{1}{1 + c\sin^2\theta} \dd \theta + \int_0^{\pi/2} \log(1 + c\sin^2\theta) \dd \theta .
\end{align*}
For the first integral, we divide both the numerator and the denominator by $\cos^2(x)$, to give
\begin{align*}
	J_1(c) = \int_0^{\pi/2} \frac{1}{1 + c\sin^2\theta} \dd \theta = \int_0^{\pi/2} \frac{\sec^2\theta}{\sec^2\theta + c \tan^2\theta} \dd \theta = \int_0^{\pi/2}\frac{\sec^2(x)}{1 + (c+1)\tan^2(x)} \dd \theta.
\end{align*}
Now, let $z = \tan\theta$, so that $\dd z = \sec^2\theta \dd \theta$ and 
\begin{align*}
	J_1(c) = \int_0^{\infty}\frac{1}{1 + (c+1)z^2} \dd z.
\end{align*}
Noting that $(\dd/\dd z) \arctan(az) = a/ (1+(az)^2)$, we have
\begin{align*}
	J_1(c) = \frac{\pi}{2 \sqrt{1+c}} = \frac{\pi}{2} \sqrt{q}.
\end{align*}
For the second integral, $J_2(c) = \int_0^{\pi/2} \log(1 + c\sin^2(x)) \dd x$, we use Feyman's trick, and define
\begin{align*}
	J_2'(c) &= \frac{\dd}{\dd c} \left[\int_0^{\pi/2} \log(1 + c\sin^2(x)) \dd x\right] = \int_0^{\pi/2} \frac{\partial}{\partial c} \log(1 + c\sin^2(x)) \dd x \\
	&= \int_0^{\pi/2} \frac{\sin^2(x)}{1 + c\sin^2(x) } \dd x = \frac{1}{c} \int_0^{\pi/2} \left( 1 - \frac{1}{1 + c\sin^2(x)} \right) \dd x\\
	&= \frac{1}{c} \left(\frac{\pi}{2} - J_1 \right) = \frac{\pi}{2c} \left( 1 - \frac{1}{\sqrt{c+1}} \right).
\end{align*}
Now we integrate $J_2'(c)$ to find $J_2(c)$:
\begin{equation*}
	J_2(c) = \frac{\pi}{2} \int \frac{1}{c} \left( 1 - \frac{1}{\sqrt{c+1}} \right) \dd c = \pi  \log\left( \sqrt{1+c} +1 \right) + K.
\end{equation*}
We use the initial condition at $c = 0$, of $J_2(0)=0$, to find $K=-\pi \log(2)$. Therefore, 
\begin{equation*}
	\int_0^1 v \dd u = J_1(c) + J_2(c) = \frac{\pi}{2} \sqrt{q} + \pi  \log\left( 1/\sqrt{q} +1 \right) -\pi \log(2).
\end{equation*}
Combining everything, we obtain
\begin{align*}
	I_C &= \frac{\pi}{2} \left[q - \log(q)\right] - \left[\frac{\pi}{2} \sqrt{q} + \pi  \log\left( 1/\sqrt{q} +1 \right) -\pi \log(2)\right]\\
	&=\frac{\pi}{2} \left(q-\sqrt{q}\right) - \pi \log\left(\frac{1+\sqrt{q}}{2}\right).
\end{align*}

%%%%%%%%%%%%%%%%%%%%%%%%%%%%
\subsection{Integral D}
The final integral is of the form given in Eq.~3.228.3 of \citet{gradshteyn2007}, and can be evaluated as
\begin{align*}
	I_D = \int_q^1 \frac{\sqrt{(1-p)(p-q)}}{p} \dd p = \frac{\pi}{2} (1-\sqrt{q})^2.
\end{align*}

%%%%%%%%%%%%%%%%%%%%%%%%%%%%
\subsection{Outer integral}
Substituting integrals $I_A$, $I_B$, $I_C$, and $I_D$ back in to $I_0(q)$ and $I_1(q)$ gives
\begin{align*}
	I_0(q) &= \frac{\pi}{2} \left[\frac{1-q}{q} \left(\sqrt{1-q}-(1-q)\right) + \sqrt{1-q}\right],\\
	I_1(q) &= \frac{\pi}{2} \left[ q-\sqrt{q}
	+ \frac{(1-\sqrt{q})^2}{\sqrt{q}} \right] - \pi \log\left(\frac{1+\sqrt{q}}{2}\right).
\end{align*}
and
\begin{align*}
	\frac{2}{\pi} I(q) &= \frac{\sqrt{1-q} - 1}{q} - \sqrt{\frac{1-q}{q}} + \frac{1}{\sqrt{q}} - 2 \log\left(\frac{1+\sqrt{q}}{2} \right).
\end{align*}
Returning to the original expression, we have $\Var(S) = \tfrac{2}{\pi}\int_0^1 I(q)\dd q = V_1 - V_2 + V_3 - 2 V_4$, where
\begin{align*}
	V_1 &= \int_0^1 \frac{\sqrt{1-q} - 1}{q} \dd q = 2\log(2)-2,\\
	V_2 &= \int_0^1 \sqrt{\frac{1-q}{q}} \dd q = \frac{\pi}{2},\\
	V_3 &= \int_0^1 \frac{1}{\sqrt{q}}\dd q = 2,\\
	V_4 &= \int_0^1 \log\left(\frac{1+\sqrt{q}}{2} \right)\dd q = \frac{1}{2} - \log(2).
\end{align*}
Combining the results gives the required solution, $\Var(S) = 4 \log(2)-\tfrac{\pi}{2} - 1$.

\end{document}